%% file: main.tex
\documentclass[12pt]{article}
\usepackage[utf8]{inputenc}

\usepackage{amsmath}
\usepackage{amssymb}
\usepackage{bm}
\usepackage{booktabs}
\usepackage[hidelinks]{hyperref}
\usepackage{graphicx}
\usepackage[hmargin=1.0in,vmargin={1.0in,1.1in},centering]{geometry}
\usepackage{setspace}
\usepackage{tikz}
\tikzset{every picture/.style={line width=0.75pt}} 
\usepackage{subcaption}
\usepackage[color=yellow]{todonotes}

\usepackage{fontspec}
\newfontface\DejaVuSans{DejaVuSans.ttf}
\def\digamma{\mbox{\DejaVuSans\char"03DD}}

\newcommand{\fr}[2]{\frac{#1}{#2}}

\newcommand{\pr}[1]{\left(#1\right)}
\newcommand{\br}[1]{\left[#1\right]}

\newtheorem{proposition}{Proposition}

\newenvironment{proof}[1][Proof]{\noindent\textbf{#1.} }{\ \rule{0.5em}{0.5em}}

\newcommand{\covid}{COVID-19}

\newcommand{\sirextend}{SAIRD}
\newcommand{\econsir}{Econ-SIR}

\newcommand{\smoothstep}{\Lambda}
\newcommand{\Fcomp}{F^{c}}
\newcommand{\partialexpec}{h}

\newcommand{\infratexc}{\beta}
\newcommand{\infrisk}{\widetilde{I}}  
\newcommand{\cost}{\psi}
\newcommand{\deathrate}{\kappa}

\newcommand{\ecoutput}{y}
\newcommand{\td}{(t)}

\newcommand{\zreslaw}{\bar{z}^{P}} 
\newcommand{\npibinary}{b}
\newcommand{\npiall}{B}

\newcommand{\params}{\theta}
\newcommand{\obsv}{x}
\newcommand{\obsvall}{X}

\RequirePackage[sort]{natbib}
\title{Economics and Epidemics: Evidence from an Estimated Spatial Econ-SIR Model\thanks{We would like to thank our discussant, Kyle Herkenhoff, as well as Adrien Auclert, Timo Boppart, Daniel S. Hamermesh, John Hassler, Per Krusell, Laura Pilossoph, Simon Mongey, and seminar participants at Bonn ECONtribute Virtual Macro Workshop, Federal Reserve Board, IIES, Safegraph Consortium Seminar, and the Virtual Macro Seminar (VMACS) for stimulating discussions.  We thank the IZA for generous research support through the Coronavirus Emergency Research Thrust.  We thank SafeGraph for providing access to their data and the SafeGraph COVID-19 Data Consortium for helpful input. We thank Ray Sandza and Homebase for providing access to their payroll data. Mitman gratefully acknowledges support from the Ragnar S\"oderbergs stiftelse, and the European Research Council grant No. 759482 under the European Union's Horizon 2020 research and innovation programme. Philipp Hochmuth and Fabian Sinn provided excellent research assistance. David Sch\"onholzer collaborated on an earlier iteration of this project. We presented earlier versions of this research project under the titles ``Economic Activity and COVID-19 Transmission: Evidence from an Estimated Economic-Epidemiological Model,'' and "Economic Activity and COVID-19 Transmission."}}
\author{Mark Bognanni\thanks{Board of Governors of the Federal Reserve System. Email: mark.j.bognanni@gmail.com. The views stated herein are those of the authors and are not necessarily those of the Board of Governors of the Federal Reserve System.}  \and Doug Hanley \thanks{University of Pittsburgh. Email: doughanley@pitt.edu} \and Daniel Kolliner\thanks{University of Maryland. Email: d.kolliner@gmail.com} \and Kurt Mitman\thanks{CEMFI, IIES, Stockholm University, CEPR \& IZA. Email: kurt.MIT.man@gmail.com}}
\date{\today}
\begin{document}

\maketitle
\begin{abstract}
Economic analysis of effective policies for managing epidemics requires an integrated economic and epidemiological approach. We develop and estimate a spatial, micro-founded model of the joint evolution of economic variables and the spread of an epidemic. We empirically discipline the model using new U.S. county-level data on health, mobility, employment outcomes, and non-pharmaceutical interventions (NPIs) at a daily frequency. Absent policy or medical interventions, the model predicts an initial period of exponential growth in new cases, followed by a protracted period of roughly constant case levels and reduced economic activity.  Nevertheless, if vaccine development proved impossible, and suppression cannot entirely eradicate the disease, a utilitarian policymaker cannot improve significantly over the laissez-faire equilibrium by using lockdowns. Conversely, if a vaccine will arrive within two years, NPIs can improve upon the laissez-faire outcome by dramatically decreasing the number of infectious agents and keeping infections low until vaccine arrival. Mitigation measures that reduce viral transmission (e.g., mask-wearing) both reduce the virus's spread and increase economic activity.
\end{abstract}

\newpage

\renewcommand{\baselinestretch}{1.15}
\onehalfspacing
\setcounter{page}{1}
\section{Introduction}


The rapidly developing COVID-19 pandemic presents perhaps the most daunting challenge to economic policymakers since the Great Depression. By engaging in economic activity, individuals subject themselves to infection risk, while infected agents impose negative externalities on others by spreading the virus.
In the absence of a vaccine or cure, policymakers can combat the virus's spread only through ``non-pharmaceutical interventions'' (NPIs) (e.g., recommending social distancing, stay-at-home orders).
While NPIs may effectively slow the rate of new infections, they do so primarily by reducing economic activity.
Alternatively, allowing the virus to propagate unchecked may also entail economic costs, as rational agents will mitigate their exposure to life-threatening risks. 
Formulating a sound policy response to the pandemic thus requires quantifying the tradeoffs between health and economic activity, as well as accurate predictions of \emph{both} economic and epidemiological variables in response to alternative policies.
Managing an epidemic thus requires an integrated assessment model, combining economics and epidemiology and rigorously disciplined by data, analogous to the integrated climate-economic models for evaluating policies to combat climate change \citep{nordhaus1993rolling}.\footnote{We thank Per Krusell for suggesting this analogy.}

\begin{figure}
    \centering
    \includegraphics[width=\textwidth]{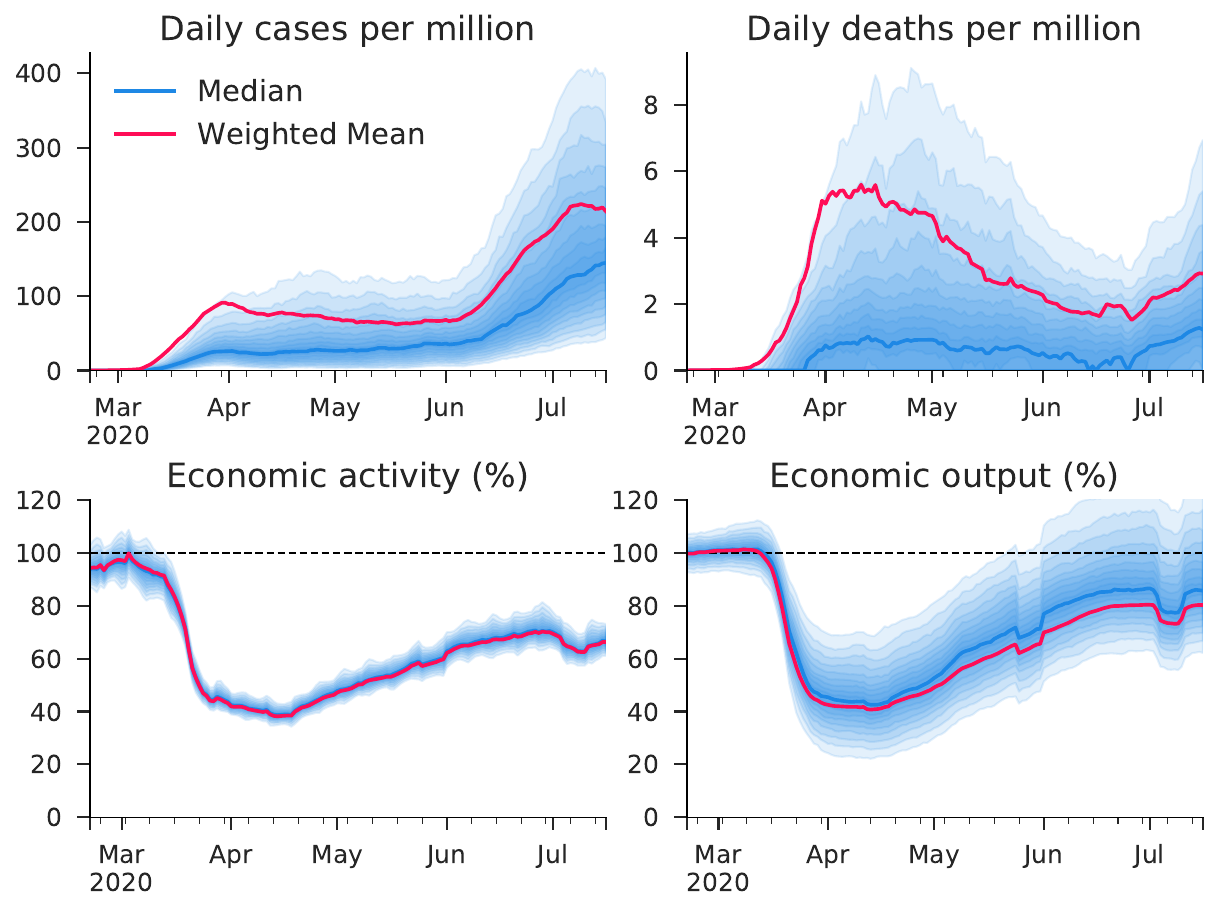}
    \caption{Clockwise from upper left: new reported \covid{ }cases, new reported \covid{} deaths, hours worked and economic activity. 
    In each panel, red lines are population-weighted means, the thick blue lines are the median county, and the light blue regions demarcate deciles of the  distribution across counties (all of these statistics computed pointwise).}
    \label{fig:data_cumul}
\end{figure}

This paper develops and estimates a parsimonious model of the joint determination of epidemiological and economic variables in an epidemic. 
We call our framework an ``economic SIR'' (Econ-SIR) model. 
The model features optimizing agents who endogenously respond to changes in both infection risk and NPIs. The model endogenously explains both the spread of the virus and the resulting impact on economy activity and output. The epidemiological block of the framework nests the standard SIR model \citep{Kermack1927} as the special case of a zero infection cost. 
When the infection cost is sufficiently high, the model endogenously gives rise to two phases of virus spread, with new infections initially growing exponentially before transitioning to a period of approximately linear growth. 
The virus's effective replication rate in the quasi-linear regime is just below one, leading to protracted dynamics of decreased economic activity and roughly constant spread of the virus. 

We are not the first to introduce optimizing behavior into models of epidemics. 
Seminal work by \cite{kremer1996integrating} introduced behavioral choice into models of HIV/AIDS. 
Over time, that strand of the economics literature progressed and became more quantitative (e.g., \citet{greenwood2019equilibrium,tox2019}). 
The focus of that literature, however, has mainly been on epidemiological and health outcomes---not economic ones \emph{per se}. In parallel, a strand of the epidemiological literature introduced the notion of ``behavioral change'' by agents in response to their environment and new information \citep{morin2013sir,fenichel2011adaptive}. See \cite{verelst2016behavioural} for an extensive review of this literature. The aforementioned papers focused on the decision of whether or not to engage in a particular activity with health consequences (e.g. sexual activity), while our framework focuses on decisions to participate in the marketplace and links those decisions to economic activity and output. 
Hence we would classify the earlier contributions as optimizing epidemiological models but not integrated-assessment models since they did not focus on explaining the trajectories of both the epidemic and the broader economy.

While other researchers have concurrently and independently formulated models similar to ours (e.g., \citet{farboodi2020internal,toxvaerd2020equilibrium,lones,brotherhood2020economic}), this paper is the first to formally bring such a model to the data to identify the extent to which the real world demands the Econ-SIR's divergence from pure epidemiological models.
We develop a novel estimation approach to fit the model to a panel of daily economic and epidemiological data for approximately 1000 U.S. counties, covering over 80 percent of the U.S. population.
Through revealed preference, we quantify both the disutility of contracting the disease and the economic consequences of mitigation policies---the two key trade-offs for policy analysis.
Our results highlight the importance of using an integrated assessment model for analyzing policy responses to epidemics, as the endogenous mitigation of private agents yields substantially different epidemiological dynamics than the standard SIR model.
While our quantitative results focus on the COVID-19 epidemic, the framework and estimation method have broader applicability for studying policy responses to future epidemics.


The estimated model closely matches the joint evolution of economic and epidemiological variables throughout the pandemic, both in the aggregate time series and lower-level units of analysis, such as U. S. states. 
In particular, the model can capture both the beginning of the pandemic and the spatial heterogeneity in the virus's ``summer resurgence.'' 
Looking beyond the data used directly for estimation, we find that our estimated model is consistent even with data on the seroprevalence of COVID-19 antibodies in numerous U.S. locations. 
At the estimated parameter values, individuals face a high infection cost, implying that the model's dynamics depart significantly from the standard SIR model. 
We find that a combination of endogenous mitigation by agents and the rollout of NPIs explains the decline in economic activity and output at the onset of the epidemic. 
These results are consistent with identified effects of ``stay-at-home'' orders on activity from empirical studies examining the rollout and reversal of NPIs at the county level \citep{goolsbee2020fear}.\footnote{See also earlier innovative empirical work on NPIs by \cite{adda2016economic}, evidence based on the Wuhan lockdown in China by \cite{fang2020human}, and recent cross-country evidence by \cite{mendiola2020}}. 
Furthermore, as in \cite{goolsbee2020fear}, we find that disentangling the endogenous mitigation of individuals from the effect of policy requires a county-level analysis (as opposed to state or national) as many counties enacted NPIs earlier than their corresponding states. 

Confident in our model's success in capturing the salient economic and epidemiological features of the epidemic, we proceed to use the framework as a laboratory to conduct counterfactual experiments. 
We characterize the implications of several policy scenarios for economic output and the virus's human toll. 
The first experiment examines the laissez-faire outcome of the model.
In particular, we simulate the model forward from August 1, 2020, assuming that all NPIs are reverted and no further NPIs are put into place. 
The model predicts that removing all restrictions on schools and restaurants would lead to an initial burst of economic activity and output, which would yield a jump in the effective replication rate and generate a surge in reported cases and a corresponding surge in daily deaths. 
Seeing the surge in infections, individuals then endogenously reduce economic activity, bringing the effective replication rate back below one by the early fall. 
Over the course of the first year of laissez-faire, the model predicts approximately 620,000 deaths from \covid{}, while output would remain more than 10\% below the pre-pandemic level.

Considering laissez-faire over the longer run, and still assuming the development of neither a vaccine nor a cure, the model reaches the ``slow burn'' regime sometime in 2021, where the behavior of the model is roughly linear on a day to day basis. 
In this phase, individuals reduce their economic activity until the effective reproductive factor ($\mathcal{R}_0$) is approximately 1. 
The pandemic then plays out over the next three to four years, with a roughly constant decline in daily deaths and cases and an increase in economic activity as the model approaches herd immunity. 
During this phase, \covid{} causes roughly 300,000 U.S. deaths per year, and economic output recovers roughly linearly from about 90 percent of pre-pandemic levels. 
These protracted dynamics imply a final toll of roughly 1.25 million \covid{} deaths over the next 4 years.

The PPF, while instructive about the tradeoffs a policymaker faces, is agnostic about how to value the tradeoff between output and lives. We proceed to compute the optimal adaptive lockdown policies (a Ramsey-like approach) of a utilitarian planner. The planner can only set the lockdown's strength as a function of the current number of infected individuals. Under the assumption that a cure or vaccine never arrives, we find that the optimal adaptive policy closely approximates the laissez-faire outcome. The intuition behind this result is that neither lockdown policies nor endogenous social distancing by agents can change the herd immunity threshold in the model (recall that when infections are close to zero the model dynamics converge to those of the basic SIR model). The only role for policy in this scenario is to minimize the ``overshoot'' of the herd immunity threshold to prevent ``unnecessary'' deaths and infections. The endogenous mitigation under laissez-faire generates very little overshoot at our estimated costs of infection, hence prescribing a small role for policy.\footnote{Thus far, we have assumed that our planner places the same value on life as do the economic agents. If the planner places a different value on life, e.g. from the notion of a statistical value of life, the optimal policy may be different.}

Given the stark predictions of the baseline model, we consider how the optimal adaptive policies change under two plausible extension. First, we rectify one of the unpleasant (and unrealistic) features of the standard SIR model---namely, we allow for the eradication of the virus.\footnote{In the standard SIR model with Poisson transition rates between states, the share of infected agents never reaches zero in finite time (absent a cure). As such, there will still be some measure of infected agents alive even after curtailing economic and social activity for an extended period (e.g., 100 years). After reversing the lock-down, the system will immediately revert to exponential growth. We allow for the share of infected individuals to fall to zero in finite time if the stock of infected individuals falls below a pre-specified threshold (e.g., if the percentage of infected implies fewer than one infected person in a county).} One need not take eradication literally, but we consider this thought experiment also to approximate an aggressive test, trace, and quarantine (TTQ) that can effectively control any community-based transmission of the virus. Second, we consider the possibility of the arrival of a vaccine.\footnote{We have also considered the arrival of a cure. The properties of the optimal adaptive policy under this alternative scenario are qualitatively the same as to the arrivals of a vaccine.}

With the possibility of eradicating the virus, a meaningful role for policy emerges. The endogenous social distancing by agents in the model is not strong enough to affect a quick reduction in infections to the eradication thresholds we consider.\footnote{Our eradication zone thresholds are motivated by the experiences of Taiwan and S. Korea at effectively either eliminating community transmission in the case of the former, or managing it with TTQ in the latter.} Now we find that a strict lockdown for two months followed by a graduate reopening over the course of a month can effectively bring infections down to the ``eradication zone'' whereby community transmission is eliminated or managed via TTQ. Note, given the spatial nature of the model, it is important that the lockdown be coordinated across locations, to prevent re-seeding via mobility. In this extension of the model, over the medium and long run, there is virtually no tradeoff between health and the economy. At one year (and longer horizons) one can maximize output (via temporary strict lockdowns followed by reopening) with minimal loss of life.

Finally, we consider the optimal adaptive policy when the arrival of a vaccine is near. We assume that a vaccine will arrive with certainty on August 1, 2021 and solve for the optimal adaptive policy until then. Under such a scenario, and given the current stock of infections, it is optimal to reduce activity to similar levels as under stay-at-home orders for two to four months. This significantly reduces the current stock of infections, after which the planner allows activity to increase until the effective replication rate is approximately 1. The planner achieves this by loosening NPIs to approximately the activity level estimated when restaurants are take-out only and schools are closed. NPIs are then rolled back once the vaccine arrives.

 Our paper relates to the mushrooming literature studying the epidemiological and economic response to the COVID-19 pandemic
(e.g. \cite{alvarez2020simple,farboodi2020internal,krueger2020macroeconomic,glover2020health,eichenbaum2020macroeconomics,guerrieri2020macroeconomic,berger2020seir,baker2020does,jones2020optimal,toxvaerd2020equilibrium,lones,brotherhood2020economic,mitman2020optimal,hagedorn2020corona,hur2020distributional,lubik2020go,kaplan2020pandemics}).\footnote{We apologize for any papers we have missed. Due to the rapidly changing landscape, this likely is not an exhaustive literature list.} 
Given the real-time and rapidly changing nature of those papers 
we think that instead of providing a discussion of each it is more instructive and expedient to summarize where we have innovated relative to that body of work and what our main contribution is.

Another strand of papers has focused on forecasting purely the epidemiological dimension of the pandemic's evolution in a panel of locations.  See \cite{schorfmoon} and \cite{thejesus}.

Relative to standard epidemiological models, we have added economic behavior --- what is sometimes called a behavioral SIR model.\footnote{We have opted for the term ``economic SIR'' model, as our agents are fully rational, optimizing agents.} 
We have estimated the model at a granular level, using data on individual movements based on GPS and data on economic outcomes---namely, employment---to discipline the economic-epidemiological connection.

\paragraph{}The rest of the paper is organized as follows. 
Section \ref{sec:econsir} presents the most parsimonious version of our \econsir{} model.
Section \ref{sec:twophases} focuses on the two distinct phases of virus spread to which the \econsir{} model gives rise.
Section \ref{sec:fullblownmodel} extends our \econsir{} model to facilitate a more rigorous quantification of the \covid{} pandemic.
Section \ref{sec:estimation} describes our estimation strategy and documents the estimated model's fit. 
Section \ref{sec:use_estimated_model} uses the estimated model to conduct and compare a variety of policy experiments. 
Section \ref{sec:conclusion} concludes. 



\section{A Minimal Economic-SIR Model}
\label{sec:econsir}

We begin by describing the simplest version of our economic SIR model to focus on the principal mechanism by which economic behavior and virus transmission interact. 
In Section \ref{sec:fullblownmodel}, we enrich the model with additional features particularly relevant to the \covid{} pandemic.

On the economic side of the model, susceptible agents can choose whether or not to participate in the economic marketplace, facing the potential risk of being exposed to infected individuals. 
They optimally trade off the gains from market participation, with the cost of potentially becoming infected with the virus.
The more widespread the infection is in the population, the higher the risk that market participants will become infected.
Further, as more infected agents participate in the marketplace, more individuals become infected.  
The model thus features bi-directional interaction between economic activity and viral transmission.

We set our models in continuous time $t\in[0,\infty]$, and adopt the following notational conventions: 
constants and parameters are denoted with lowercase Greek letters, e.g., \(\infratexc\);
stock variables, such as the measure of agents in a given  epidemiological state, are denoted with capital English letters, e.g., \(S\td\);
flow-variable counterparts to each stock variable use the same capital letter but add a ``\(\cdot\)'' over the character, e.g., \(\dot{S}\); 
other endogenous variables are denoted with lowercase English letters, e.g., \(n\td\).

\subsection{Epidemiological Structure}

\paragraph{Epidemiological States.} The model is populated with a time-invariant, unit measure of individuals.
An individual can be in one of three epidemiological states:  
Susceptible (\(S\)), Infected (\(I\)), or Recovered (\(R\)). 
Susceptible agents can contract the virus, but have not yet done so.
Infected agents have the virus and can spread it.
Recovered individuals have previously had the virus, can no longer spread it, and possess immunity from reinfection.
At a given time \(t\), $S\td,$ $I\td,$ and $R\td$ denote the measures of individuals in each epidemiological state.

\paragraph{Epidemiological state transitions.} Given initial values \(S_0\), \(I_0\), and \(R_0\), the measures of agents in each state evolve according to the following system of equations,
\begin{align}
    \dot{S} &= - n\td S\td 
    \label{eq:delta_S_sir} \\
        \dot{I} &= n\td S\td
    - \delta I\td 
    \label{eq:delta_I_sir} \\
        \dot{R} &= \delta I\td,
    \label{eq:delta_R_sir}
\end{align}  
in which Susceptibles contract the virus at rate \(n\td\) and Infecteds recover from the virus at rate $\delta$.
The number of new infections is given by 
\(N\td = n\td S\td\).
If we declared \(n\td\) equal to \(\infratexc I\td\) and stopped here,
then we would have the most standard SIR model.
We next describe the economic environment that endogenizes the rate of infection, \(n\td\).




\subsection{Economics and the Virus}
\label{eq:econvirus}

We now introduce the novel economic dimension of the model.
Agents face economic needs, which we model as a random variable with realizations denoted 
\(z \in (0, \infty)\).
Satisfying the economic need and receiving benefit equal to \(z\) requires making an excursion, while forgoing the excursion yields a benefit of only 0.
For a Susceptible, an excursion exposes the agent to infection risk.
For an Infected, an excursion exposes Susceptibles to increased infection risk.  
All else equal, we assume that agents prefer not to contract the virus and hence incur a cost \(\cost\) from infection.\footnote{As we discuss in \ref{sec:fullblownmodel}, this cost captures all disutility associated with infection.}
To capture the fact that not all economic needs are equally important, we model \(z\) as drawn from a distribution \(F(z)\).

Pausing at this level of generality, one can already see that the virus's economic implications will differ according to an agent's infection status.
\(I\) and \(R\) agents face no infection risk, so they may undertake excursions without fear of infection to satisfy any \(z\) (though we will assume that infected agents partially self-quarantine).
Susceptibles, however, face real infection risk and thus a tradeoff between the benefits of acquiring \(z\) with the potential cost of infection.
This tradeoff makes Susceptibles the model's key economic decision-makers.

\paragraph{To \texorpdfstring{$\mathbf{z}$}{z}, or Not To \texorpdfstring{$\mathbf{z}$}{z}?}  Conditional on taking an excursion, we model the rate of infection for Susceptibles as proportional to the economically active infected population share: 
\begin{align}
    \infrisk\td \equiv 
    \infratexc I\td a_I
    \label{eq:infrisk_def}
\end{align}
where 
\(\infratexc\) is a time-invariant rate of virus transmission per unit of Infected
and 
\(a_I\) is the rate at which infected agents take excursions.  
We refer to \(\infrisk\td\) as the ``infection risk'' (to an excursing Susceptible).

In light of their disutility from infection, Susceptibles optimally choose to satisfy a given economic need only if the benefit exceeds the expected cost of taking an excursion,
\begin{align}
    z\td 
    >
    \!\!\!
    \underbrace{
        \infrisk\td
        \cost \vphantom{\big(}
    }_{
        \substack{
            \text{expected cost}\\
            \text{of excursion}
        }
    }
    \!\!\!
    \equiv \bar{z}\td.
    \label{eq:to_z_or_not_to_z_sir}
\end{align}
\(\bar{z}\td\) represents the reservation value of
\(z\td\) for Susceptibles to excurse, which varies endogenously over time through its dependence on \(I\td\).

\paragraph{Economic Activity's Endogenous Response to Infection Risk.} Let \(a_{S}\td\) denote the rate of economic activity by a Susceptible--the rate at which the agent undertakes excursion.   
The optimizing decision rule in \eqref{eq:to_z_or_not_to_z_sir} implies that
\(a_{S}\td\) is given by the rate at which Susceptibles face economic needs \(z\td\) greater than \(\bar{z}\td\).
Denoting the complementary cumulative distribution function of \(z\) as \(\Fcomp(z) = 1 - F(z)\), the activity rate of Susceptibles is given by
\begin{align}
    a_{S}\td 
    = 
    1 - F\big( \bar{z}\td \big) 
    = 
    \Fcomp \big( \bar{z}\td \big)
    = 
    \Fcomp \big( \infrisk\td \cost \big).
    \label{eq:activeshare_s_sir}
\end{align}
With \(\bar{z}\td\) as given in  \eqref{eq:to_z_or_not_to_z_sir}, equation \eqref{eq:activeshare_s_sir} reveals the the key channel by which the infectiousness of the environment effects the economic activity of Susceptibles:
as the virus becomes more widespread (\(I\td\) increases), infection risk rises 
(\(
    \infrisk\td = 
    \infratexc I\td a_I
\) increases), 
which causes Susceptibles to endogenously curtail their economic activity  
(as long as \(\cost > 0\), \(\bar{z}\td\) increases with \(\infrisk\td\) and \(\Fcomp(\bar{z}\td)\) falls).
Noting that in the early stages of an epidemic the vast majority of agents belong to \(S\td\), this is the primary channel by which the virus affects aggregate activity.
Thus, the model contains a channel by which the virus can suppress economic activity even in the absence of outside interventions.  

\paragraph{Virus Transmission Rate's Endogenous Response to Economic Activity.}
The rate of infection for Susceptibles, $n\td$, is given by the product of their infection risk per excursion with the rate at which they take excursions.
Using \eqref{eq:infrisk_def} and \eqref{eq:activeshare_s_sir}, the infection rate is thus
\begin{align}
    n\td 
    = 
    \infrisk\td a_{S}\td 
    = 
    \infrisk\td 
    \Fcomp \big( \infrisk\td \cost \big)
    \label{eq:endog_trans_sir}
\end{align}
The term \(\Fcomp( \infrisk\td \cost)\), derived from the optimizing economic behavior of Susceptibles, introduces a nonlinear relationship between \(I\td\) and the new infection rate.
Using \eqref{eq:endog_trans_sir}, the flow rates \(\dot{S}\) and \(\dot{I}\) in \econsir{} (see equations \eqref{eq:delta_S_sir} and \eqref{eq:delta_I_sir}) 
become 
\begin{align}
    \dot{S} 
        &= 
        - \infrisk\td 
        \Fcomp \big( \infrisk\td \cost \big)
        S\td
    ~~,
    \qquad
    \dot{I} 
    = 
    \infrisk\td 
    \Fcomp \big( \infrisk\td \cost \big)
    S\td 
    - 
    \delta I\td
    \label{eq:flow_SI_sir}
\end{align}
The term 
\(
    N\td 
    = n\td S\td 
    = 
    \infrisk\td 
    \Fcomp \big( \infrisk\td \cost \big)
    S\td
\) 
is the number of new infections.  
Equations \eqref{eq:endog_trans_sir} and \eqref{eq:flow_SI_sir} characterize the key channel by which the
economic activity of Susceptibles affects the infectiousness of the environment:
as Susceptibles reduce their economic activity (\(a_S\td\) falls),
new infections decrease (\(n\td\) falls), 
which causes the number of infectious agents to fall
(\(\dot{I}\) falls so \(I\td\)).

\paragraph{Discussion.} From equations \eqref{eq:activeshare_s_sir} and \eqref{eq:flow_SI_sir}, one can see that not only does virus prevalence have implications for economic activity, but economic activity has implications for virus prevalence.
Lower economic activity by Susceptibles reduces the infection rate, which in turn reduces the number of new infections and the size of the Infected pool.  
When \(I\td\) increases, the endogenous reduction of activity by Susceptibles will tend to mitigate the rate of new infections.



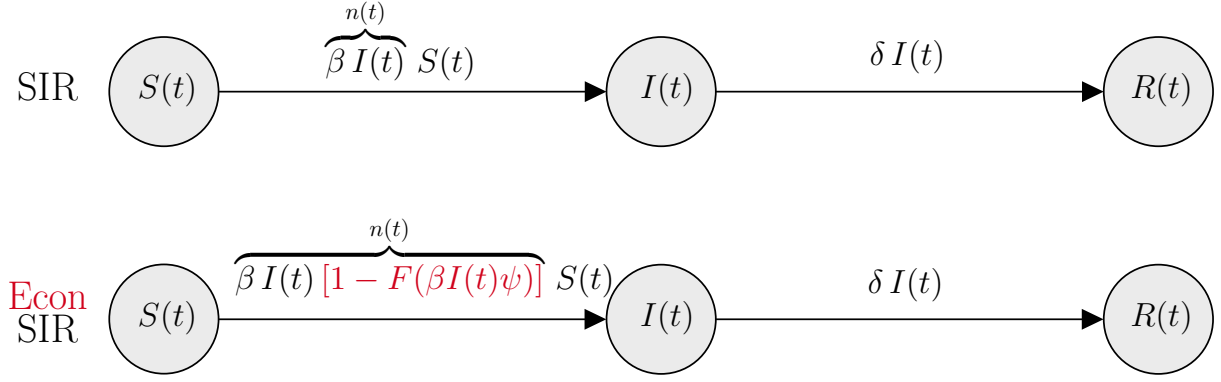
\begin{figure}[t!]
    \centering
    \resizebox{1.0\textwidth}{!}{
    \input{diagrams/sir_vs_econsir_big_tikzcode}
    }
    \caption{Structure of SIR and Economic SIR Models}
    \label{fig:sir_figure}
\end{figure}

\subsection{Policy: Nonpharmaceutical Interventions in \econsir{}} 
\label{sec:sirpolicy}

The primary policy lever exercised in response to \covid{} has been the imposition of NPIs, such as lockdown orders.
Particularly in the U.S., such orders have typically restricted the economic transactions that may take place (e.g. grocery shopping may be allowed out of necessity while tattoo parlors are required to close). 
We model such NPIs as a minimum \(\zreslaw\td\) that all agents internalize when deciding whether or not to undertake an excursion.
Hence the \emph{effective} reservation value of \(z\) for agents of type \(\tau\) (\(\in \{S, I, R\}\)) becomes
\begin{align}
    \bar{z}_{\tau}^{*}\td 
    &= 
    \max \big\{ 
    \bar{z}_{\tau}\td\,, \zreslaw\td 
    \big\}.
    \label{eq:zstar}
\end{align}
And hence activity by agents of type \(\tau\) becomes 
\begin{align}
    a_{\tau}\td 
    =
    \Fcomp \big( \bar{z}_\tau^*\td \big).
\end{align}

A few important aspects of how such a policy works are worth pointing out. 
First, for \(S\) agents, the endogenous \(\bar{z}_S\td\) forms an effective lower bound for \(\bar{z}^*\): if a policymaker chooses a \(\zreslaw\td\) below agents' endogenous \(\bar{z}\td\), then the policy is inoperative and has no effect on agents of type \(\tau\).  
Policy options are thus notably asymmetric. 
Stated in real-world terms, policymakers can force people to stay home, but they cannot force people to go out.

Second, shutdowns are a blunt tool.
Sufficiently restrictive shutdowns will reduce the activity of \emph{all} agents.
While such measures can mitigate or suppress the virus in the model, they are clearly a brute force way of doing so in light of the fact that, in principle, new infections could be driven to zero by driving just \(a_I\) to zero. 

\subsection{Economic Activity and Output}

We will later treat economic activity and output as observables when estimating the model.
However, in the data we cannot separately observe the activity of agents by epidemiological state, so the  observables will correspond to the model's aggregate values.  
To fix ideas, it is helpful to first define these aggregate concepts within the basic \econsir{} model.

Aggregate activity, denoted \(a\td\), is the measure-weighted mean of activity rates by each group
\begin{align}
    a\td 
    &= 
    a_{S}\td S\td + 
    a_{I}\td I\td + 
    a_{R}\td R\td 
    \label{eq:aggact_sir_constructive}
\end{align}
In the absence of NPIs, \(a_R\td = 1\) and 
\(a_I\td = a_I = \Fcomp(\bar{z}_I)\) 
and one can write \(a\td\) as\footnote{The following expression comes from substituting for \(a_{S}\td\) and using the fact that
\(S\td + I\td + R\td = 1\).}
\begin{align}
    a\td 
    &= 
    1 -
    \!\!\!\!
    \underbrace{F\big(\bar{z}_S\td\big)}_{
        \substack{
            \text{rate of forgone} \\ 
            \text{excursions by \(S\)}
        }
    }
    \!\!\!\!
    S\td
    -
    \!\!\!\!\!\!\!
    \underbrace{F(\bar{z}_I)}_{
        \substack{
            \text{rate of forgone} \\ 
            \text{excursions by \(I\)}
        }
    } 
    \!\!\!\!\!\!\!
    I\td
    \label{eq:total_activity_rate_sir}
\end{align}
When \(I\td\) equals zero, the virus is behaviorally irrelevant and activity occurs at its full rate of 1. 

Output is the rate at which value is accumulated from activity. 
The model differentiates between activity and output because whenever agents restrain their economic activity, they optimally sacrifice their least valuable activities first (forgo the lowest values of \(z\), recall \eqref{eq:to_z_or_not_to_z_sir}).
The optimizing behavior leads to a potentially non-linear relationship between activity and output, with the exact relationship determined by the form of \(F\).

Letting \(f(z) = F'(z)\) and 
\(\partialexpec(\bar{z})\) be the partial expectation of \(z\), defined as
\begin{align}
    \partialexpec(\bar{z}) 
    &=
    \int_{\bar{z}}^{\infty}
        z f(z) d z
    = 
    \mathbb{E}[z \,|\, z > \bar{z}] \Fcomp(\bar{z})
\end{align}
the output rate is given by
\begin{align}
    \ecoutput\td 
    &=
    \partialexpec\big( \bar{z}_S\td \big)
    S\td
    + 
    \partialexpec\big( \bar{z}_I\td \big)
    I\td
    + 
    \partialexpec\big( \bar{z}_R\td \big)
    R\td
    \label{eq:output_sir}
\end{align}
Equation \eqref{eq:output_sir} makes clear the differing output rates that result from the differing behavior between agents who are concerned with infection and agents who are not.
In our estimated model, we use \(F(z)\) such that $\log(z)$ is normally distributed with mean $-\fr{\sigma^2}{2}$ and variance $\sigma^2$, in which case $\mathbb{E}[z]=1$. 
Given this assumption, the variance parameter $\sigma^2$ summarizes the relationship between output and activity, as shown in \autoref{fig:actout}.

\begin{figure}
    \centering
    \includegraphics[width=0.5\textwidth]{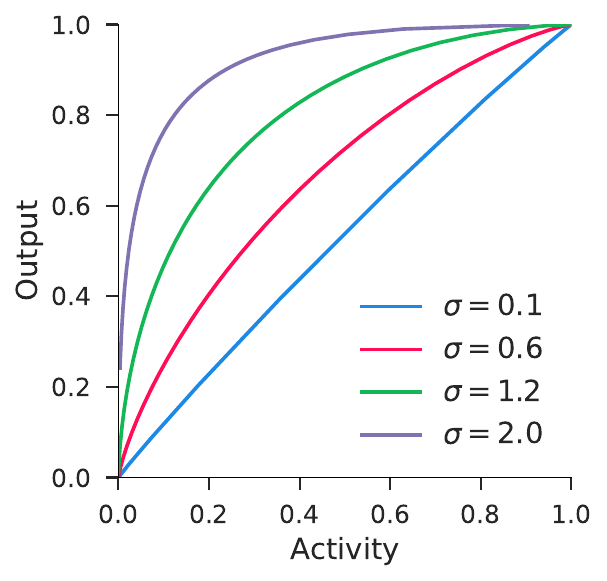}
    \caption{Relationship between economic output and activity given various values of $\sigma$, the standard deviation of $\log(z)$.}
    \label{fig:actout}
\end{figure}

\section{The Spread of the Virus in \econsir{}:  Exponential Growth and ``The Slow Burn''}
\label{sec:twophases}

In this section we describe the dynamics of the Econ-SIR model, in particular the nature of the two phases of virus spread to which the model endogenously gives rise: an initial phase of exponential growth following by a phase of approximately linear growth.
When $\cost=0$ Econ-SIR collapses to the standard SIR model whose dynamics are well-known from the epidemiological literature. 
As we describe below, increasing the cost of infection $\cost$ will lead to a more muted region of exponential growth and a more protracted phase of approximately linear growth.

\subsection{Phase 1: Exponential Growth of Infections}

The baseline rate at which an Infected agent in the model reproduces the infection in the broader population, which we call the baseline reproductive rate \(\mathcal{R}_0\), is given by
\begin{align}
    \mathcal{R}_0 = \fr{\beta}{\delta}.
\end{align}
Namely, \(\mathcal{R}_0\) is the ratio of the rate at which Infecteds transmit the virus per unit of Susceptibles to the rate at which the Infecteds recover and stop transmitting the virus altogether.
The baseline reproductive rate will obtain when \(S\td\) is near 1.

Taking into account the transmission mitigation from reduced activity by Susceptibles, the \emph{effective reproductive rate} in \econsir{} is
\begin{align}
    \mathcal{R}_e = a_S\td \mathcal{R}_0 \,.
\end{align}
Using \eqref{eq:delta_I_sir} and \eqref{eq:endog_trans_sir}, the growth rate of infections in the model is given by
\begin{align}
    \frac{\dot{I}\td}{I\td} 
    = 
     \frac{
        \infratexc I\td a_{S}\td 
        S\td
        - \delta I\td
    }{I\td}
    =
    \infratexc a_{S}\td S\td - \delta
    \label{eq:I_growthrate}
\end{align}
For very small values of $I\td$ or $\cost$, the activity rate of Susceptibles will be $a_{S}\td \approx 1$. Thus if initial infections $I_0$ are very small, to first-order dynamics of SIR and Econ-SIR are equivalent.
Further, if $S\td\approx 1$, Equation \eqref{eq:I_growthrate} then implies exponential growth of the infected population share according to rate
\begin{align}
    \frac{\dot{I}\td}{I\td} 
    = \beta - \delta.
\end{align}
which is positive whenever $\mathcal{R}_0 > 1$.

\subsection{Phase 2: The Slow Burn} 

As the virus spreads and \(I\td\) continues to increase, $a_{S}\td$ endogenously declines and plays a more important role in mitigating the the pace of new infections.
Here we develop intuition for the second phase of virus spread in which $I\td$ is approximately constant, i.e. \(\dot{I}\td \approx 0\) .

Recall that \(\dot{I}\td\) is  given by
\begin{align}
    \dot{I}\td
    =
    \infratexc I\td a_{S}\td 
        S\td
        - \delta I\td ,
\end{align}
it follows that \(\dot{I}\td = 0\) when
\begin{align}
    \infratexc a_{S}\td S\td
    = \delta .
    \label{eq:I_nochange}
\end{align}
Considering the case where the infected share has been low enough that \(R\td\) is still negligible so that 
\(
    S\td = 1 - I\td - R\td 
    \approx 1 - I\td
\),
substituting for \(a_{S}\td\), and dividing through by \(\infratexc\),
one can write the condition in \eqref{eq:I_nochange} as
\begin{align}
    \br{1-F(\beta I\td \cost)} 
    (1-I\td)
    = \frac{\delta}{\infratexc}
    = \fr{1}{\mathcal{R}_0}.
\end{align}
Recalling that 
\(\bar{z}\td = \infratexc I\td \cost\), one can express the previous relationship in terms of $\bar{z}\td$ as 
\begin{align}
    \br{1-F(\bar{z})} \pr{1-\fr{\bar{z}\td}{\beta \cost}} 
    = 
    \fr{1}{R_0}
    \label{eq:dougramblings1}
\end{align}
For large $\beta \cost$, equation \eqref{eq:dougramblings1} has the approximate solution for the endogenous object \(\bar{z}\td\) of 
\begin{align}
    \bar{z}\td \approx F^{-1}\pr{1-\fr{1}{\mathcal{R}_0}} .
    \label{eq:zbar_approx_slowburn}
\end{align}
Again using \(\bar{z}\td = \infratexc I\td \cost\), equation \eqref{eq:dougramblings1} implies that the pseudo-steady-state entails a constant fraction of agents in the infected state given by 
\begin{align}
    I\td = 
    \fr{1}{\beta \cost} \cdot F^{-1}\big(1 - \mathcal{R}_0^{-1} \big)
\end{align}
Since new infections \(N\td\) must equal outflows, \(\delta I\td\), the constant levels of new infections are given by 
\begin{align}
    N\td 
    = \fr{\delta}{\infratexc \cost} \cdot 
    F^{-1}\big(1 - \mathcal{R}_0^{-1} \big)
    =
    \fr{1}{\cost \mathcal{R}_0} \cdot 
    F^{-1}\big(1 - \mathcal{R}_0^{-1} \big) .
\end{align}
The expression for \(\bar{z}\td\) in \eqref{eq:zbar_approx_slowburn} immediately implies that the activity rate of Susceptibles satisfies
\(
    a_{S}\td 
    = 
    1 - F\left( F^{-1}( 1 - \mathcal{R}_0^{-1} ) \right)
    =
    \mathcal{R}_0^{-1}
\)
and hence 
\begin{align}
    \mathcal{R}_e = a_{S} \td \mathcal{R}_0 = 1
    \label{eq:slowburn_reproduction}
\end{align}

Thus, for sufficiently large $\beta\psi$, and while $R\td$ remains relatively small, we observe an effective reproduction rate of approximately 1. 
Absent some other intervention, this ``slow burn'' regime will last until $R\td$ becomes appreciably large, meaning ``herd immunity'' dynamics begin to have an impact.



\subsection{Comparing SIR and \econsir{} Epidemiological Dynamics} 

Here we highlight the key differences in the behavior of the SIR and \econsir{} models.
To understand the dynamics of both models it is instructive to understand the the curve in ($S,I$) space where $\dot{I}\td=0$, known as the nullcline of $I\td$.
Recall that \(\dot{I}\td\) is  given by
\begin{align}
    \dot{I}\td
    =
    \infratexc I\td a_{S}\td 
        S\td
        - \delta I\td ,
\end{align}
it follows that the nullcline is characterized by
\begin{align}
    \infratexc a_{S}\td S\td
    = \delta .
\end{align}
In the SIR model, since $a_S\td=1~\forall t$, the nullcline of $I\td$ is independent of $I\td$ and in given by a vertical line at $S=1/\mathcal{R}_0$, which is the herd immunity threshold (HIT).

\begin{proposition}
When $\cost>0$ and $F(0)=0$, the nullcline of $I\td$ is characterized by an upward sloping curve in $(S,I)$ space, denoted $I^N(S)$, that originates at $(1/\mathcal{R}_0,0)$. Furthermore, the slope of the nullcline is decreasing in $\psi$.
\end{proposition}
\begin{proof}
Substituting \eqref{eq:activeshare_s_sir} into \eqref{eq:I_nochange} and rearranging, the nullcline is characterized by
\begin{equation*}
     I^N  = \frac{1}{\infratexc\cost}F^{-1}\left(1-\frac{1}{\mathcal{R}_0S}\right)
\end{equation*}
Taking the limit as $S\rightarrow 1/\mathcal{R}_0\Rightarrow I^N\rightarrow 0$. Differentiating $I^N$ with respect to $S$ yields:
\begin{equation*}
    \frac{\partial I^N}{\partial S}  
    = 
    \frac{1}{\infratexc\psi}
    \frac{\mathcal{R}_0 S^2}{
        f\left(F^{-1}\left(1-\frac{1}{\mathcal{R}_0S}\right)\right)
    } 
    > 0
\end{equation*}
where the positivity follows since all parameters are positive and the pdf $f$ is non-negative. 
The relationship is also decreasing in $\psi$.
\end{proof}

Assuming initially that no agents have immunity to the virus $R_0=0$, if the initial seeding of the $I_0$ is such that $I_0<I^N(1-I_0)$, then the nullcline will describe the point that is the ``peak'' of the epidemic in terms of the stock of the number of infected agents and the effective replication rate is equal to unity, $\mathcal{R}_e =1$. The differences between SIR and \econsir{} are characterized graphically using the phase diagrams shown in \autoref{fig:phase}.

\begin{figure}
    \centering
    \includegraphics[width=\textwidth]{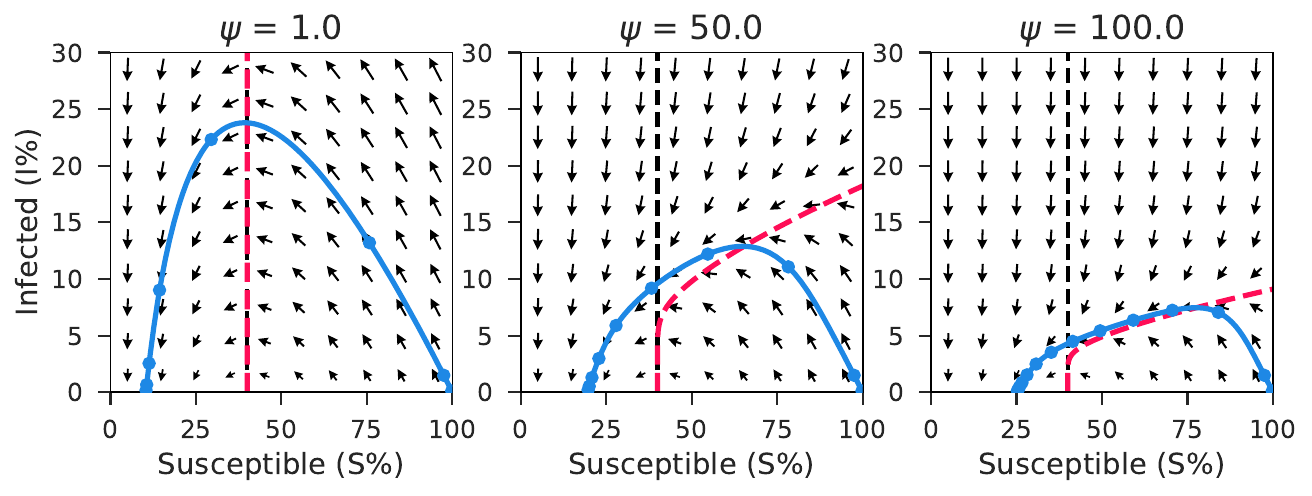}
    \caption{Phase diagrams for \econsir{} models under alternative infection costs $\psi$. 
    Blue lines represent sample time paths, dots along the time paths are spaced at one week intervals, black vertical dashed lines represent the herd immunity threshold, red dashed lines are nullclines of $I(t)$.}
    \label{fig:phase}
\end{figure}


In this example, the parameters are chosen so that \(\delta = \beta/2.5\), which yields \(R_0 = 2.5\).
In the SIR model, \(I\td\) then increases until   \(S\td\) falls below 0.4, which is the value at which \(\dot{I}\) turns negative.\footnote{Reduction of \(I\td\) mass in SIR does not occur until hitting ``herd immunity.''} Reduction of \(I\td\) mass in \econsir{} occurs well before reaching herd immunity, as agents try to avoid infection (as is visible from the bending of the nullcline from a verticle line). The strength of the mitigation increases with the cost of infection $\cost$. This results in both slower model dynamics --- in SIR herd immunity is reached in 4-5 weeks, whereas in the rightmost example it takes 8-9 weeks. Further, the stronger mitigation results in smaller overshoot of the herd immunity threshold. Summing up, \econsir{} generates more protracted dynamics and less overshoot relative to the SIR model. The extent of overshoot will become important later as this is the main externality that a policy maker is able to correct with lockdown policy.


\section{The Spatial Econ-\sirextend{} Model}
\label{sec:fullblownmodel}

This section extends our \econsir{} model along multiple empirically-relevant dimensions towards rigorously quantifying the policy tradeoffs posed by the \covid{} pandemic.
We leave some of the mathematical specifics to Appendix \ref{sec:saird_model_details}.

First, we add two features to the epidemiological environment that are particularly relevant for modeling \covid{}: asymptomatic virus carriers and mortality risk.
It has been well documented that many \covid{} carriers show few or no symptoms, sometimes only for a few days prior to becoming more seriously ill and sometimes agents recover with no increase in severity.  
As we now know, \covid{} poses mortality risk uncommonly high for a virus so easily transmitted.
Accordingly, when agents initially contract the virus in the extended model we assume that they move to an asymptomatic state (\(A\)).
Asymptomatic agents in the model unknowingly transmit the virus.
We model asymptomatic agents as typically unaware of their infectiousness, and hence they behave similarly to Susceptibles.\footnote{We model this as an activity rate \(a_{A}\td\) that is a trailing moving average of the activity rate of Susceptibles.} We should note that we keep $\cost$ to denote the cost of being infected (as opposed to the mortality cost) to capture the fact that even non-fatal cases can entail significant medical costs to individuals \citep{kniesner2020forgotten}.
Asymptomatic agents recover directly from the \(A\) state at rate \(\gamma\) or become symptomatic and infectious, moving to \(I\) at rate \(\lambda\).
Infectious agents die (move to the \(D\) state) at rate \(\kappa\).
The other epidemiological states have the same interpretation as before.
Individuals in the \(I\) state, being generally symptomatic, we model as having some awareness of their status and choosing to partially self-quarantine by undertaking actions no less valuable than some level $\bar{z}_{I}$.
Activity among \(I\) agents is then given by 
$a_I = \Fcomp(\bar{z}_{I})$.
We refer to the epidemiological state space of the model as \sirextend{}.

A second important feature of the real world is the presence of many separate locations, each with a distinct outbreak severity and virus transmission environment.
Furthermore, real world agents travel across locations, potentially spreading the virus from one region to another.
We add a spatial dimension to the model consisting of \(J\) distinct locations, each with its own measures of agents in the five epidemiological states (denoted with subscripts, as in \(S_j\td\)) and its own values of 
of \(\infratexc_j\) and \(\cost_j\).
All other parameters are the same in all locations.
Denote the rate of economic activity by infectious agents in location \(j\) as
\begin{align}
    e_j\td 
    = 
    a_{I,j}\td I_j\td + 
    a_{A,j}\td A_j\td .
\end{align}
If location \(j\) were an island, then the infection risk to an excursing Susceptible would be 
\(\infrisk_j\td = \infratexc_j e_j\td\). 
However, allowing for movement across locations, we model the infection risk per \(S_j\td\) excursion as a weighted average of the local force of infection and the population-weighted national average force of infection,
\begin{align}
    \infrisk_j\td 
    = 
    (1-\alpha) \infratexc_{j} e_{j}\td
    + 
    \alpha \sum_{k} w_k \infratexc_{k} e_{k}\td ,
    \label{eq:spatial_activity}
\end{align}
where \(\alpha \in [0,1]\) and \(w_k\) is a weight proportional to location \(j\)'s relative population. 
We also allow for an information friction in which agents perceive their infection risk to be determined according to the form of \eqref{eq:spatial_activity}, but with a value \(\alpha^*\) that may differ from the true \(\alpha\).
Denote an \(S_j\) agent's \emph{perceived} infection risk as  
\(\infrisk_j^{*}\td\).

The economic decision agents face, ``to \(z\), or not to \(z\),'' is similar to before.
Susceptibles can take an excursion with (perceived) expected benefit of 
\(
    z\td - 
    \infrisk_j^*\td
    \cost_j ,
\)
or not.
However, we enrich the environment to allow for a ``no-contact consumption'' option, by which agents can forgo an excursion and still obtain some of the value from a given \(z\), namely \(d \cdot z\td\), where \(d \in [0,1]\).
This model mechanism corresponds to the availability of delivery or takeout services for a variety of goods.
With the delivery option, Susceptibles optimally take an excursion only if
\(z_j\td > \infrisk_j\td \cost_j / (1-d) = \bar{z}_j\td\).

We allow for the economic environment to evolve over time with trends in both the value of the delivery option (\(d\td\)) and the safety of activity 
(\(\infratexc\td\)).
For \(\infratexc\td\), we represent the evolution via a common factor \(\tilde{\infratexc}\td\)
that proportionally scales the values of \(\infratexc_j\) everywhere.
We model the trends as deterministic processes
\begin{align}
    d\td
    &=
    \bar{d} \cdot \smoothstep( \lambda_{d} t )
    \qquad 
    \text{and}
    \qquad 
    \tilde{\infratexc}\td
    =
    1 + \big( \bar{\infratexc} - 1 \big) 
    \smoothstep ( \lambda_{\infratexc} t )
    \label{eq:adaptation_decay}
\end{align} 
where \(\bar{d}\) and \(\bar{\infratexc}\) represent the terminal values of the processes, 
\(\lambda_{d}, \lambda_{\infratexc} \geq 0\) govern the speed of the transitions to the terminal values, and $\smoothstep$ is the ``smoothstep'' sigmoid function, which is increasing and for which \(\smoothstep(0) = 0\).\footnote{The smoothstep function has a simple polynomial expression over the transition domain and is commonly used in computer graphics due to its ease of computability.} 
Accounting for the common trend component, the value of the infection rate parameter in \(j\) at time \(t\) is given by  
\(
    \infratexc_j\td 
    = 
    \infratexc_j \cdot 
    \tilde{\infratexc}_j\td .
\)


NPI policy operates in fundamentally the same way in the extended model as in the \econsir{} model (see Section \ref{sec:sirpolicy}).
The activity of all agents in the model is potentially constrained by NPIs. 
This includes recovered \(R\) agents since NPIs have not been enforced differentially based on a person's immunity status.
When taking the model to data and estimating the severity of realized NPIs, we will allow for three different NPI severities (\(\bar{z}^{P_1}\), \(\bar{z}^{P_2}\), and \(\bar{z}^{P_3}\)), corresponding to school closures, restaurant take-out only orders, and stay-at-home orders (the most broadly used NPIs).

We model the distribution of economic needs as  \(\mathrm{Lognormal}(\mu_z, \sigma_z)\).\footnote{With 
    \(z \sim \mathrm{Lognormal}(\mu_z, \sigma_z)\), our parameters \(\mu_z\) and \(\sigma_z\) represent the mean and standard deviation of the random variable \(\log(z)\).
}
In addition to having support \((0, \infty)\), the Lognormal distribution allows for known expressions for partial and conditional expectations, which facilitates the calculation of quantities such as \eqref{eq:output_sir}.
When taking the model to the data there will be two additional considerations for the handling of \(z\). 
Because we will normalize output in the data, we parameterize the distribution to have \(\mu_z\) fixed and equal to \(-\sigma_z^2/2\), which ensures \(E[z]=1\).
Lastly, we will estimate the model from daily data, which necessitates some special considerations for the systematic fluctuations of activity and hours with day of the week.
To accommodate the systematic day-of-week variation, we allow day-specific values of \(\mu_z\).

\section{Estimating Spatial Econ-\sirextend{}'s Parameters}
\label{sec:estimation}


This section describes our approach to estimating the Econ-\sirextend{} model's parameters and shows that the estimated model can generate the paths of economic and epidemiological variables that we observe over the course of the  \covid{} pandemic.
Going forward, let \(\params\) denote the set of model parameters.

\subsection{Data}

Our estimation uses a panel of daily county-level observations over the period from February 15, 2020 to July 15, 2020 on NPIs, new \covid{} cases, new \covid{} deaths, point-of-interest visits (economic activity), and hours worked.\footnote{Appendix 
    \ref{sec:dataappend} contains additional details on the data.
}
After removing counties with insufficient data, the estimation uses 921 U.S. counties for inference, accounting for about 80 percent of the U.S. population.\footnote{We 
    use only counties with populations of at least 50,000. 
    Because of missing or sparse observations in the Homebase hours data, we exclude a few additional counties that meet our minimum population threshold. 
}
The economic variables are normalized county-by-county.
For the next section, it will be useful to denote the vector of county \(i\)'s observables at time \(t\) as
\begin{align}
    \mathbf{\obsv}_{i,t}
    =
    \big[
        Activity_{i,t}\,,~
        Hours_{i,t}\,,~
        NewCases_{i,t}\,,~
        NewDeaths_{i,t}
    \big]'
\end{align}
and let 
\(\obsvall \equiv \{\mathbf{\obsv}_{i,t}\}_{i,t}\) 
denote the full panel of the observations.

The NPI data consists of three different types of NPIs: stay-at-home orders, restaurants and bars takeout-only orders, and school closures.
While we observe the enactments of NPI orders, an NPI's implications for behavior must be inferred.
Hence, an NPI in location \(i\) at time \(t\) of type \(k\) enters our data as a binary variable  \(\npibinary_{i,t,k}\), with a value of 0 for ``off'' and 1 for ``on.''
The model represents the relevance of a type-\(k\) NPI through the value of the parameter \(\bar{z}^{P_k}\), which we include in \(\params\) and estimate for each NPI type. 
For the next section, it will be useful to denote the full panel of NPIs as 
\(\npiall \equiv \{b_{i,t,k}\}_{i,t,k}\) 
and let \(\npiall_{0:t}\) be the panel of all NPIs up through date \(t\).

\subsection{Likelihood}

Given \(\params\) and the panel of enacted NPIs, the \sirextend{} model generates a deterministic path of values for epidemiological and economic variables in each location.\footnote{The 
    initial conditions for the epidemiological states are relevant as well, however we set them as a deterministic function of the elements of \(\params\) that we have already introduced.
    See Appendix \ref{sec:initialstate}.
}
In particular, the model generates paths for newly infected and deceased agents, as well as economic activity and output.
We take the deterministic structural model, along with measurement errors, to be the generating process for these four data series.
The measurement errors create a probabilistic relationship between the data and model-generated paths that forms the basis of the likelihood function.
We next describe our model for the data generating processes of the four variables of interest and then follow with a description of the resulting likelihood function.



\paragraph{Economic Data Generating Process.} We write \(a_{i,t}^{\params, \npiall_{0:t}}\) to denote the model's simulated value for economic activity in location \(i\) at period \(t\), where the superscripts emphasize the simulated path's dependence on the the particular values of the parameters and NPIs.
For both of the economic variables we assume mean-zero additive Gaussian measurement errors around the model's value. 
For economic activity, the data generating process is thus 
\begin{align}
    Activity_{i,t}
    &= 
    a_{i,t}^{\params, \npiall_{0:t}}
    + 
    \epsilon_{i,t,a} 
    ~, \quad 
    \epsilon_{i,t,a} \sim
    \mathrm{Normal}\big( 0, \sigma_{a}^2 / w_{i} \big),
    \label{eq:measurement_itj}
\end{align}
where \(w_{i}\) is a weight proportional to county \(i\)'s relative population among counties in the sample.
The model for \(Output_{i,t}\) is analogous to \eqref{eq:measurement_itj}, with the model object on the right-hand side instead being 
\(\ecoutput_{i,t}^{\params, \npiall_{0:t}}\).
We estimate the series specific variances \(\sigma_{a}^2\) and \(\sigma_{y}^2\).



\paragraph{Health Data Generating Process.} Modeling the relationship between the model objects and health variables introduces some additional challenges.
First, the data on new \covid{} cases and deaths are discrete counts and their scale varies considerably with county populations.
Meanwhile, the model's predictions for new cases and deaths take the form of continuous population shares.
Taking new cases as an example, 
one can straightforwardly map the model-implied flows into expected counts for a county with population $pop_i$ as
\(
    pop_i \times N_{i}\td\).\footnote{Note 
    that the observed value \(\obsv_{i,t,j}\) satisfies 
    \(\obsv_{i,t,j} = n_i p_{i,t,j}\) for a true population fraction \(p_{i,t,j}\).
The 
    relevant model flows are given in equations  \eqref{eq:delta_I} and \eqref{eq:delta_D} as 
    \(    
        N_{i}\td 
        =
        \lambda A_i\td
    \)
    and
    \(
        \dot{D}_{i}\td 
        =
        \kappa I_i\td.
    \)
}
The expected count can then discipline the \emph{mean} of a discrete probability distribution over the positive integers.



Second, the data are subject to systematic measurement error of at least two types.
The first type is reporting delays: the tendency of new infections or deaths occurring at time \(t\) to not show up in the data until \(t + d\).
We examined the empirical distribution of reporting delays in locations that separately report both day-of-death and ``new deaths,'' which lead us to incorporate a delay of \(d = 10\) days for both cases and deaths (see \autoref{fig:deathdelay} in the Appendix).

The second type of systematic measurement error in the health data is underreporting. 
We account for underreporting by assuming that asymptomatic cases are generally not reported (i.e. agents in \(A\td\) do not show up in new cases data) and that symptomatic cases (new entrants to \(I\td\)) are systematically underreported, particularly in the early phase of the pandemic.
We model the evolution of underreporting of new cases using a sequence of reporting factors \(r_{t} \in [0,1]\) that converge to 1 at an exponential rate.
In the limit, cases are correctly reported.\footnote{While 
    we build under-reporting of cases into the model and estimation, we do not do so for deaths.
    We estimate the model to best match the \emph{reported} data for \covid{} deaths (including the reporting delays).
    A number of journalistic outlets have documented evidence of ``excess deaths'' in multiple locations, which exceed the number of reported \covid{} deaths in those locations.
    To the extent that official death totals are meaningfully under-reported, the true mortality rate in our model (\(\deathrate\)) should be higher and readers of this would want to accordingly adjust upwards our mortality predictions.
    }
Accounting for the systematic delays and underreporting, we have the following implication that will center the DGP for new cases,
\begin{align}
    E[NewCases_{i,t}] 
    = 
    r_{t-d} \cdot pop_{i} \cdot 
    N_{i,t-d}^{\params, \npiall_{0:t-d}}.
    \label{eq:epi_model_data_center}
\end{align}


Lastly, having accounted for the systematic dimensions of the measurement error, we model the idiosyncratic variation as arising according to a  Conway-Maxwell-Poisson (CMP) distribution.\footnote{The CMP  
    distribution of \cite{conway1962} is also known as the COM-Poisson distribution.
    See \cite{sellers2012} and \cite{daly2016} for useful recent overviews of the distribution's features.
}
The \(\mathrm{CMP}(\lambda, \nu)\) distribution extends the Poisson distribution by allowing the center and dispersion to be specified separately.
The parameter \(\nu\) governs the distribution's dispersion, allowing us to model the difference in measurement accuracy between cases and deaths.
Thus, we model the new cases data as
\begin{align}
    NewCases_{i,t} 
    \sim 
    \mathrm{CMP}\big( 
        \lambda_{C}, \nu_{C} 
    \big)
\end{align}
where \(\nu_{C} = 0.1\) and, conditional on \(\nu_{C}\),  \(\lambda_{C}\) is chosen to fix the mean of the CMP distribution at the value prescribed by \eqref{eq:epi_model_data_center}.
We model the data on new deaths analogously, but with \(\nu_D = 100\).
The differing dispersion parameters reflect the view that the deaths data are considerably more reliable.

\paragraph{Location-specific parameters.} All parameters except \(\{\infratexc_{i}, \cost_{i}\}_{i}\) are common to all locations.
We shrink the location-specific values of \(\infratexc_{i}\) and \(\cost_{i}\) towards common values by modeling them as random effects with 
\(
    \infratexc_{i} \sim LogNormal(\bar{\infratexc},         
        \sigma^2_{\infratexc})
\) 
and 
\(
    \cost_{i} \sim LogNormal(\bar{\cost},         
        \sigma^2_{\cost}).
\)

\paragraph{Initial conditions.} We initiate the virus for county \(j\) at the date when a positive number of confirmed cases first occurs in the data.
At that date, we set \(I\td = NewCases/Pop\) and 
\(A\td = \digamma I\td\), where
$\digamma$ is a function of other parameters (see Appendix \ref{sec:initialstate}).\footnote{At 
    the initialization date, \(R\td = D\td = 0\) and 
    \(S\td = 1 - I\td - A\td\).
}

\paragraph{Likelihood.}
Accounting for the DGPs and random effects, the 
likelihood takes the form
\begin{align}
    L(\params | \obsvall, \npiall)
    \propto
    p(\obsvall | \params, \npiall)
    \prod_{i}
    p\big(
        \infratexc_i | \bar{\infratexc}, \sigma^2_{\infratexc} 
    \big)
    \,
    p\big(
        \cost_i | \bar{\cost}, \sigma^2_{\cost}
    \big).
    \label{eq:likfull}
\end{align}
We assume the measurement errors are idiosyncratic in the sense of being independent across time, location, and data series, thus 
\(p(\obsvall|\params, \npiall)\) decomposes as
\(
    p(\obsvall | \params, \npiall)
    =
    \prod_{i,t} 
    p(\mathbf{\obsv}_{i,t} | \params, \npiall_{0:t})
\)
where 
\begin{align}
    \begin{split}
    p(\mathbf{\obsv}_{i,t} | \params, \npiall_{0:t})
    =\, 
    &\phi \big(
        Activity_{i,t} | 
        a_{i,t}^{\params, \npiall_{0:t}}, 
        \sigma_{a}^2 / w_{i}
    \big)
    \cdot
    \phi \big(
        Output_{i,t} |
        \ecoutput_{i,t}^{\params, \npiall_{0:t}}, 
        \sigma_{\ecoutput}^2 / w_{i}
    \big)
    \\
    \cdot\,
    &g(NewCases_{i,t} | \lambda_{C}, \nu_{C})
    \cdot
    g(NewDeaths_{i,t} | \lambda_{D}, \nu_{D})
    \end{split}
\end{align}
for \(\phi(\cdot | \mu, \varsigma)\) the pdf of a Normal distribution with mean \(\mu\) and variance \(\varsigma\),
\(g(\cdot|\lambda,\nu)\) the probability mass function of a CMP distribution, 
and
\(\lambda_{C}\) and \(\lambda_{D}\) set as described above.
We estimate \(\params\) as the argmax of (the \(\log\) of) \(L\) in \eqref{eq:likfull}.

\paragraph{Discussion.} A few aspects of the estimator warrant further discussion. 
First, the simulated trajectories from the model might be interpreted as impulse responses to a ``\covid{} shock.''
The entirety of the paths reflect a single \covid{} infection shock that perturbs the measures of \(S\td\) and \(I\td\) in each county at the beginning of the pandemic and the policy paths in each county.
In other words, the paths have no mechanism for correcting mid-course deviations from the data.
In the language sometimes used when assessing DSGE models, one could say that our model fit results exclusively from \emph{internal} propagation mechanisms following the infection impulses in early March.

Second, for the parameters common across locations, the estimator prioritizes fitting the dynamics of higher population areas.
The economic variables are normalized county-by-county, but the weights \(w_i\) in \eqref{eq:measurement_itj} decrease the relative measurement error as county population increases.
A similar relationship is implicitly induced by the CMP distribution for the health variables. 
One can show that increasing the mean of the CMP while holding \(\nu\) fixed at a given value results in probability mass more concentrated around the distribution's center.
Thus the relative size of the idiosyncratic variation falls with higher expected counts, which, all else equal, occur in higher population areas.
The greater ``weight'' on matching the dynamics in large population counties helps to ensure good fit to aggregates.  




\subsection{Fit}

We show the estimated Econ-\sirextend{} model's success at fitting the data along multiple dimensions, beginning with the time-series of aggregates and period-by-period cross section.

\paragraph{Time series.}  The two panels of \autoref{fig:modelfit} compare the data and model predictions for the four variables of interest.  In light of this fact, we think the model's ability to generally match the paths of all 4 variables over a roughly 5 month period, including case resurgences in multiple states, based on only initial conditions is fairly remarkable.

\begin{figure}[t]
    \centering
    \begin{subfigure}{.49\textwidth}
      \centering
      \includegraphics[width=1.0\textwidth]{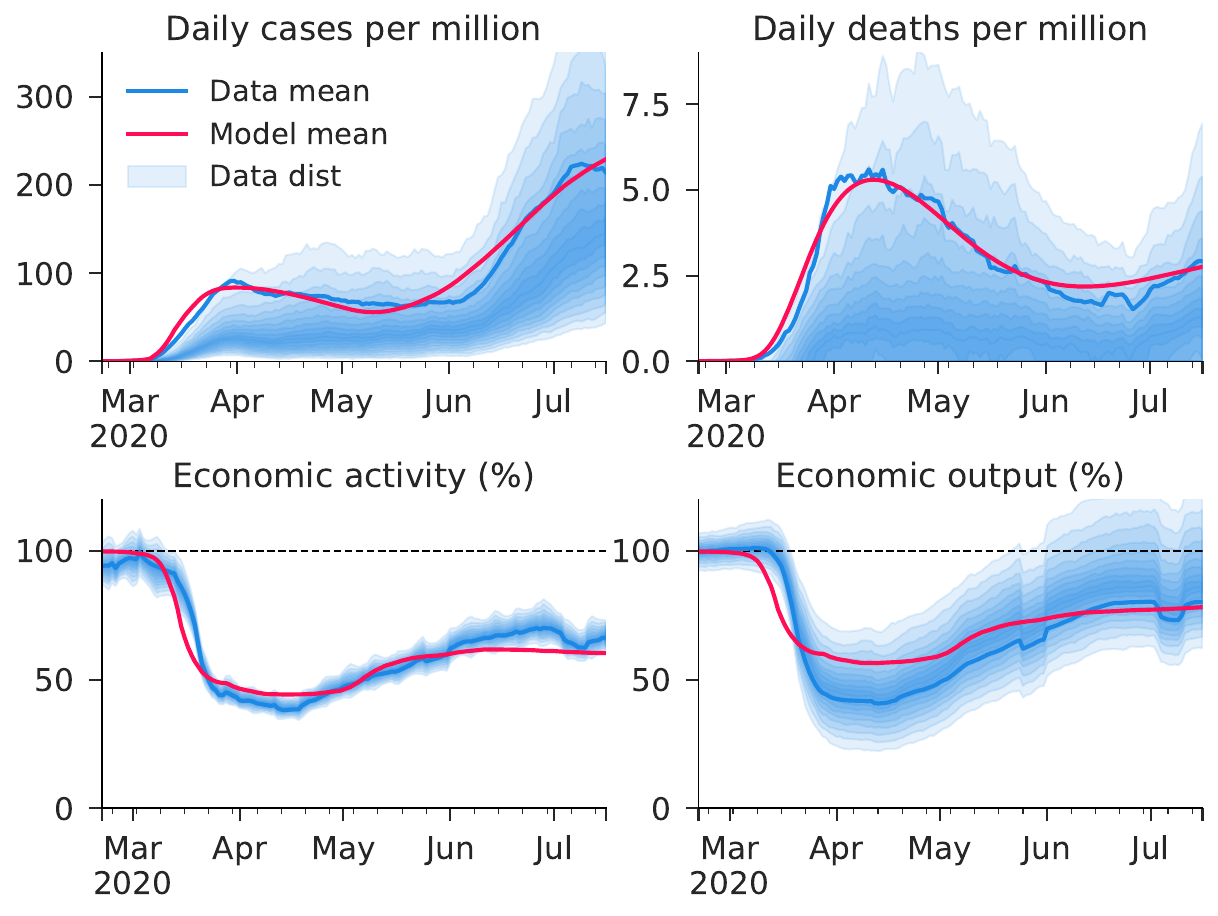}
      \caption{Cross-section from the data.}
      \label{fig:modelfit1}
    \end{subfigure}%
    ~~
    \begin{subfigure}{.49\textwidth}
      \centering
      \includegraphics[width=1.0\textwidth]{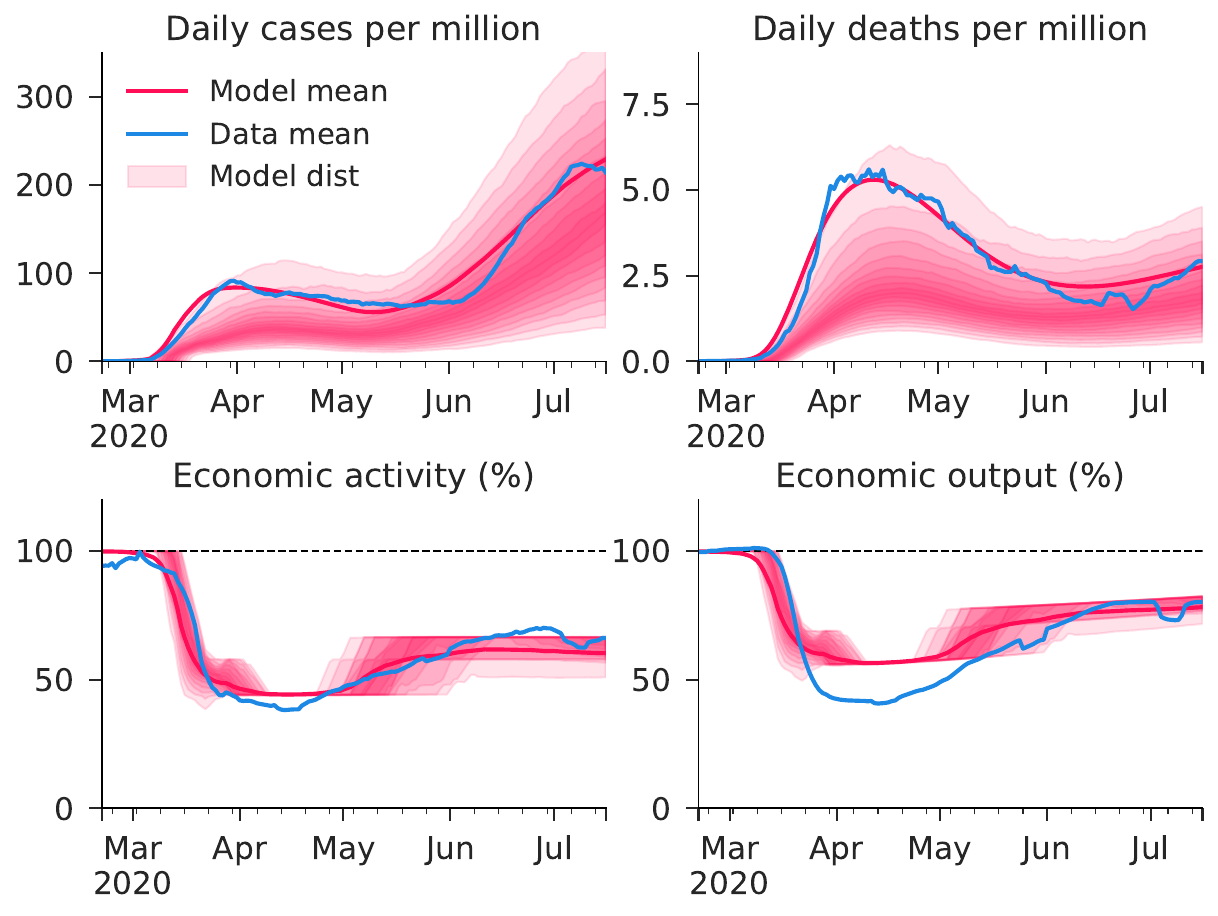}
      \caption{Cross-section from the model.}
      \label{fig:modelfit2}
    \end{subfigure}
    \caption{Data vs model: time series and cross sectional distributions.}
    \label{fig:modelfit}
\end{figure}

\paragraph{Cross sectional units.} While \autoref{fig:modelfit} demonstrates the model's fit to the aggregate time series, it is also important to show that the model accurately captures the time series evolution at lower levels of aggregation.\footnote{We would not want to get the average right (in the form of aggregates) while getting all individual units very wrong.}
\autoref{fig:modelfit_twostate} compares the model and data for two states of particular interest: Florida and California.
Appendix \ref{sec:statefit} contains analogous figures for each of the 44 states in our sample.\footnote{Our data does not include any counties in six states because of lack of coverage in the hours worked data.}

\begin{figure}[t]
    \captionsetup[subfigure]{labelformat=empty}
    \centering
    \begin{subfigure}{.49\textwidth}
      \centering
      \caption{Florida}
      \includegraphics[width=1.0\textwidth]{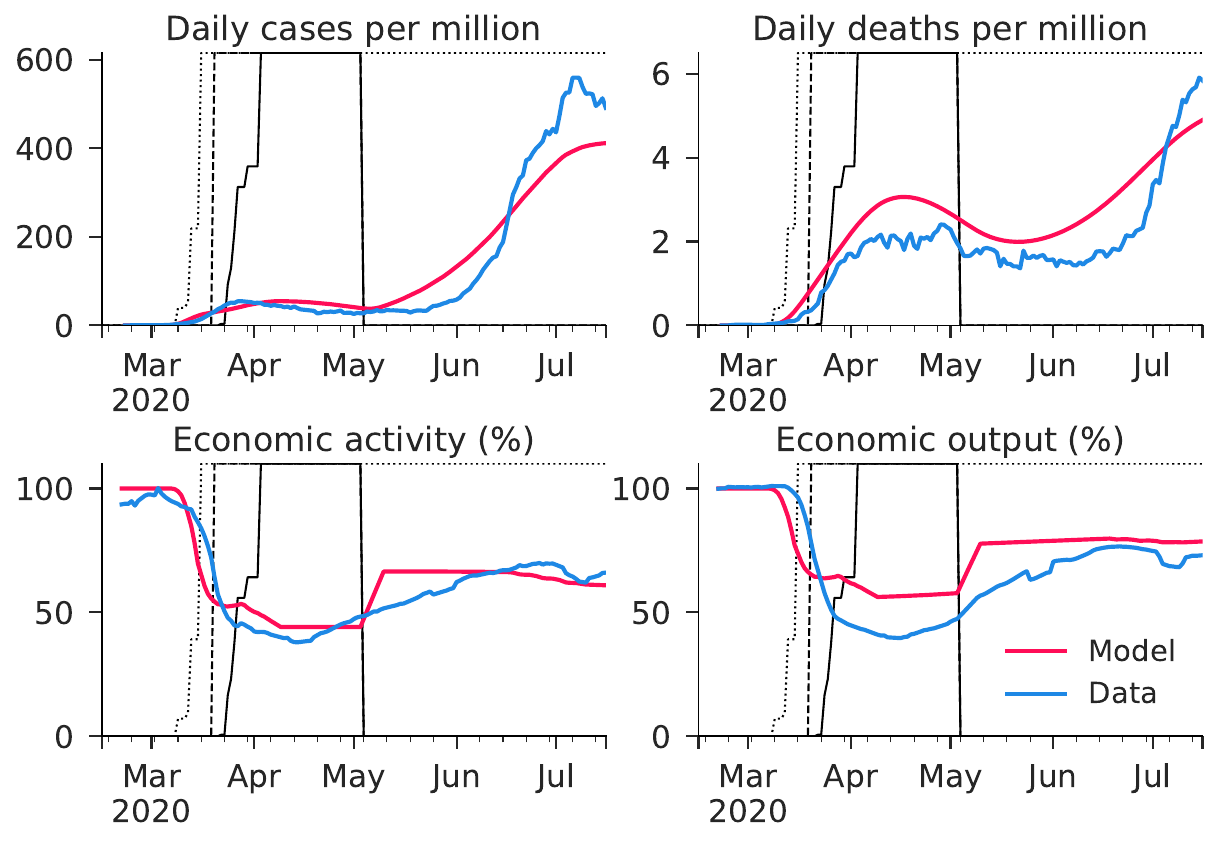}
      \label{fig:flo-rida}
    \end{subfigure}%
    ~~
    \begin{subfigure}{.49\textwidth}
      \centering
      \caption{California}
      \includegraphics[width=1.0\textwidth]{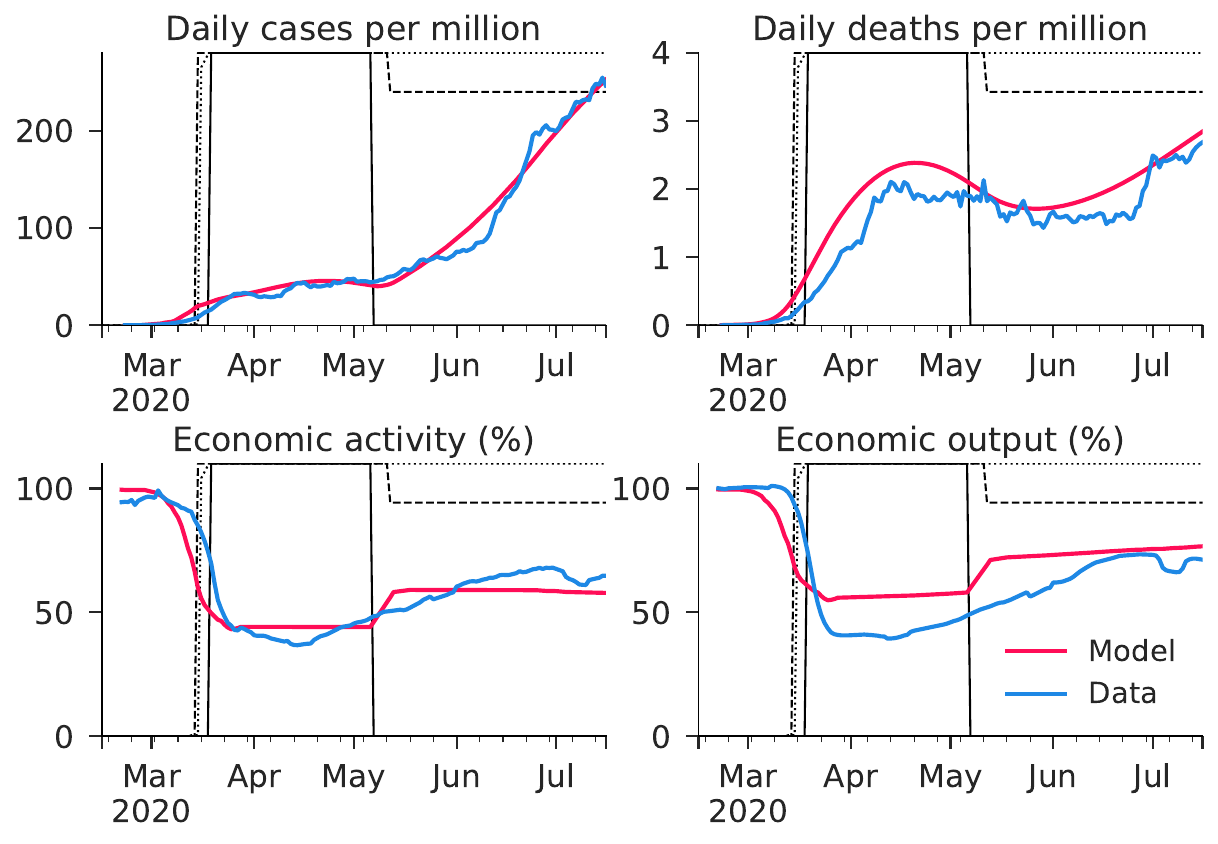}
      \label{fig:cali}
    \end{subfigure}
    \caption{State-level fit. Solid black line indicates stay-at-home order, dashed black line indicates restaurant take-out-only order, dotted black line indicates school closure.}
    \label{fig:modelfit_twostate}
\end{figure}

\paragraph{Seroprevalence.}  Since the model assumes underreporting of cases in the data, the model predicts that true cases always exceed the reported case data. 
Using seroprevalence studies we can quantitatively validate the model's view of the number of unreported \covid{} cases. This is an especially useful validation of the model mechanisms, as our estimation does not target these quantities at the few times and locations for which they are available. 

The CDC reports the results from 10 local seroprevalence studies, each conducted at two different points in time.\footnote{See \url{https://www.cdc.gov/coronavirus/2019-ncov/cases-updates/commercial-labs-interactive-serology-dashboard.html}.
}
Because our model provides daily county-level predictions for \covid{} exposure, we can generally construct the model's assessment of exposure in the same location, and at the same dates, as each seroprevalence study.
The first panel of \autoref{fig:reported_actual} shows the model's aggregate implications for true cases vs reported cases.
The second panel shows the model's implications for exposure in the locations and on the dates of the available seroprevalence studies.\footnote{The 
    values are taken from the CDC's Commercial Laboratory Seroprevalence Survey Data. 
    The CDC offers the following broad summary of the nature of the underlying data,
    ``The survey includes people who had blood specimens tested for reasons unrelated to COVID-19, such as for a routine or sick visit during which blood was collected and tested by commercial laboratories in participating areas from each of the 10 sites.'' 
    As such the data collection should not reflect \covid{} cases that are either patently symptomatic or have already resolved via mortality.
}
We find this a particularly interesting dimension of the model to evaluate because our estimation procedure does not ``target'' or in any way use these features of the data. Further, given the uncertainty over what fraction of past infected people demonstrate antibodies, the fact that the model slightly over-estimates the data point estimates is reassuring.




\begin{figure}[t]
    \centering
    \begin{subfigure}{.45\textwidth}
      \centering
      \includegraphics[width=1.0\textwidth]{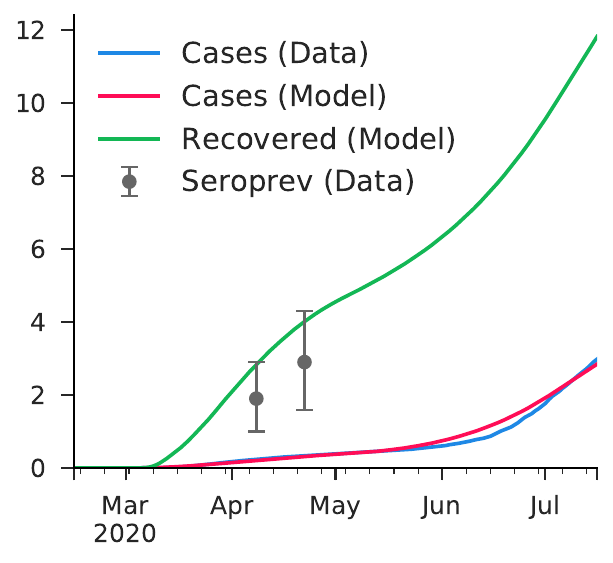}
      \caption{South Florida}
      \label{fig:reported_actual}
    \end{subfigure}%
    \begin{subfigure}{.55\textwidth}
      \centering
      \includegraphics[width=1.0\textwidth]{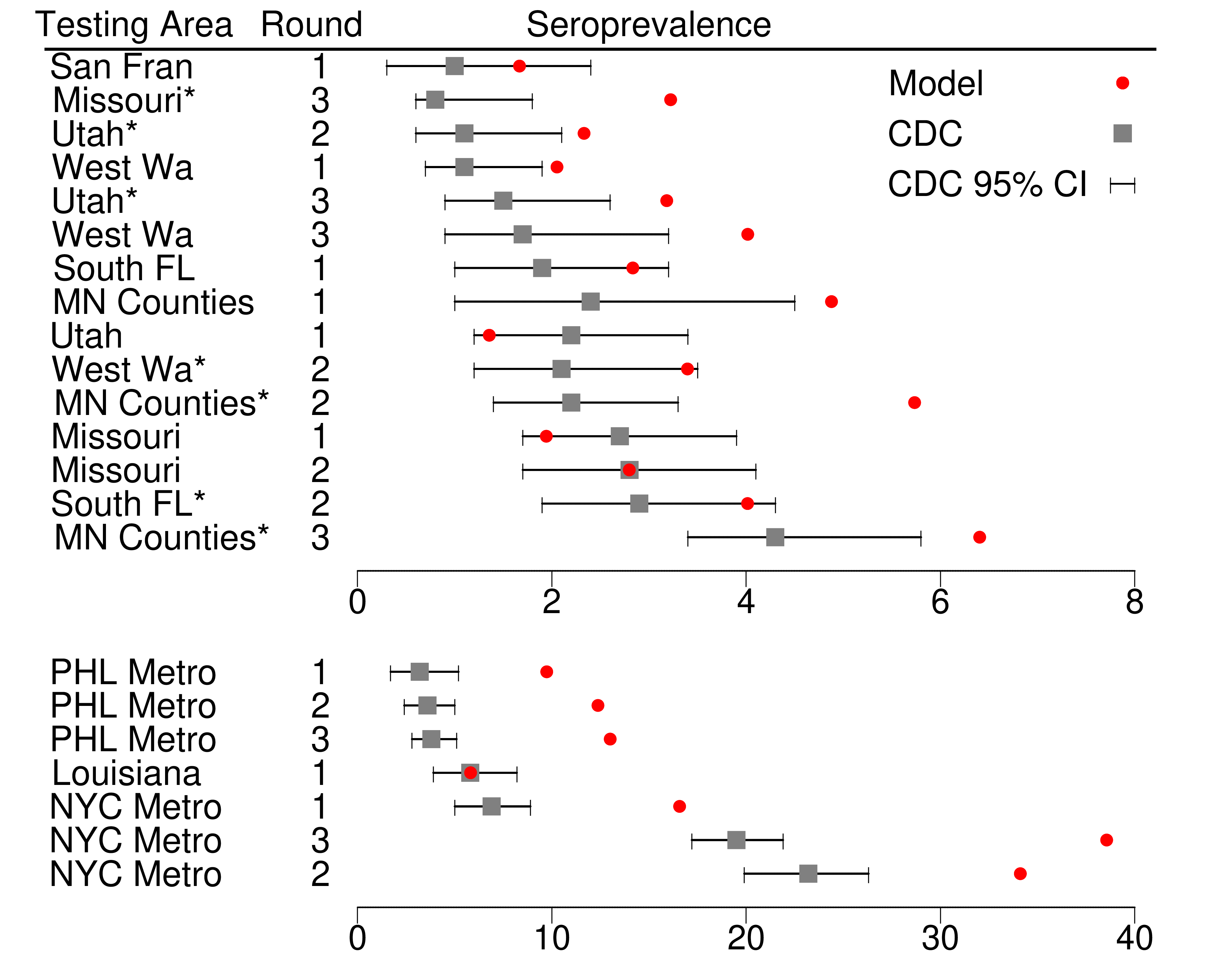}
      \caption{Seroprevalance Studies}
      \label{fig:sero}
    \end{subfigure}
    \caption{Seroprevalence in the model and data.}
    \label{fig:test}
\end{figure}

\section{Counterfactuals and Policy Experiments}
\label{sec:use_estimated_model}

In this section we examine the quantitative implications of the estimated Econ-\sirextend{} model for the trajectories of epidemiological and economic variables under different policy scenarios.

\paragraph{Eradication Zone.} Under some policies, the model's implications depend importantly on whether or not it is possible to eradicate the virus by driving the number of carriers sufficiently low.
An artifact of standard SIR-based models is that it is not possible to eradicate a virus until the population reaches herd immunity.
This is because some measure of agents remain infected forever in this class of models and, as long as a large share of Susceptibles remains, the virus will propagate in the presence of any positive measure of infectious agents. 
We will consider a version of the model where \(A\td\) and \(I\td\) jump discontinuously to $0$ if the total measure of infectious agents ($A\td + I\td$) falls below a particular threshold.
At which point, we say that the virus has been eradicated. 

When using this variant of the model, we entertain two sizes for this ``Eradication Zone,'' one at smaller (and hence weaker) one at 100 per million and a stronger one at 400 per million. Under reasonable asymptomatic transmission and case reporting assumptions, these $A+I$ numbers correspond very roughly to reported new case rates of 1 and 5 per million per day, respectively. Both of these zones are still fairly stringent, as daily reported case number in the US have been between 100 and 200 per million for some time. They are closer to numbers seen in S. Korea which are typically around 1 per million but spike up to 5-10 per million during small outbreaks.

\subsection{Laissez-Faire}
\label{sec:laissezherd}

\begin{figure}[t]
    \centering
    \includegraphics[width=\textwidth]{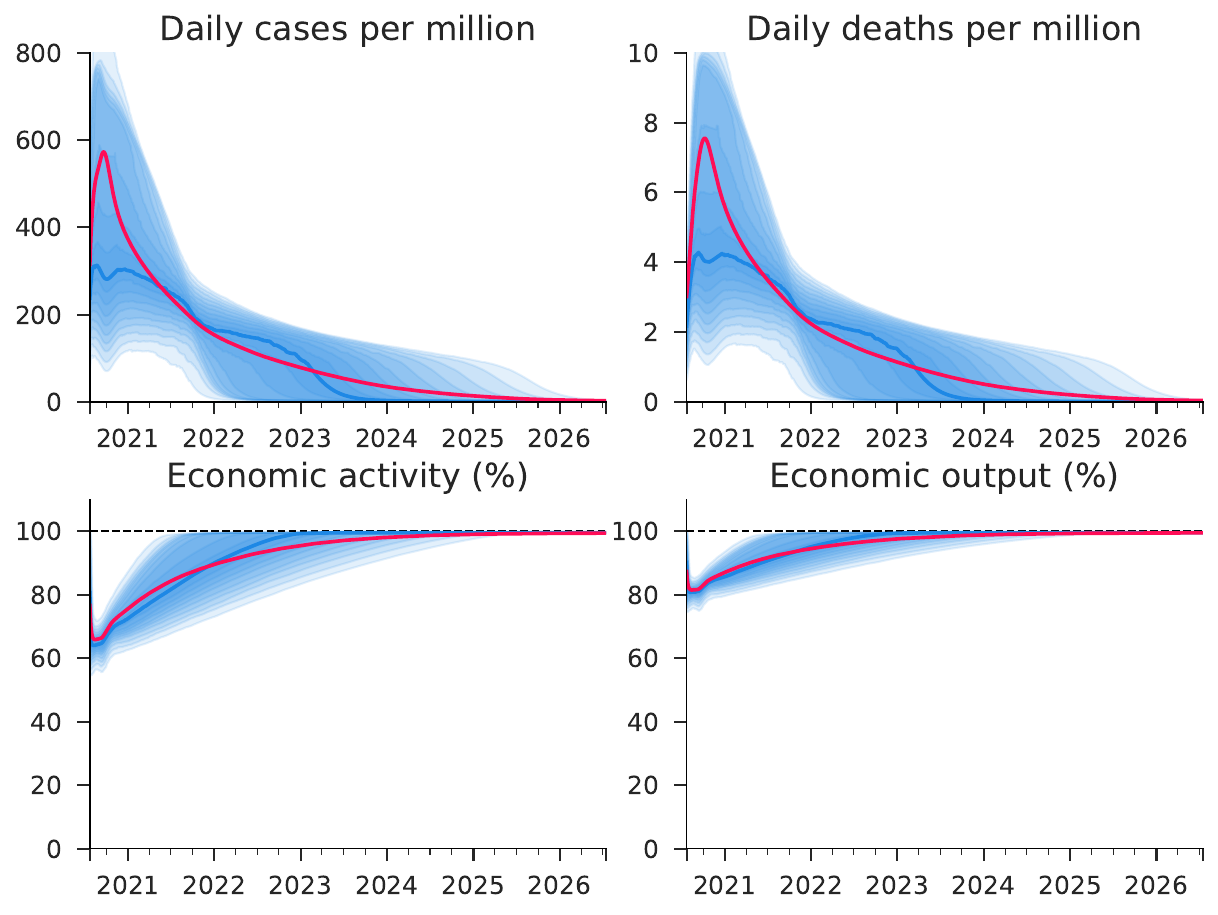}
    \caption{\textbf{Laissez-Faire.} Red lines are population-weighted means, thick blue lines are for the median county, and the light blue shaded bands indicate percentiles of the cross sectional distribution of counties.}
    \label{fig:laissezfaire_long}
\end{figure}

We first examine the scenario in which policymakers do nothing, i.e. laissez-faire policy.
This scenario demonstrates the behavior that the model absent NPIs simply from the endogenous responses of Susceptibles with respect to the force of infection. As in other SIR-based models, the virus continues to spread until the population reaches herd immunity. As discussed in Section \ref{sec:econsir}, infections peak earlier than the herd immunity threshold --- at the point at which the inflow rate of new infections equals the outflow rate of infectious agents.
In the Econ-\sirextend{} model achieving herd immunity in the full cross-section requires about 71 percent of the population having been infected. 
From a policy perspective, the important questions about the laissez-faire option are how long it takes to get to herd immunity, what are the costs of doing so, and, perhaps most importantly, what are the other options?

\autoref{fig:laissezfaire_long} shows the long term trajectories of economic and epidemiological variables under laissez-faire policy in the Econ-\sirextend{} model.
The critically different implication of our \econsir{}-based model, relative to a standard SIR model, is the protractedness of the path back to normality.
The endogenous virus avoidance of Susceptible agents, who reduce activity even in the absence of formal lock downs, generates a path of roughly 3 years to reach herd immunity, after which infections steadily die out over the subsequent 2 years.
Along this path, activity and output remain significantly below their pre-pandemic levels.
Our results crystallize the notion that ``reopening'' is easier said than done. 
As long as people fear infection from the virus, they will continue to curtail their economic activity and output will remain below normal levels.
As the model approaches herd immunity, economic activity increases and approaches its baseline level, ultimately yielding dynamics similar to those in models with a more simplistic infection mechanism, albeit occuring at a considerably later date in \econsir{}.
As this process plays out, cumulative mortality over the five years is about 0.38 percent of the population, which equates to roughly 1.25 million deaths in the United States. 

\subsection{Alternative Policies}
\label{sec:altpol}

With the model's baseline behavior as context, we next consider the outcomes under various policies.
In each case, we assume the policy is in effect from July 12, 2020 forward.


\begin{figure}[t]
    \centering
    \includegraphics[width=\textwidth]{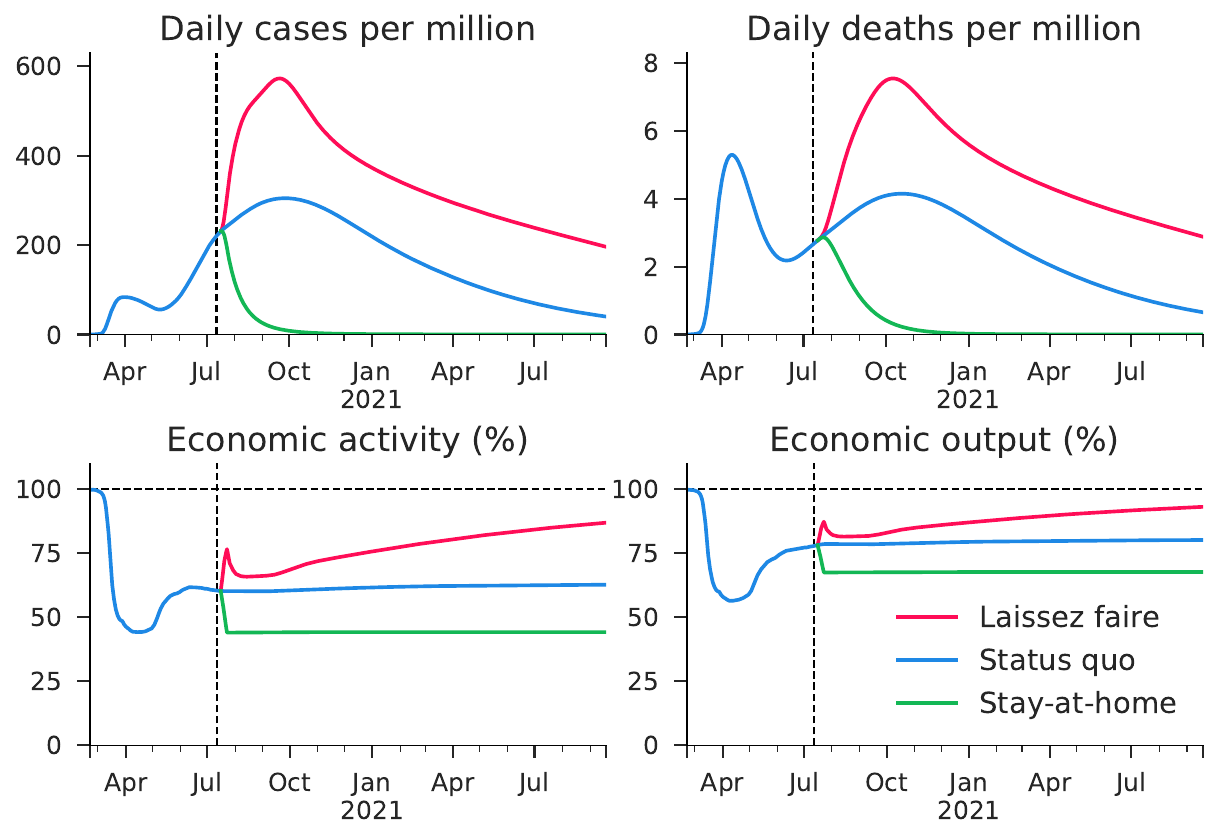}
    \caption{Forward simulations under alternative policies.}
    \label{fig:altpols}
\end{figure}

\begin{itemize}
    \item \textbf{Status quo.} The status quo option entails the continuation of whatever NPI policies were in place at the end of our estimation period.
    \item \textbf{Stay-at-home indefinitely.} (Self explanatory).
    \item \textbf{Severe temporary lockdowns, followed by full reopening.} It has been suggested that a severe temporary lockdown would be an effective way to reduce U.S. infection levels to the point where test-and-trace strategies would be viable for containment.  
    We consider the implications of such a severe shutdown if it were followed by simply a naive reopening.
    \item \textbf{Optimal adaptive local lockdowns.} We next consider optimal policy rules that enforce location-specific lockdowns of varying strength as a function of the local share of infectious agents.
\end{itemize}

\autoref{fig:altpols} shows the predicted trajectory of our benchmark variables under laissez-faire, status quo, and perpetual stay-at-home.
From these trajectories one can see the main tradeoffs that the model allows.
While laissez-faire generates the highest output (even despite some degree of private mitigation by agents), it does so at the cost of vastly more infections than the other options over the near term. 
At the other end of the spectrum, our estimated stay-at-home orders would be sufficient to dramatically curtain the virus's spread, they do at the cost of also dramatically curtailing economic output.

The long run implications of a temporary severe lockdown (if followed by a full reopening), depend critically on whether or not the eradication zone is ``fair game.''
\autoref{fig:templockdowns} in the appendix  
shows the model outcome given a fairly stringent lockdown lasting for two months.
In the short run, the lockdown is effective at reducing infection rates at the cost of reducing output by about 50 percent.
If the eradication zone exists, then this is an extremely effective long run strategy for knocking out the virus and fully normal activity resumes within 4 months. 
However, with no eradication zone, infections come surging back after only a few months.
On the economic side, there is a brief period of high output at the early phase of the virus resurgence, followed by a subsequent collapse due to renewed avoidance.


We next consider optimal policy rules that enforce location-specific lockdowns of varying strength as a function of the local share of infectious agents.
We optimize a policy rule of the form
\begin{align}
    \bar{z}_{j,t} 
    =
    \bar{z}^{\max} \cdot 
    \smoothstep \big( A_{j,t-1} + I_{j,t-1} \big)
    \label{eq:optruleform}
\end{align}
where \(\smoothstep(s)\) is the smoothstep sigmoid function with left edge of 0 and right edge of \(\bar{s}\).
The parameters to optimize in \eqref{eq:optruleform}
are the maximal lockdown intensity, \(\bar{z}^{\max}\), and the value of \(\bar{s}\) of local virus prevalence that triggers the maximal lockdown. 
For each simulation, we compute welfare via
\begin{align}
    W(t_0, T) 
    = 
    \sum_{t=t_0}^{T} \sum_{j=1}^{J} 
    \big( \ecoutput_{j,t} - \cost_{j} N_{j,t} \big).
    \label{eq:welfare}
\end{align}
and determine \((\bar{z}^{\max}, \bar{s})\) jointly as
\begin{equation}
    (\bar{z}^{\max}, \bar{s})
    = 
    \underset{\bar{z}^{\max}, \bar{s}}{\text{argmax}} 
    ~~ W(t_0, T)
    \label{eq:optpolicyprob}
\end{equation}
where \(t_0\) is the end of the estimation sample, and \(T\) is 365 days after \(t_0\). 
For each location, the model-implied trajectories of output \(\ecoutput_{i,t_0:\bar{T}}\) and new infections 
\(N_{i,t_0:\bar{T}}\) depend on the sequences \(\bar{z}_{i,t_0:\bar{T}}\) generated by the policy rule.
We will invoke the welfare calculation in \eqref{eq:welfare} in additional experiments going forward.
\(W\) is our preferred measure for comparing alternative policies, as it is the most internally-model-consistent summary of how agents value trajectories of the model's variables. 
Nevertheless, to see the model-implied tradeoffs in possibly more concrete dimensions, we also summarize the ramifications of alternative policies in terms of output and deaths.
One can then form judgments about a given policy's appropriateness without fully accepting the primacy of \(W\).

\begin{figure}[t]
    \centering
    \includegraphics[width=\textwidth]{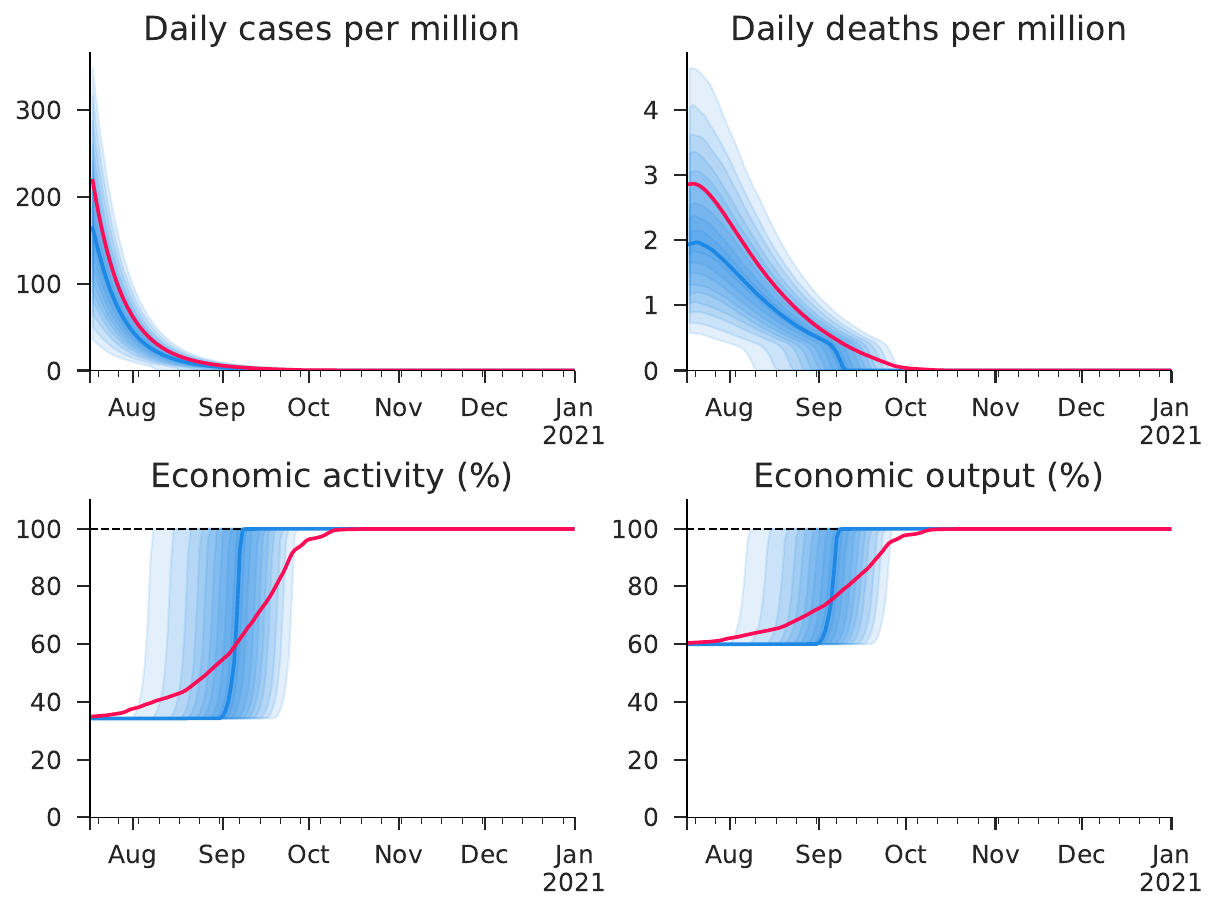}
    \caption{Optimal adaptive lockdown policy, with eradication zone.}
    \label{fig:opt_adapt_lock_kz}
\end{figure}

In the presence of a weak eradication zone (\autoref{fig:opt_adapt_lock_kz}), we estimate that the optimal policy takes the form of a lockdown level \(\bar{z}^{\max} = 1.05\) and \(\bar{s} = 212\) infectious agents per million. Meanwhile, with a strong eradication zone, we find the optimal policy to be \(\bar{z}^{\max} = 1.08\) and \(\bar{s} = 883\).
To make the implications more concrete, such a policy (when active in a given location) would restrict economic activity to a level somewhat lower than that of our estimated stay-at-home order. The threshold for introduction is based on total active cases per million. Mapping this into daily reported cases per million, we can divide by about 4 to account for asymptomatic cases and under-reporting and divide by 20 ($\approx 1/\delta$) to go from a stock to a daily flow (so a total of 80). In this case, the thresholds map to $2.7$ (for 212) and $11$ (for 883).



Within the model, the nature of the tradeoff between deaths and output depends importantly on both the time horizon and the eradication zone.
We follow others in the literature and refer to the set of achievable outcomes for deaths and output as the pandemic possibilities frontier, or PPF.
\autoref{fig:ppfs} plots the PPFs pertinent to 1 month and 1 year horizons under optimal adaptive lockdown with and without the eradication zone assumption.
Without the eradication zone, (top row) we find classical tradeoffs between output and deaths.
Note that the lower right most point in each PPF pertains to laissez-faire policy.  
If eradication is possible (bottom row), then the estimated model implies a PPF with an upward sloping region at a one year horizon.   
This is achievable because once the virus is eradicated, output returns to normal levels, which can in principle outweigh sharp initial drops in output associated with a temporary stringent lockdown to eradicate the virus.
It turns out, that such an approach yields gains in both dimensions, generating the upward sloping portion of the PPF.

\begin{figure}[t]
    \centering
    \includegraphics[width=\textwidth]{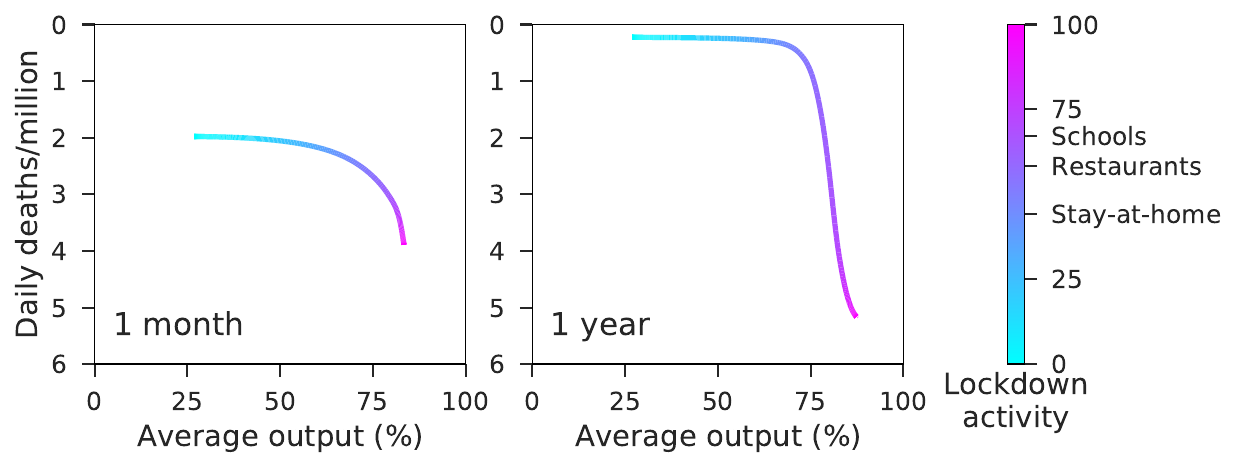}
    \includegraphics[width=\textwidth]{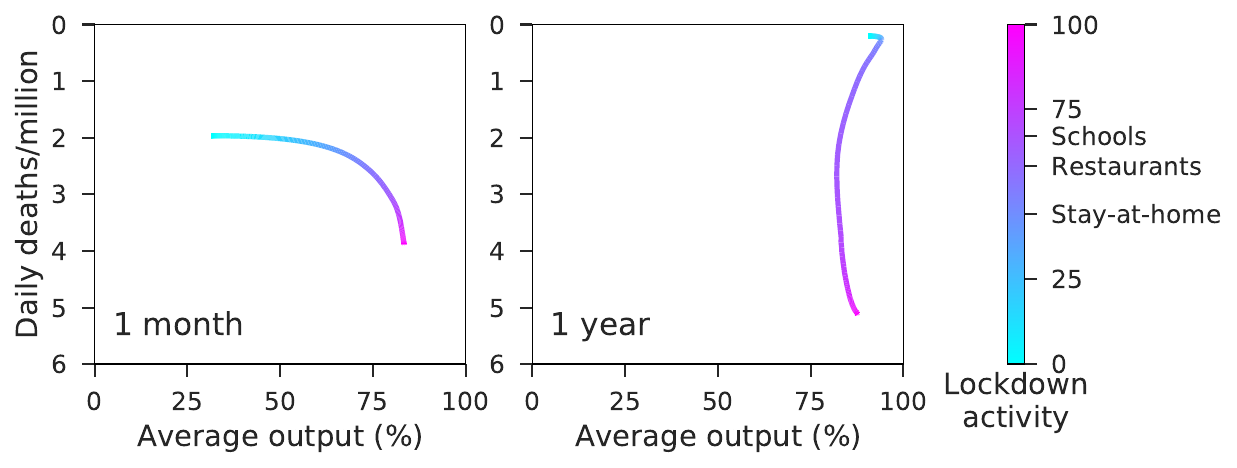}
    \caption{Forward looking PPF under adaptive lockdown policy. Top row of panels assumes no eradication zone exists. 
    Bottom row of panels assumes a strong eradication zone.}
    \label{fig:ppfs}
\end{figure}



\subsection{Optimal Policy when a Vaccine is Near}

The previous section's results underscore the fact that there is little scope for NPIs to decrease the virus's final toll if eradication is impossible and a vaccine never arrives.
In this section we consider optimal policy when a vaccine is likely in the relatively near term --- one year from now. The question then becomes, would any NPIs be worthwhile over such a time horizon?

One can develop the basic intuition for why NPIs could yield welfare gains under such a scenario by re-examining the laissez-faire path in \autoref{fig:laissezfaire_long}.  
Along the laissez-faire path, cases and deaths are frontloaded. 
With a vaccine ready at time \(T_{vax}\), we know that the virus's total cost beyond \(T_{vax}\) will be small regardless of how many infections have occurred by that date.
In that case, it could conceivably be optimal to prevent the burst of infections and deaths that occur in the pandemic's early stages since the vaccine renders it unnecessary ever to incur them at all.
One can contrast this with the scenario in which no vaccine could ever be developed, in which case simply shifting the timing of the costs (i.e. whether negatives are frontloaded or backloaded) would be largely irrelevant (for a sufficiently patient planner).

\begin{figure}[t]
    \centering
    \includegraphics[width=\textwidth]{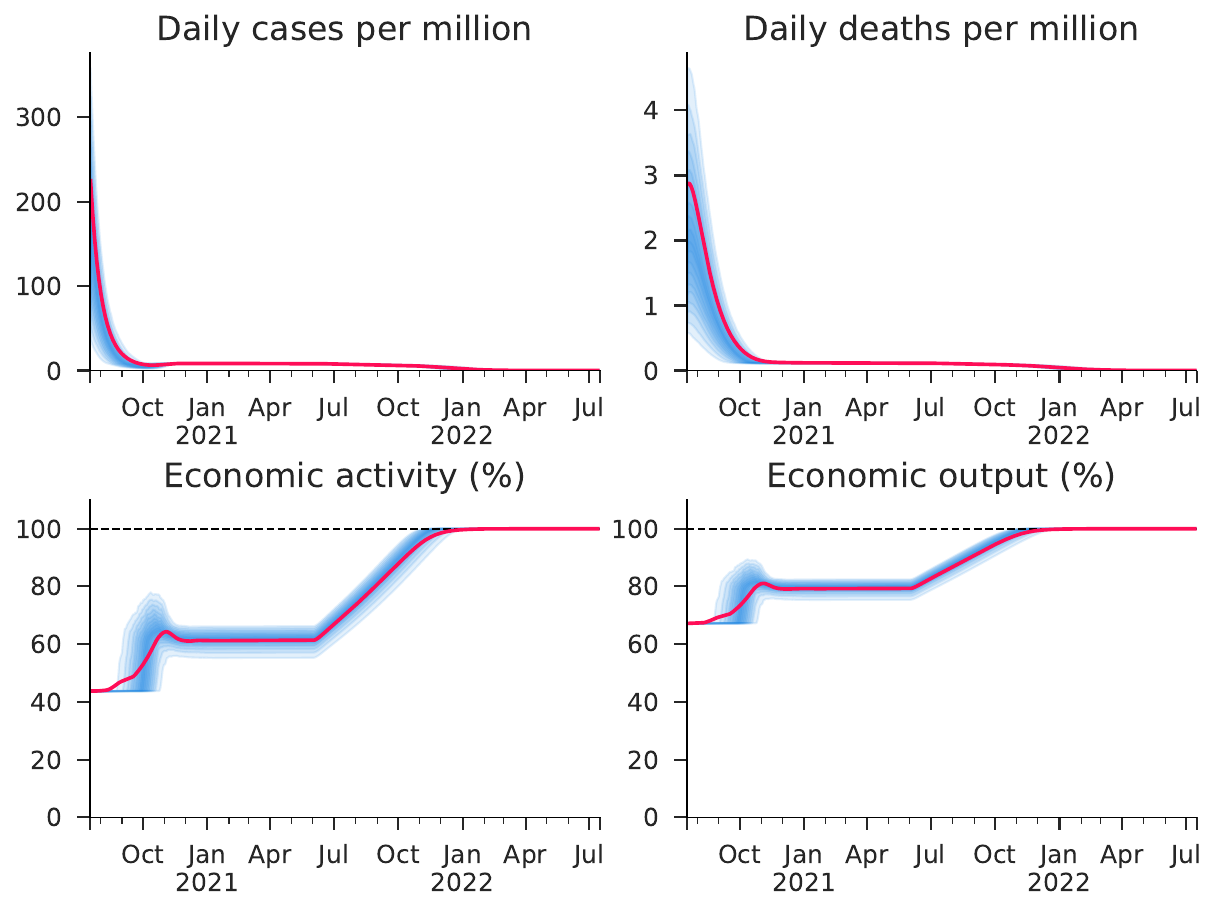}
    \caption{Optimal adaptive lockdown with vaccine 1 year away.}
    \label{fig:optadaptlock_vax1y}
\end{figure}

In particular we compute an optimal policy of the form of \eqref{eq:optpolicyprob}, tailored to the case of a vaccine 1 year away, and beginning from the observed state of the world in mid-July 2020.  
\autoref{fig:optadaptlock_vax1y} shows the trajectories of the four benchmark variables under the optimal adaptive NPI policy when a vaccine is one year away.\footnote{We also considered the cases of 6 months and 2 years, which turned out to give results very similar to the 1 year case presented here. However, if it takes more than 4 years for a vaccine to arrive the optimal policy again resembles the laissez-faire equilibrium.}   
The resulting trajectories of cases and activity are particularly informative for understanding the form of the policy.\footnote{When the vaccine arrives, we model it as being administered to 2 million people per day, or 0.6 percent of the population.  In the model, vaccinated agents are a mass of \(S\td\) shifted directly to \(R\td\).}  
In light of the initial stock of infectious agents, the policy begins by imposing a considerable reduction in activity to levels similar to our estimated effects of stay-at-home orders (activity slightly above 40 percent, output just below 70 percent).  
This phase remains in effect in most locations for 2 to 4 months, depending on the initial stock of infectious agents, during which time the activity reduction is great enough to dramatically reduce new infections and the stock of infectious agents.
After this period, the optimal NPI loosens modestly, but is not removed entirely.  
The policy then prescribes a roughly 8 month period of reduced activity at levels similar to our estimated restaurants-takeout-only orders (activity around 60 percent, output around 80 percent).   
Having already reduced the stock of infectious agents, the somewhat higher activity level remains sufficient to keep the effective reproduction number around 1, and thus keep cases from-reentering an exponential growth regime.
Once the vaccine arrives and is distributed, the measure of Susceptibles begins to fall and NPIs can be safely progressively weakened while maintaining low levels of new infections.  
Finally, activity and output would return to normal in the fourth quarter of 2021.


\subsection{The Costs of the Pandemic}

\autoref{fig:welfarecost} shows the welfare cost of the pandemic, from July 2020 onward, under a variety of scenarios. 
Lower values in the figure indicate higher welfare. 
So that one might see the way different policies tradeoff the two welfare channels, we decompose the total cost along each path into the cost of infections and cost of lost output.  
Relative to the other policy options we consider, the key feature of the laissez-faire policy is that its relatively modest output losses are outweighed by its large disease cost, provided that a vaccine is on the horizon or eradication is possible.  
Conversely, the other options typically entail somewhat larger output costs, particularly early on, but with the benefit of large reductions in total infections.

\begin{figure}[t]
    \centering
    \includegraphics[width=0.92\textwidth]{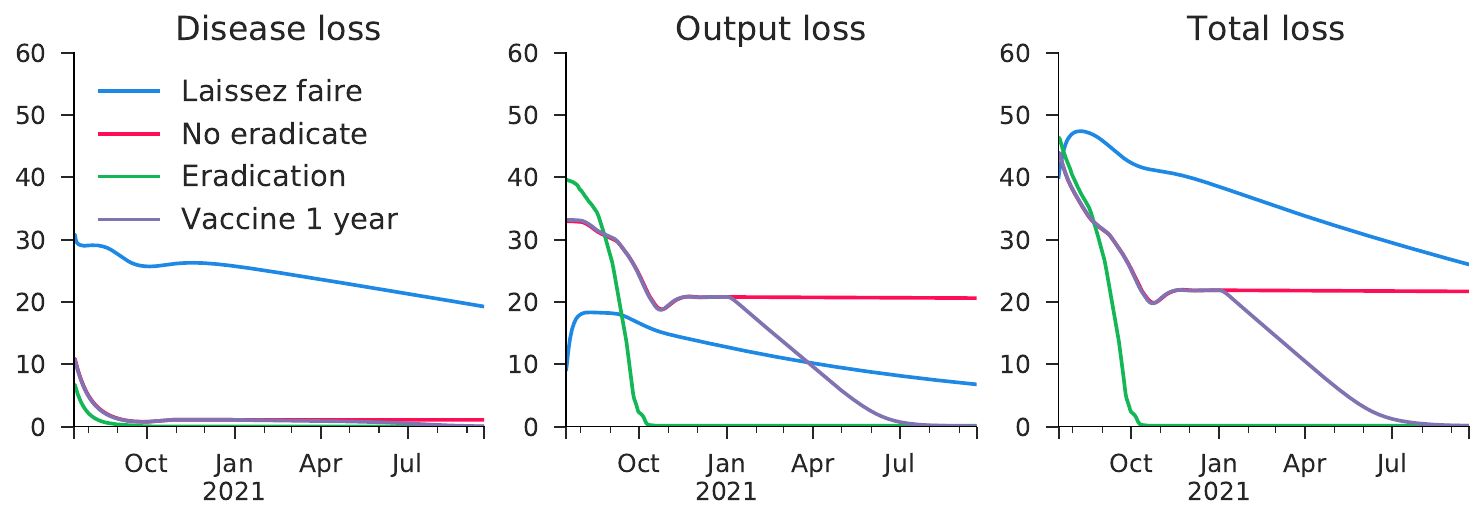}
    \caption{Period by period costs of the pandemic under various policy scenarios.}
    \label{fig:welfarecost}
\end{figure}


\section{Conclusion}
\label{sec:conclusion}

In this paper we developed and estimated a spatial model of the joint evolution of economic variables and the spread of an epidemic across U.S. counties using high-frequency granular data. The model predicts a significant endogenous reduction in economic activity by agents in response to the spread of the virus --- highlighting the importance of using an integrate-assessment model for evaluating policy responses to epidemics.
Absent pharmaceutical advances to vaccinate against or treat the virus, our estimated model predicts a protracted march towards herd immunity. Infections and deaths continue to increase roughly linearly for the next 3 to 5 years, with an ultimate death toll of around 1,250,000.

Lockdowns effectively lower the spread of the virus at the cost of lower economic output --- giving rise to a pandemic possibility frontier. From a welfare perspective, however, lockdowns are desirable only under two conditions. First, if the policymaker expects a vaccine or cure to arrive within the next 2-3 years, lives saved from locking down outweigh the output costs. Second, suppose there is the possibility to curtail community transmission (either through eradication or management through testing, tracing, and quarantine). In that case, a strict lockdown for a few months can affect a quick return to pre-pandemic levels of activity. Policies that lower the rate of transmission (e.g., mask wearing), however, can successfully raise output and save lives.

One caveat to our conclusions is that we have assumed that individuals receive permanent immunity to the virus after recovering. Recent medical evidence suggests that reinfection is rare but possible (albeit in a milder form). While beyond the scope of the current analysis, temporary immunity would likely increase the welfare benefits of both finding a vaccine and curtailing community transmission through NPIs.


\clearpage



\bibliographystyle{chicago}
\bibliography{covid} 

\clearpage

\appendix


\renewcommand\thefigure{\thesection.\arabic{figure}}    
\setcounter{figure}{0}

\section{County-Level Data at Daily Frequency}
\label{sec:dataappend}

For the empirical analysis and model estimation, we build a panel of daily county-level observations on two economic variables, two measures of virus spread, and NPIs.
While the economic variables can be measured at even finer geographic granularity, we aggregate each of them to the county level to match the finest available granularity of the epidemiological data. We use daily observations from January 1, 2020 through May 24, 2020.

\paragraph{Epidemiological Data.} On the epidemiological side, we obtain daily county-level data on confirmed and probable \covid{} cases and deaths from the \textit{The New York Times}.\footnote{Data from \textit{The New York Times}, based on reports from state and local health agencies. See \url{https://www.nytimes.com/interactive/2020/us/coronavirus-us-cases.html}.} 
The \textit{The New York Times} collects this data on a recurring basis from local public health authorities.\footnote{Note that it is unclear in the mortality data if the date corresponds to when the death occurred or when the death was officially reported. 
In our model estimation we account for this uncertainty by allowing for  measurement error in recorded mortality outcomes.}
This is with one exception: the NYT data only records deaths for NYC as a whole, although the Bronx and Queens are distinct counties. 
We collect the county-level observations directly from New York City's website and add them to our data set.

One shortcoming of the epidemiological data is that both cases and deaths are subject to reporting delays, i.e. a discrepancy between when the medical event occurred for the individual and the day to which the event is ascribed or reported in the data. 
A few states publicly report the time series of assigned days of death, in addition to the more commonly available ``new deaths.''
We use the observed gap between the two series in some of these states to inform our our choice of reporting delay during estimation.
\autoref{fig:deathdelay} shows how the average reporting delay has evolved over time in three such states. 

\begin{figure}[t]
    \centering
    \caption{Average Delay in \covid{} Death Reporting}
    \includegraphics[width=0.6\textwidth]{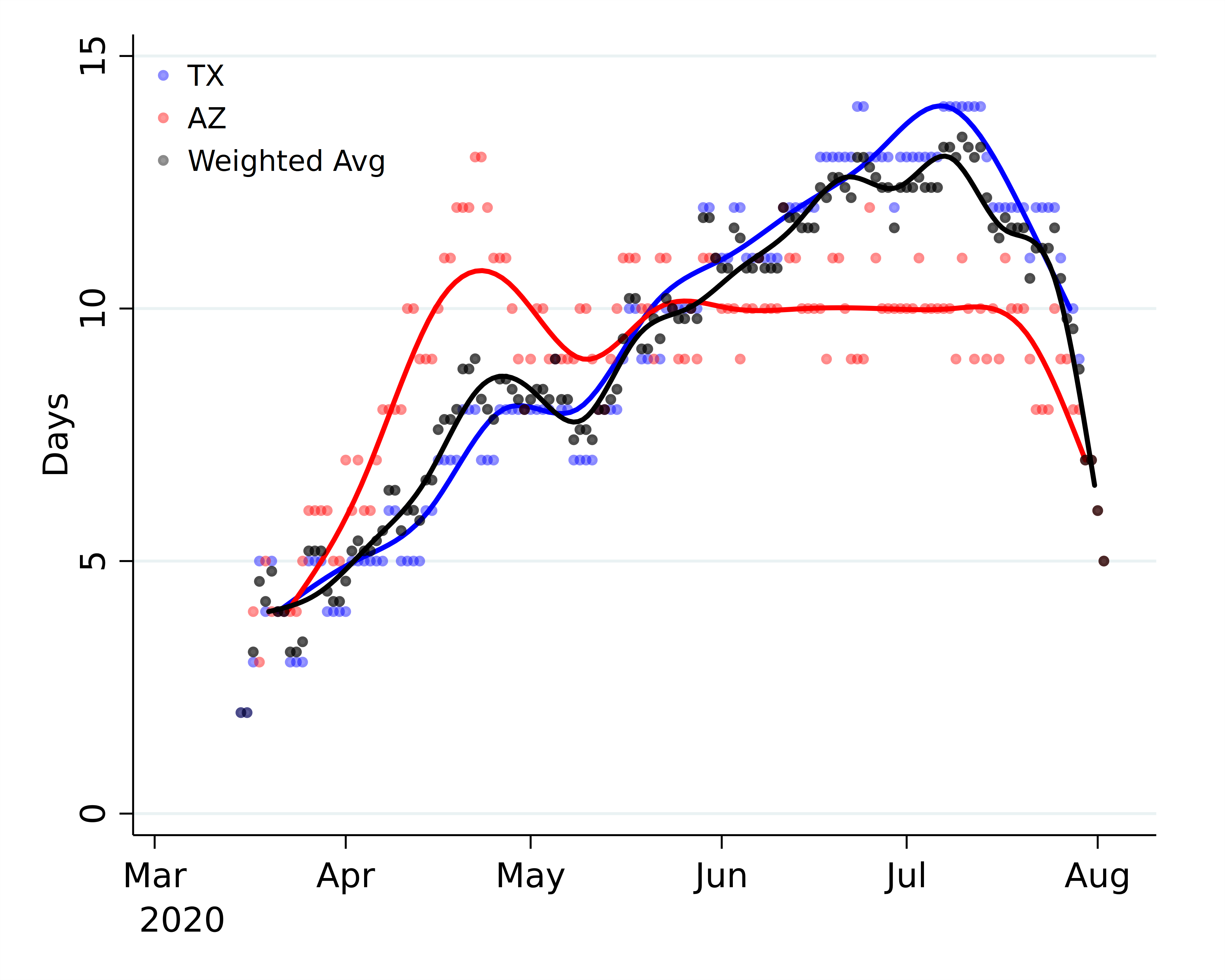}
    \label{fig:deathdelay}
\end{figure}

\paragraph{Economic Data.} Our high-frequency economic variables consist of detailed data on both foot traffic and hours worked. 
We obtain data on foot traffic to over 3.6 million consumer ``Points-of-Interest'' (POIs), including grocery stores, restaurants, hospitals, and an array of other public and private establishments from SafeGraph Inc. 
SafeGraph collects the POI data as anonymized GPS location information from a large panel of more than 45 million smartphones in the U.S. 
We view this data as a highly accurate measure of a variable we will refer to throughout the paper as ``economic activity'' or just ``activity.'' 
Being derived from a panel of smartphones, the Safegraph data obviously does not represent a random sample of the U.S. population. However, \cite{squire:2019} investigated potential sources of sampling bias in the SafeGraph data. 
The data are well representative at the county level (when compared to U.S. Census data from the American Community Survey) along a number of demographic dimensions, such as educational attainment and household income, although very low-income individuals (less than $\$10,000$ annual income) tend to be under-represented in the data.\footnote{The SafeGraph analysis does not address bias in age sampling, but we suspect that given that the panel is based on smartphone users that very old households would tend to be under-represented.}

We obtain hours worked from a large worker-firm matched data set provided by Homebase. 
Homebase is a scheduling and time clock software provider. 
The data consist of daily time-card records at the establishment level. 
The time-cards include information on hours worked and wages. 
The data are anonymized, but do specify the zip code and industry of the establishment. 
The data are provided in real-time when the time cards are reported by the business that uses the software. 
As we detail in the appendix, small establishments and establishments in the food services and retail sectors tend to be overrepresented in the Homebase data. 
We refer the interested reader to \cite{kurmann} for a more thorough discussion of the representativeness of the Homebase data and how it compares to the Quarterly Census of Employment and Wages conducted by the U.S. Census. 

\paragraph{Policy.} Finally, to accurately estimate the role of NPIs, it is important to know when NPIs were in place in each location.
To do so we use data from multiple sources, including other researchers and our own collection efforts using local and national news sources.
A challenging aspect of the NPI accounting is that NPIs orders  have been separately enacted at both state and local levels, with local orders often leading state-level action in the early phase of the pandemic.  
Because our model operates at the county level, we must document NPI orders from both state and local levels of government so that the effective policy environment for a given county is correctly characterized at each date. 

To get a sense of how much the effective policy environments differ from those one would infer from simply using state-level orders, the left panel of \autoref{fig:stayathome_state_vs_local} shows, at each date, the share of our sample's population affected by a ``stay-at-home'' order and whether or not the order was uniform across the state.
One can see considerable discrepancies between the population shares subject to orders from any level of government and population shares subject to "within-state homogeneous orders."
The right panel focuses on the discrepancies, showing the gap between the left panel's two lines at each date. 
One can see that, on multiple dates, the policy environment faced by nearly 25 percent of the population would not be correctly characterized by a researcher recording only state-level NPIs. 

\begin{figure}[t]
    \centering
    \caption{Population Shares Under ``Stay-at-home'' Orders -- Statewide vs. Locality-specific}
    \includegraphics[width=0.8\textwidth] {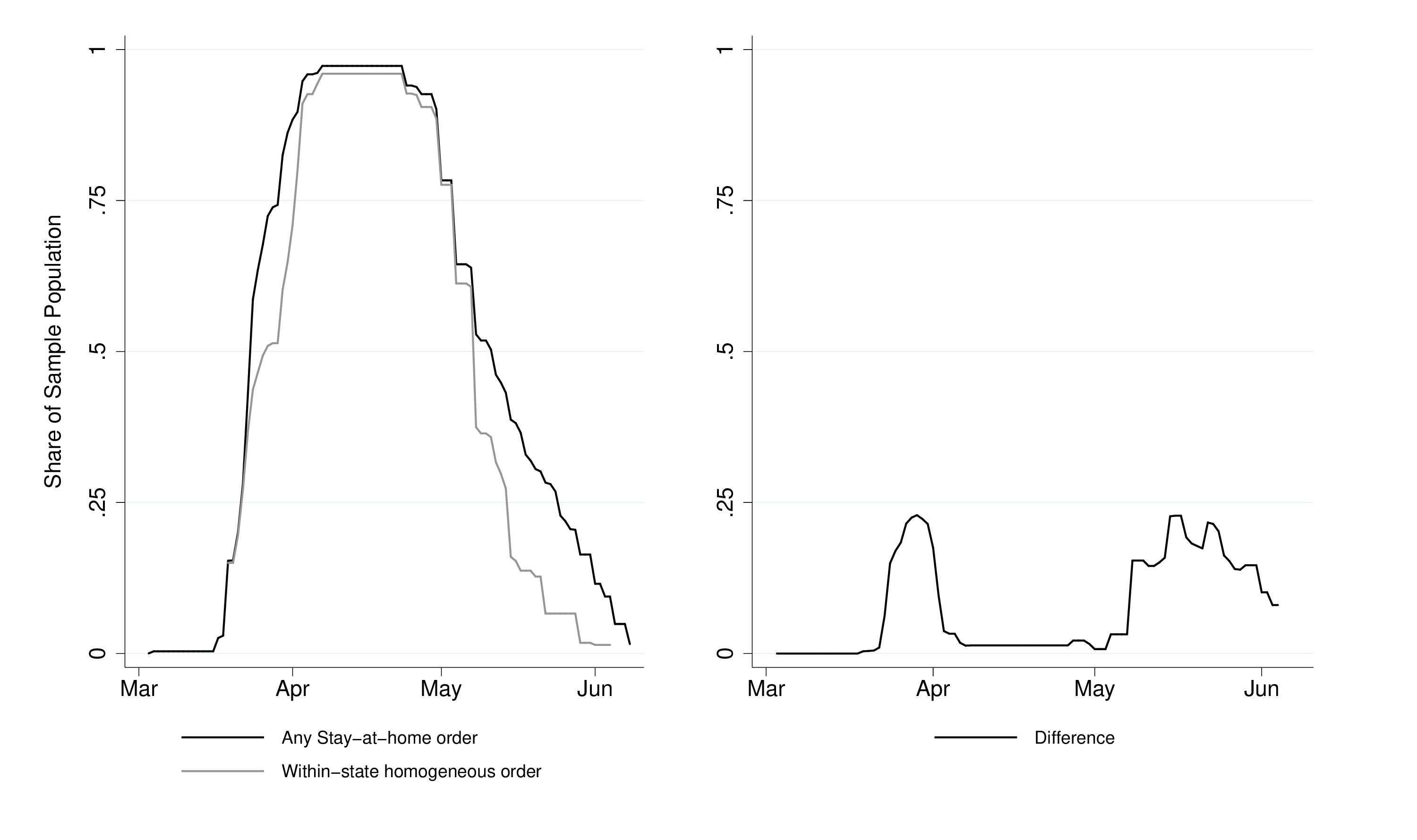}
    \label{fig:stayathome_state_vs_local}
\end{figure}


\begin{table}[]
    \centering
\begin{tabular}{lllllll} \hline
NAICS & Firms HB & Firms QC & Wages HB & Wages QC & Weekly Wage HB & Weekly Wage QC \\ \hline
21 &     & 0.00 &     & 0.01 &       & 2,406 \\
22 &     & 0.00 &     & 0.01 &       & 2,641 \\
23 &     & 0.08 &     & 0.06 &       & 1,193 \\
31 &     & 0.04 &     & 0.12 &       & 1,419 \\
42 &     & 0.06 &     & 0.06 &       & 1,612 \\
44 & 0.14 & 0.11 & 0.15 & 0.07 & 331    & 636 \\
48 & 0.01 & 0.03 & 0.01 & 0.04 & 363    & 1,086 \\
51 &     & 0.02 &     & 0.05 &       & 2,506 \\
52 &     & 0.05 &     & 0.12 &       & 2,867 \\
53 &     & 0.04 &     & 0.02 &       & 1,256 \\
54 & 0.05 & 0.13 & 0.04 & 0.12 & 336    & 1,949 \\
55 &     & 0.01 &     & 0.05 &       & 3,001 \\
56 &     & 0.06 &     & 0.05 &       & 834 \\
61 & 0.04 & 0.01 & 0.02 & 0.02 & 158    & 973 \\
62 & 0.08 & 0.17 & 0.06 & 0.13 & 282    & 963 \\
71 & 0.02 & 0.02 & 0.02 & 0.01 & 181    & 712 \\
72 & 0.37 & 0.07 & 0.46 & 0.04 & 205    & 420 \\
81 & 0.04 & 0.09 & 0.02 & 0.02 & 313    & 758 \\
99 & 0.25 & 0.02 & 0.22 & 0.00 & 242    & 1,070 \\ \hline
\end{tabular}    \caption{Comparison of Homebase and Quarterly Census of Employment and Wages}
    \label{tab:my_label}
\end{table}


\newpage

\section{Econ-SAIRD Model Details}
\label{sec:saird_model_details}

\subsection{Epidemiological Environment}

\paragraph{Epidemiological State Transitions.} 
As in \econsir{}, the \(S\) agents contract the virus at rate \(n\td\).
The \(S\) category has no inflows from other states and hence the measure of Susceptibles decreases over time according to 
\begin{align}
     \dot{S}\td &= - n\td S\td.
     \label{eq:delta_S}
\end{align}

We now assume that newly infected agents transition to the asymptomatic state.
This feature of the model is consistent with reports that, when initially infected, agents are often asymptomatic but able to spread the virus for multiple days prior to developing symptoms. 
Asymptomatic agents transition to \(I\) at rate \(\lambda\).
Furthermore, since many virus carriers never exhibit symptoms at all, \citep{russell2020estimating}, we assume that agents in the \(A\) state can recover without ever developing symptoms, transitioning to \(R\) at rate \(\gamma\).
The measure \(A\td\) thus changes according to 
\begin{align}
    \dot{A}\td 
    &= n\td S\td - \lambda A\td - \gamma A\td
    \label{eq:delta_E}
\end{align}

The measure \(I\td\) increases with the inflow of \(\lambda A\td\). 
As before, \(I\) agents recover at rate \(\delta\). 
We now also assume that \(I\) agents die at rate \(\kappa.\)
\(I\td\) thus evolves according to
\begin{align}
    \dot{I}\td 
    &= \lambda A\td - \delta I\td - \kappa I\td
    \label{eq:delta_I}
\end{align}

Lastly, the \(R\) and \(D\) states are absorbing states, with their measures only increasing over time, as they both have inflows but no outflows.\footnote{We are implicitly assuming that past infection to \covid{} provides future immunity, at least over the time horizon relevant to the model estimation and usage.}
With asymptomatic and infected agents recovering at rates \(\gamma\) and \(\delta\) respectively, \(R\td\) increases at rate
\begin{align}
    \dot{R}\td &= \gamma A\td + \delta I\td.
    \label{eq:delta_R}
\end{align}
\(D\td\) grows with the deaths of \(I\) agents, and hence
\begin{align}
    \dot{D}\td &= \kappa I\td.
    \label{eq:delta_D}
\end{align}
Given the parameters governing epidemiological state transition rates, \(\lambda, \delta, \kappa,\) and \(\gamma\), as well as the  infection rate of Susceptibles, \(n\td\), 
equations \eqref{eq:delta_S}--\eqref{eq:delta_D} characterize the epidemiological dynamics in the economic-\sirextend{} model.

\subsection{Activity and Output}

\paragraph{Economic Activity.}  The rate of activity is dictated by the recovered agents and the active shares of Susceptibles, Asymptomatics, and Infecteds.
One can then characterize the rate of activity either constructively, as in equation \eqref{eq:totalactivity_con},  
or as a deviation from the full activity rate,
as in equation \eqref{eq:total_activity_rate}:
\begin{align}
    a\td 
    &= 
    a_{S}\td 
    \big[ S\td + A\td \big] 
    +
    a_I I\td + R\td 
    \label{eq:totalactivity_con}
    \\
    &=
    1 -
    \!\!\!\!\!\!
    \underbrace{F(\bar{z}\td)}_{
        \substack{
            \text{rate of foregone} \\ 
            \text{excursions by \(S\) \& \(A\)}
        }
    }
    \!\!\!\!\!\!\!\!
    \cdot
    \big[ S\td+A\td \big]
    - 
    \underbrace{
        F(\bar{z}_{I}) I\td 
        - D\td
    }_{
        \substack{
            \text{missing excursions by} \\ 
            \text{quarantined and deceased}
        }
    }
    \label{eq:total_activity_rate}
\end{align}

\paragraph{Economic Output.}  Similarly, output can be equivalently characterized either constructively or in terms of shortfall from output when there are no infections: 
\begin{align}
    \ecoutput\td
    \ = \ &
    \overbrace{
        \mathbb{E}[z\td] 
    }^{
        \substack{
            \text{``no virus''} \\ \text{output rate}
        }
    }
    - 
    ~
    \big[ S\td + A\td \big]
    \cdot
    \!
    \!
    \overbrace{
        F(\bar{z}\td)
    }^{
        \substack{
            \text{\(S\td\) foregone} \\ 
            \text{excursion rate}
        }
    }
    \!\!
    \cdot
    \overbrace{
        \mathbb{E}[z\td | z\td \leq \bar{z}\td]
    }^{
        \substack{
            \text{value of foregone} \\
            \text{excursion}
        }
    } 
    \\
    & - 
    \underbrace{
        F(\bar{z}_{I}) I\td
        \mathbb{E}[z\td | z\td\le\bar{z}_I]
    }_{
        \substack{
             \text{foregone value by quarantined \(I\)}
        }
    }
    \ - \ D\td \cdot \mathbb{E}[z\td]
    \notag
\end{align}

\subsection{Epidemiological Statistics}

Once a person enters the Asymptomatic state, their transitions are stochastic but nonetheless exogenously determined by parameters. 
This makes it relatively straightforward to calculate standard epidemiological statistics. 
For instance, the case fatality rate (CFR), the number of deaths per confirmed cases, which in our model corresponds to the probability of dying conditional on reaching the Infected state is
\begin{align}
    CFR = \frac{\kappa}{\delta+\kappa}
\end{align}
Similarly, the infection fatality rate (IFR), the number of deaths per total cases, which (somewhat confusingly) corresponds to the probability of dying conditional on reaching the Asymptomatic state is
\begin{align}
    IFR = \frac{\lambda}{\gamma+\lambda} \cdot CFR = \frac{\lambda}{\gamma+\lambda} \cdot \frac{\kappa}{\delta+\kappa}
\end{align}
Next we turn to the various reproductive rates. Because there are many time varying quantities in our model (both from endogenous and exogenous sources), it is difficult to express in closed form how many additional individuals an infected person would expect to infect. However, we can answer that same question assuming conditions at that time prevailed in perpetuity. In this case we find
\begin{align}
    \mathcal{R}_e(t) = \mathcal{R}_0 \cdot e(t) \cdot a_S(t) S(t)
\end{align}
where we again have $\mathcal{R}_0 = \beta/\delta$. Notice that this figure reflects the fundamental reproductive number $\mathcal{R}_0$ as well as contributions from the exursions of Asymptomatic and Infected $e(t)$, the activity of Susceptibles $a_S(t)$, and the share of Susceptibles $S(t)$.


\section{Additional Estimation Details}

\subsection{Initial State}
\label{sec:initialstate}

We do not observe the full underlying state of each county at the start of their outbreak.
From the data, we know the initial number of confirmed cases (infected people) and deaths (generally zero), from which we can set \(I_0\) and \(D_0\).  
However, we do not know the number of asymptomatic or recovered individuals. 
We assume that \(R_{0} = 0\) and that \(A_0\) is proportional to \(I_0\), with $\digamma_0 \equiv A_0/I_0$ denoting the ratio between the two. 

Identifying $\digamma_0$ from the data is challenging, so, 
rather than treat it as an additional free parameter, we derive a theory-driven estimate of it as a function of other model parameters.
Letting \(\digamma\td \equiv A\td / I\td\), the law of motion for $\digamma\td$ is
\begin{align}
    \frac{\dot{\digamma}\td}{\digamma\td} 
    &= 
    n\td 
    \frac{S\td}{A\td} 
    - \lambda - \gamma - \lambda 
    \frac{A\td}{I\td} 
    + \delta + \kappa 
    \\
    &= \infratexc S\td 
    \left( a_A\td + a_I \frac{1}{\digamma} \right) 
    - \lambda - \gamma - 
    \lambda \digamma\td 
    + \delta + \kappa
\end{align}
In the early stages of an outbreak, 
$S\td \approx 1$ and 
$a_A\td \approx a_S\td \approx 1$. 
Further assuming $a_I \approx 0$ and imposing stationarity of \(\digamma\td\) ratio yields the estimate
\begin{align}
    \digamma_0 
    = \frac{\infratexc+\delta+\kappa-\lambda-\gamma}{\lambda}
    .
\end{align}

\subsection{Parameter Estimates}

\autoref{tab:parameterestimates} gives the estimates of the model parameters. 

\begin{table}
    \centering
    \begin{tabular}{cl}
    \toprule
    Model Parameter & Value \\
    \midrule
    $\mu[\beta]$ & 0.6167 \\
    $\sigma[\beta]$ & 0.3533 \\
    $\mu[\cost]$ & 354.5 \\
    $\sigma[\cost]$ & 224.1 \\
    $\lambda$ & 0.0527 \\
    $\gamma$ & 0.0827 \\
    $\delta$ & 0.0460 \\
    $\kappa$ & 0.0006 \\
    $\sigma$ & 0.2893 \\
    $f_0$ & 10.024 \\
    $z_i$ & 1.7039 \\
    $\bar{z}^L_1$ & 0.8440 \\
    $\bar{z}^L_2$ & 0.9028 \\
    $\bar{z}^L_3$ & 0.9996 \\
    $\sigma_a$ & 0.1202 \\
    $\sigma_o$ & 0.2116 \\
    \bottomrule
    \end{tabular}
    \caption{Estimates of model parameters.}
    \label{tab:parameterestimates}
\end{table}

\section{A Historical Counterfactual: No May Reopening}

\begin{figure}[t]
    \centering
    \includegraphics[width=\textwidth]{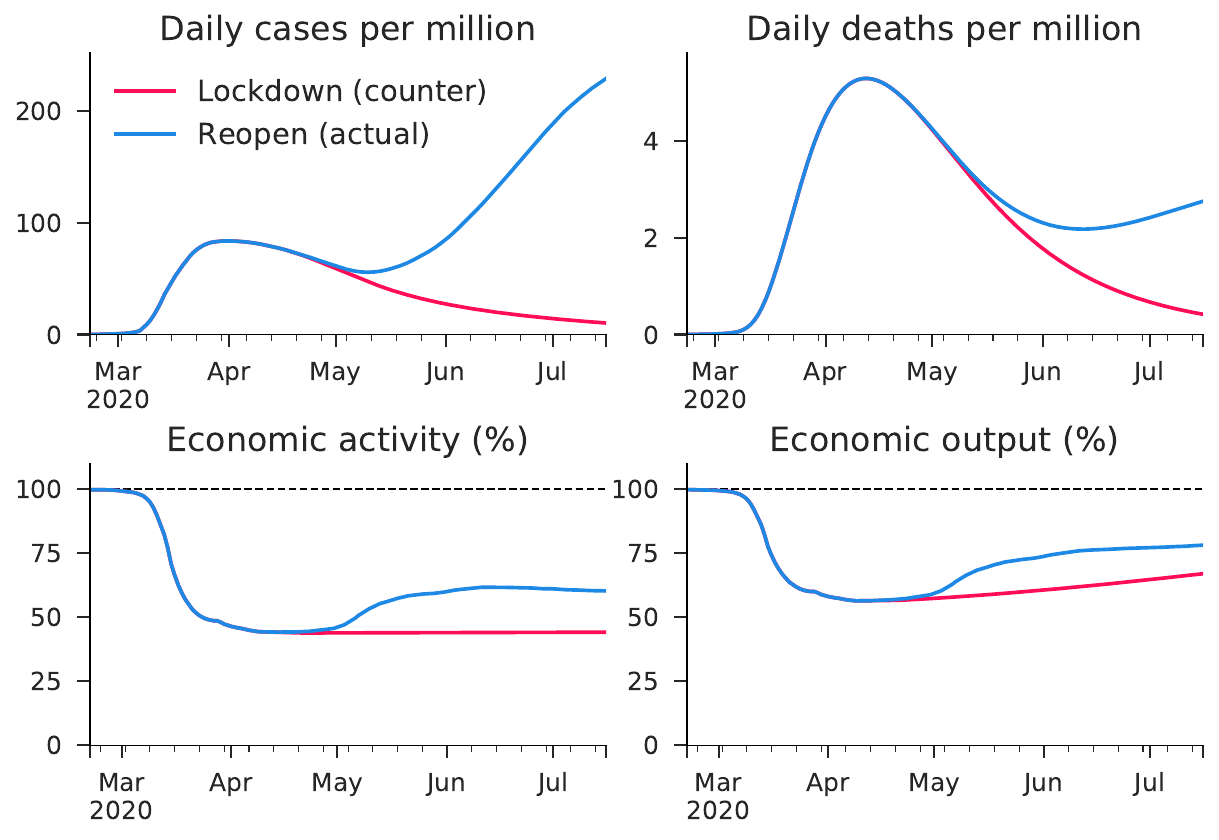}
    \caption{\textbf{No May Reopening.} Counterfactual path in which lockdown continued indefinitely versus actual path in which reopening occurred.}
    \label{fig:repopen_vs_lockdown}
\end{figure}

With the estimated model in hand, one can straightforwardly consider a variety of counterfactuals.
In \autoref{fig:repopen_vs_lockdown} we consider the possibility of continued lockdowns rather the reopenings implemented around the country in May.  


\newpage
\clearpage

\subsection{The Roles of Mitigation and Activity Reduction}
\label{sec:disturbingthetrends}

\begin{figure}[t]
    \centering
    \includegraphics[width=\textwidth]{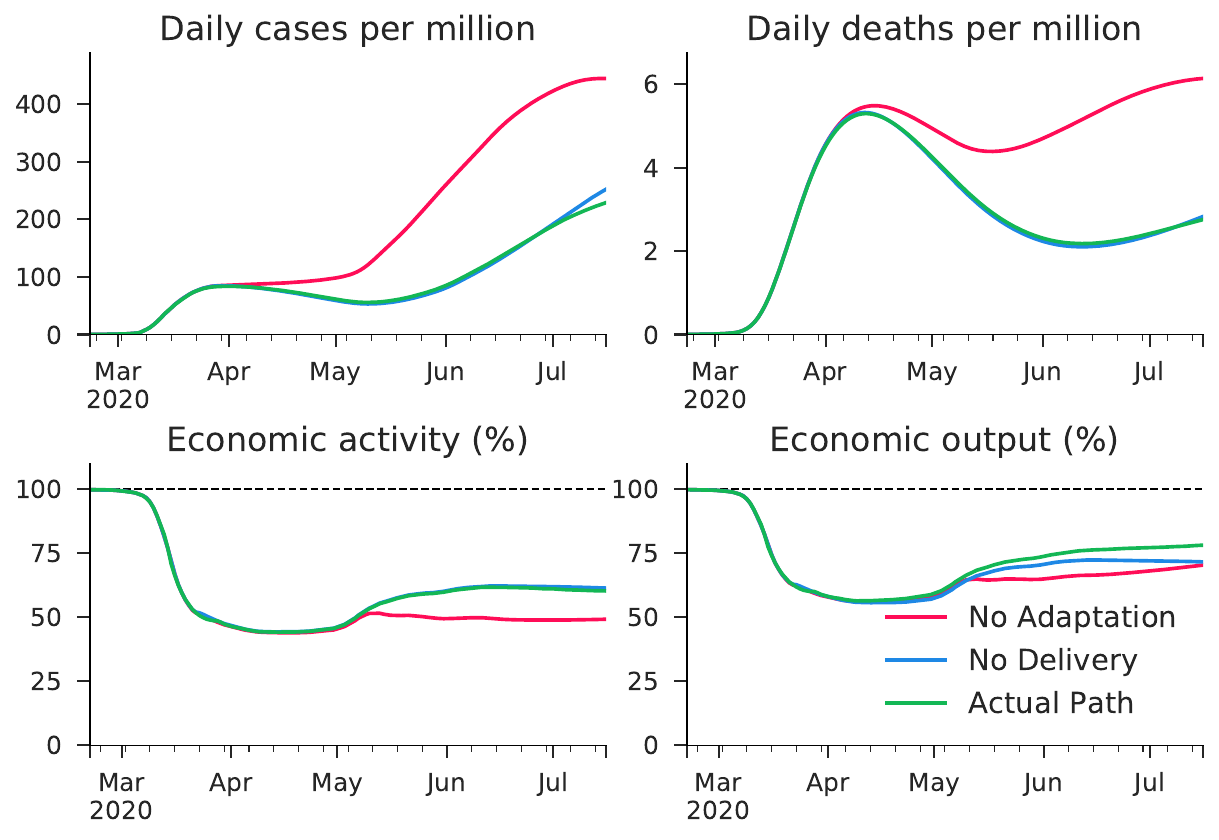}
    \caption{\textbf{Decomposing the Role of Technological Trends.}}
    \label{fig:trenddecomp}
\end{figure}

The estimated model allows for the possibility that agents have taken two types of precautions over time. 
First, agents have adopted measures to reduce virus transmission per unit of activity, modeled via declining \(\infratexc_t\).
The most salient real world example of such \(\infratexc\)-reducing measures is perhaps the wearing of face masks, but this mechanism in the model should be interpreted as  also encompassing myriad other measures, such as the installation of plexiglass barriers at grocery stores, shifting restaurant dining outdoors, and simple six foot social distancing. 
Second, agents have increased their use of services that allow consumption with less infection risk, namely delivery options for a variety of goods, modeled via the \(z\)-substitute option. Here we consider the counterfactual in which these measures had not taken place.\footnote{This 
    is straightforward to implement by setting \(\infratexc_t = \infratexc_0\) for all \(t\).
}
\autoref{fig:trenddecomp} shows the model's predicted evolution of the main variables in the absence of these trends.

\newpage
\clearpage
\begin{center}
\Huge{ONLINE APPENDIX}
    \setcounter{page}{1}
\end{center}

\section{Additional Figures}










\begin{figure}[h!]
    \centering
    \includegraphics[width=\textwidth]{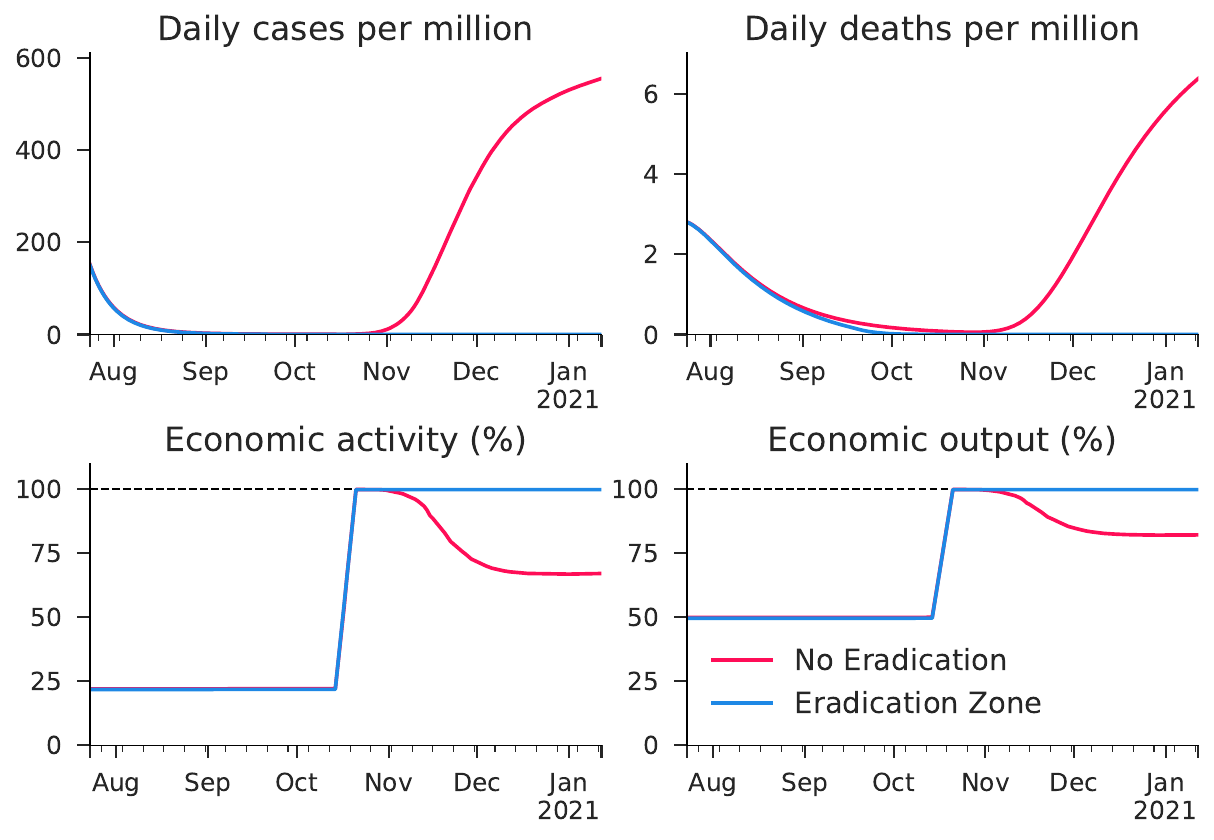}
    \caption{Strong temporary lockdown with and without eradication zone.}
    \label{fig:templockdowns}
\end{figure}



\begin{figure}[th!]
    \centering
    \includegraphics[width=\textwidth]{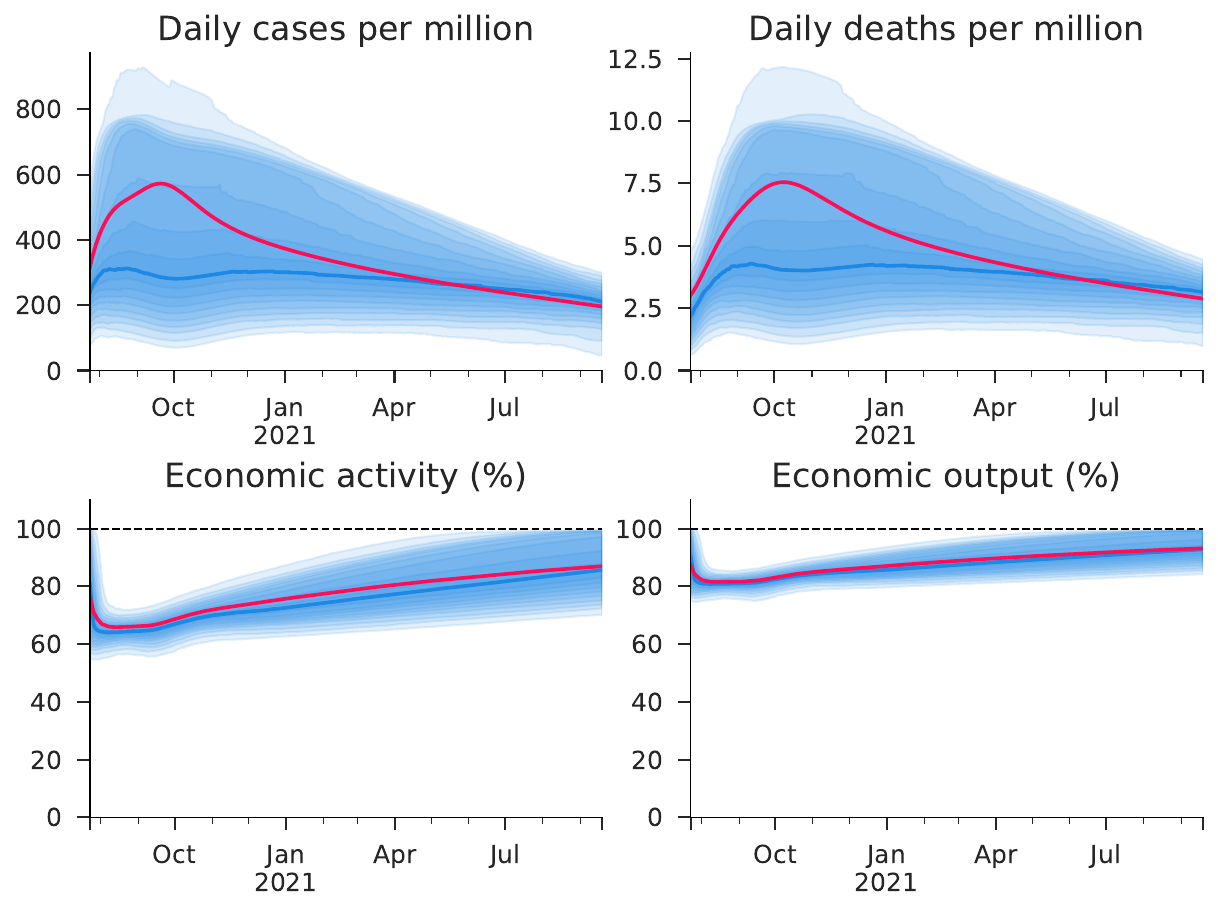}
    \caption{\textbf{Laissez-Faire until herd immunity---One year horizon.} Red lines are population-weighted means, thick blue lines are for the median county, and the blue shaded bands indicate percentiles of the cross sectional distribution of counties.}
    \label{fig:laissezfaire_short}
\end{figure}

\begin{figure}
    \centering
    \includegraphics[width=\textwidth]{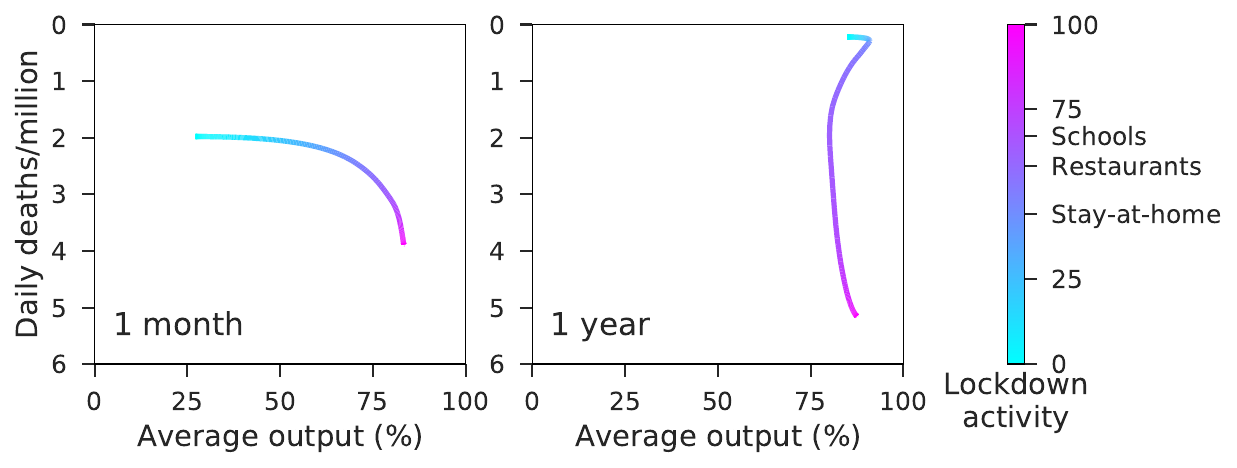}
    \caption{Forward looking PPF under adaptive lockdown policy: weak eradication zone.}
    \label{fig:ppfs_weakkz}
\end{figure}

\begin{figure}
    \centering
    \includegraphics[width=\textwidth]{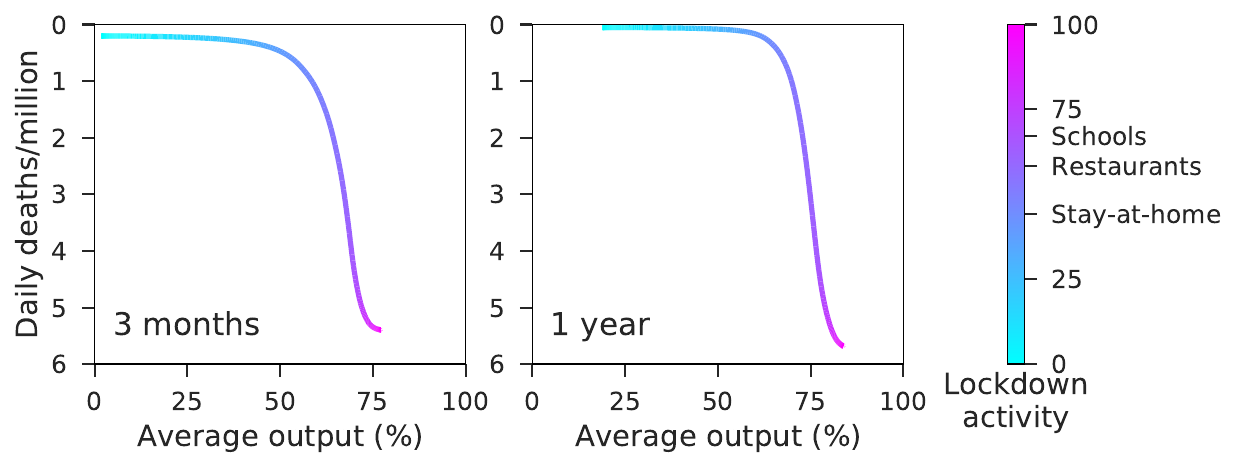}
    \includegraphics[width=\textwidth]{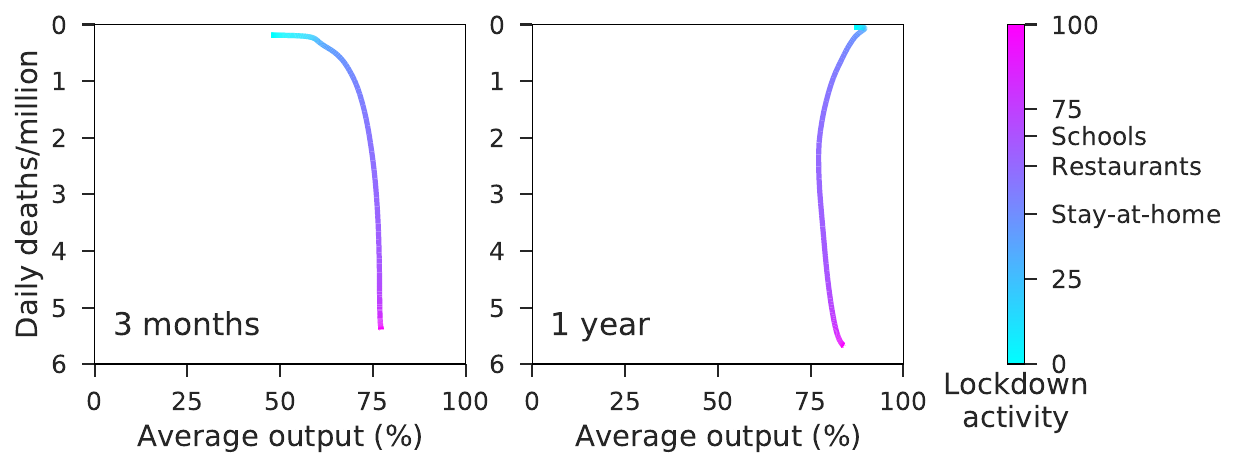}
    \includegraphics[width=\textwidth]{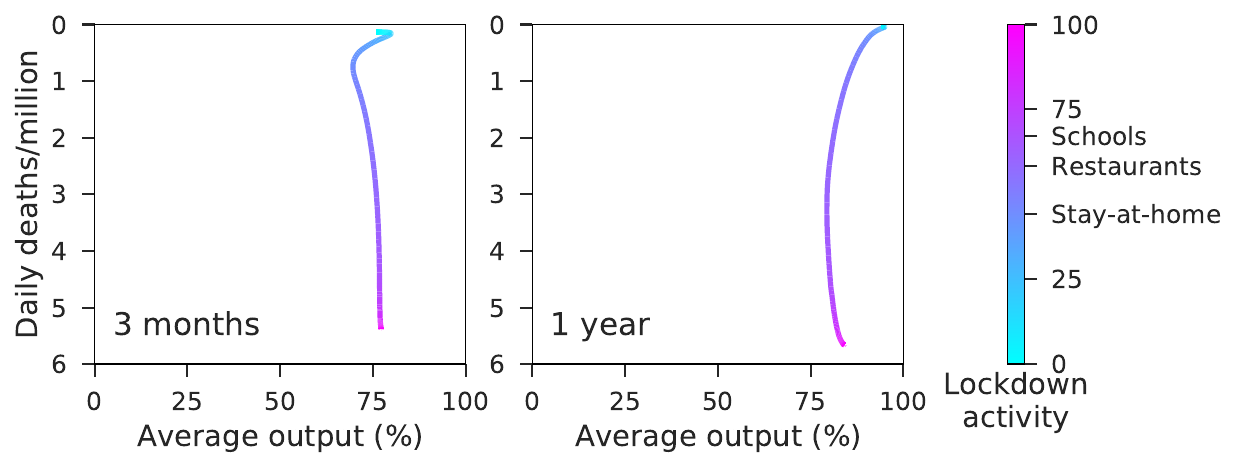}
    \caption{Counterfactual PPF under adaptive lockdown policy: (a) without eradication, (b) with weak eradication zone, and (c) with strong eradication zone.}
    \label{fig:ppfs_counter}
\end{figure}





\clearpage
\newpage






\section{State Fit}
\label{sec:statefit}

\begin{figure}[ht!]
    \captionsetup[subfigure]{labelformat=empty}
    \centering
    \begin{subfigure}{.48\textwidth}
      \centering
      \caption{Alabama}
      \includegraphics[width=1.0\textwidth]{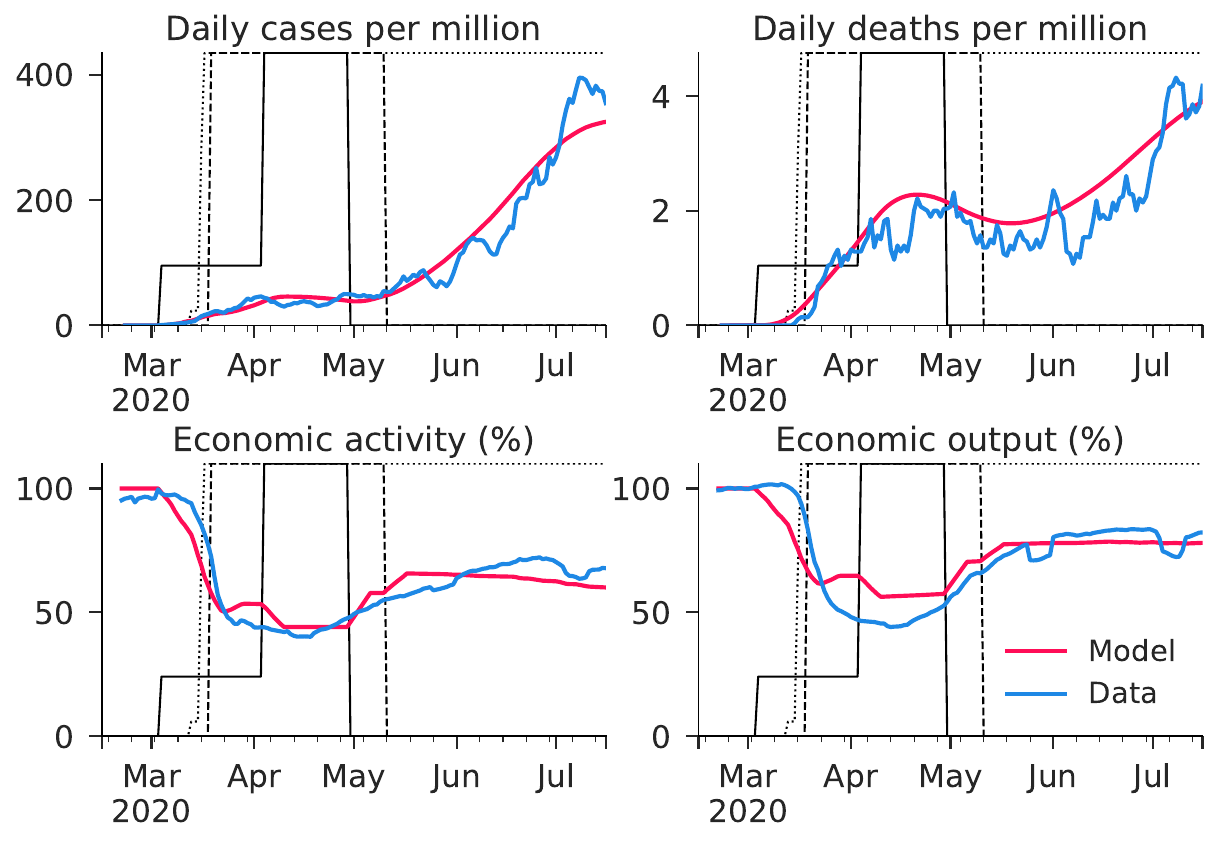}
    \end{subfigure}
    ~~
    \begin{subfigure}{.48\textwidth}
      \centering
      \caption{Alaska}
      \includegraphics[width=1.0\textwidth]{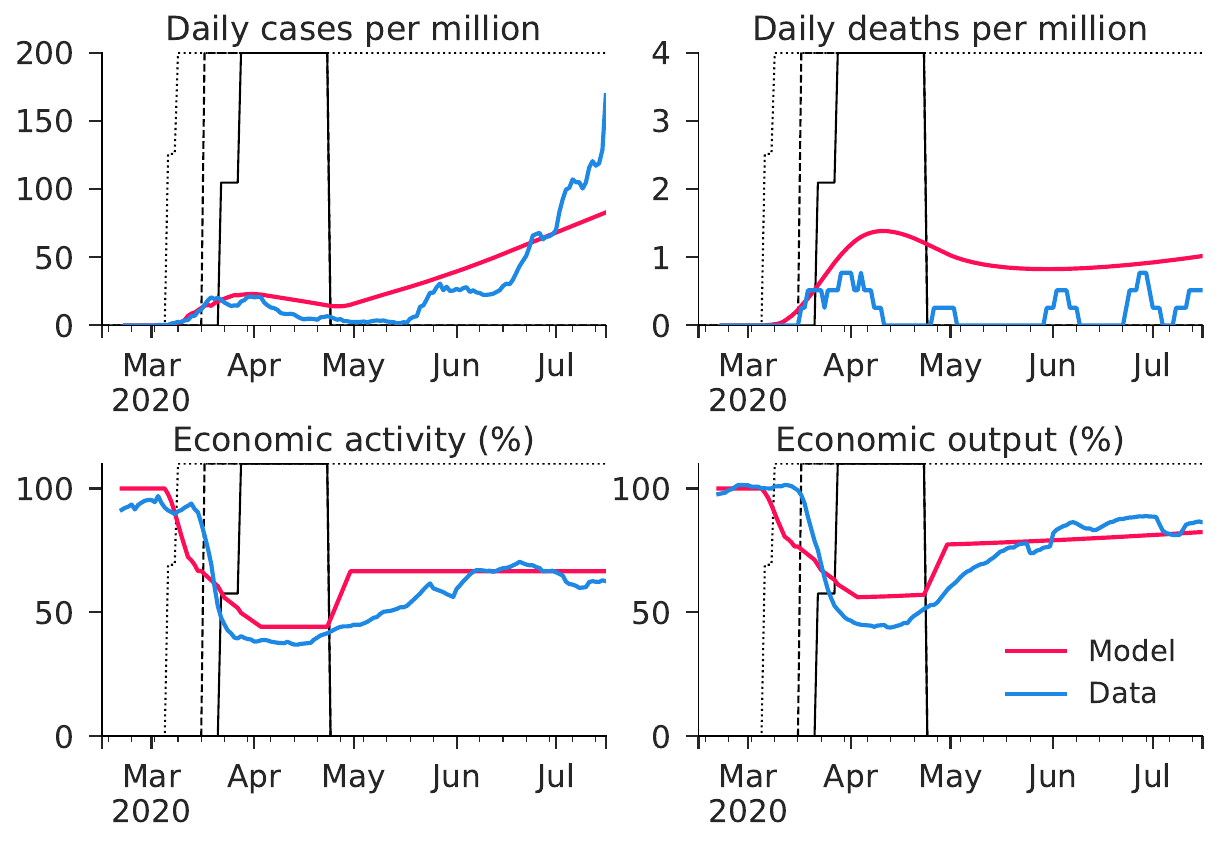}
    \end{subfigure}
    
    \bigskip

    \begin{subfigure}{.48\textwidth}
      \captionsetup[subfigure]{labelformat=empty}
      \centering
      \caption{Arizona}
      \includegraphics[width=1.0\textwidth]{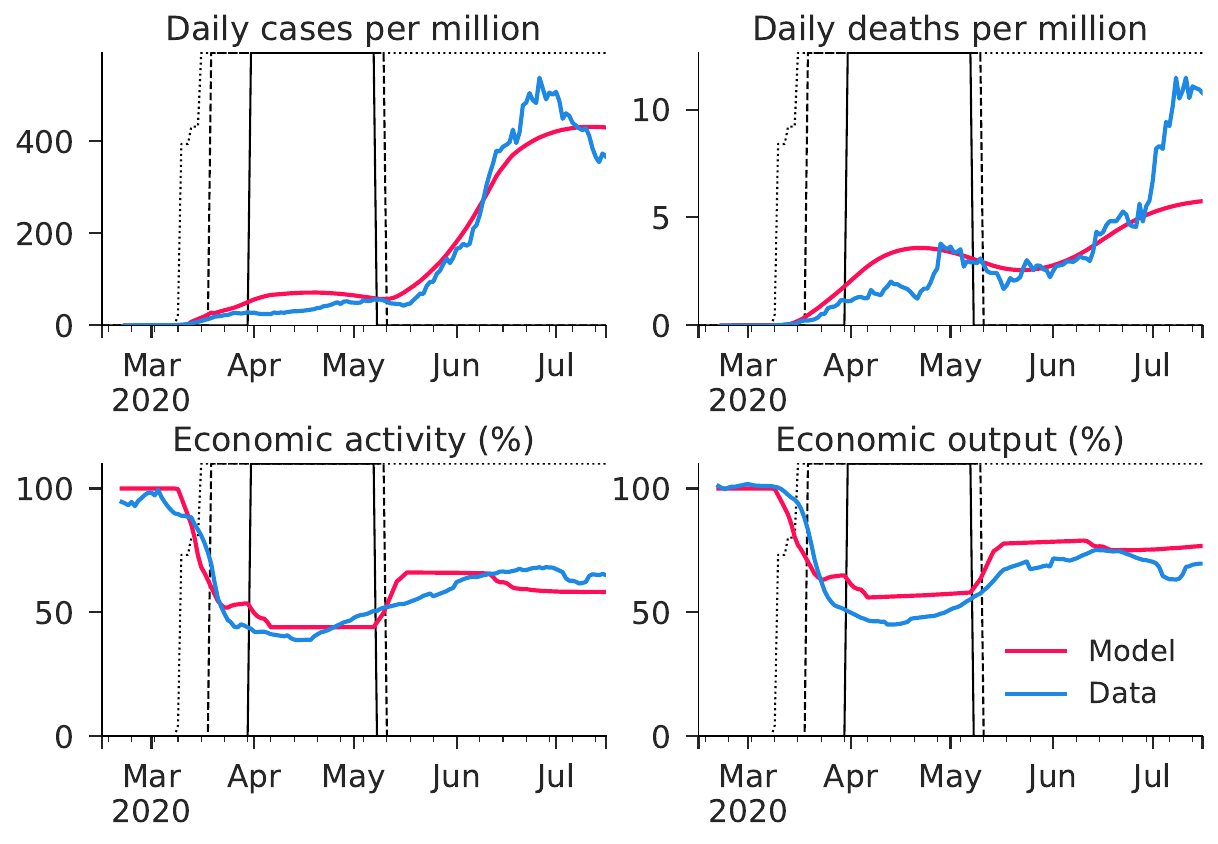}
    \end{subfigure}
    ~~
    \begin{subfigure}{.48\textwidth}
      \centering
      \caption{Arkansas}
      \includegraphics[width=1.0\textwidth]{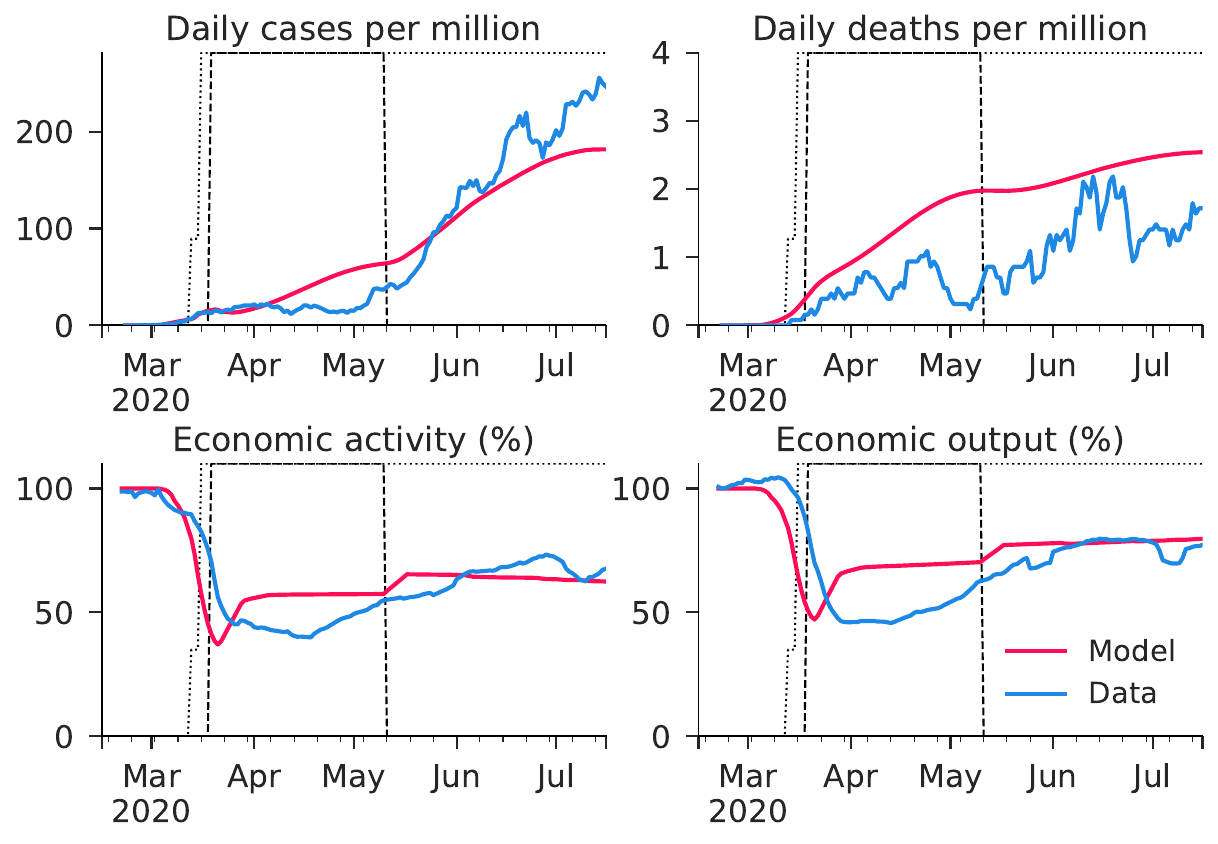}
    \end{subfigure}
    
    \bigskip

    \begin{subfigure}{.48\textwidth}
      \centering
      \caption{California}
      \includegraphics[width=1.0\textwidth]{figures/states/California.pdf}
    \end{subfigure}
    ~~
    \begin{subfigure}{.48\textwidth}
      \centering
      \caption{Colorado}
      \includegraphics[width=1.0\textwidth]{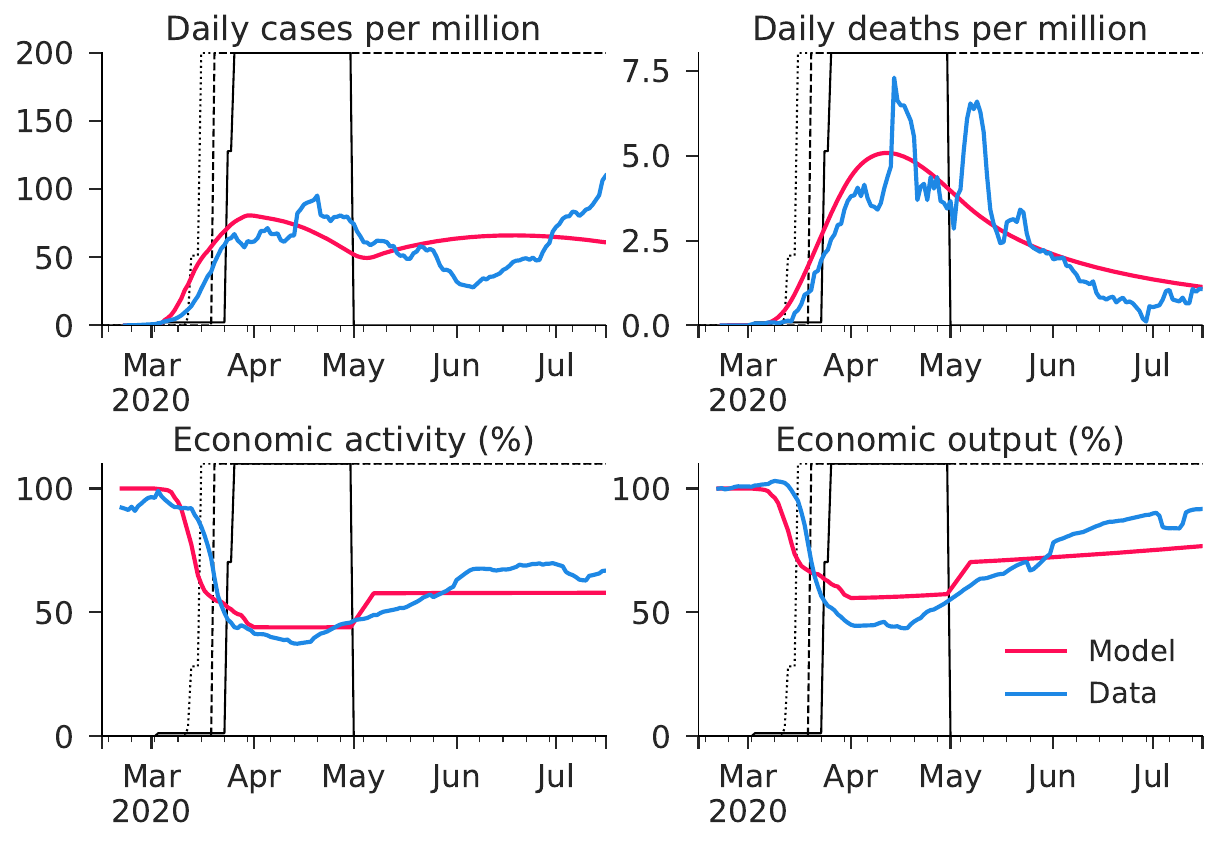}
    \end{subfigure}
\end{figure}

\begin{figure}[t]
    \captionsetup[subfigure]{labelformat=empty}
    \centering
    \begin{subfigure}{.48\textwidth}
      \centering
      \caption{Delaware}
      \includegraphics[width=1.0\textwidth]{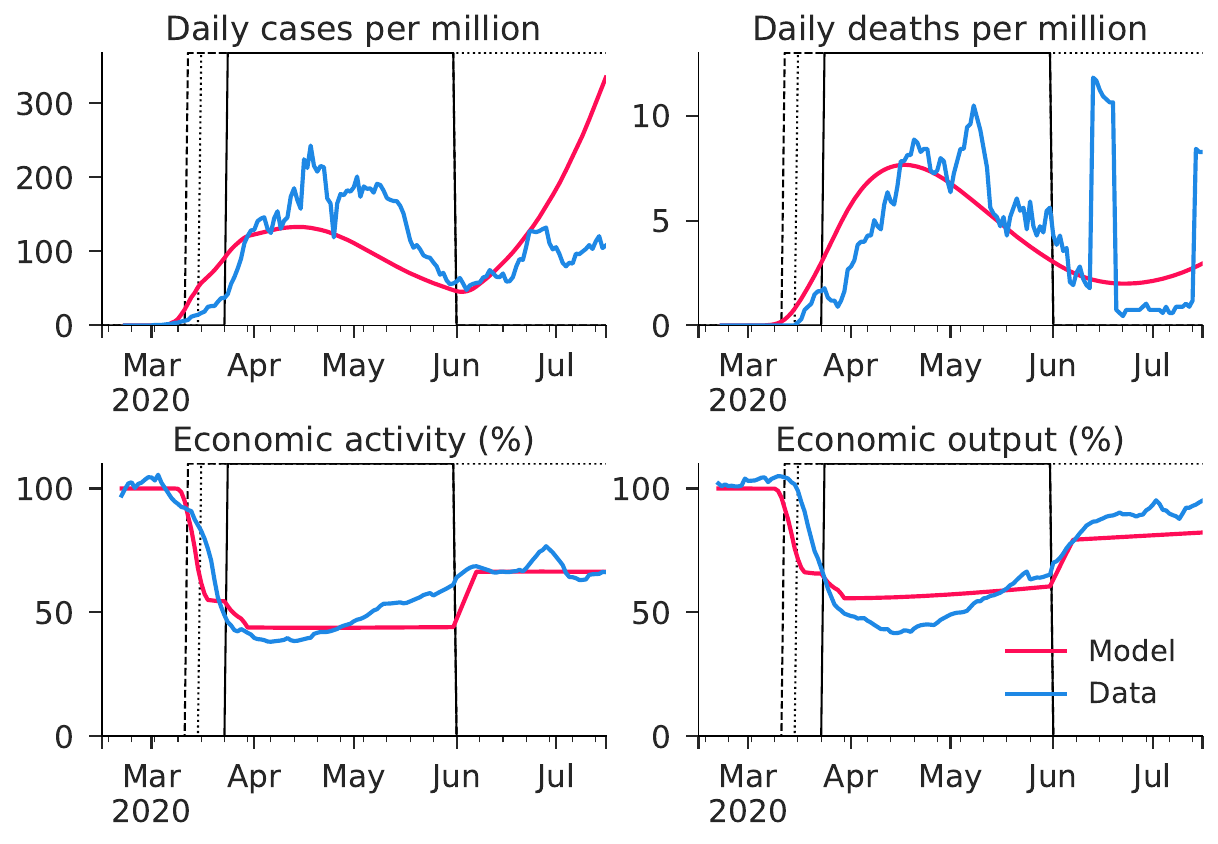}
    \end{subfigure}
    ~~
    \begin{subfigure}{.48\textwidth}
      \centering
      \caption{District of Columbia}
      \includegraphics[width=1.0\textwidth]{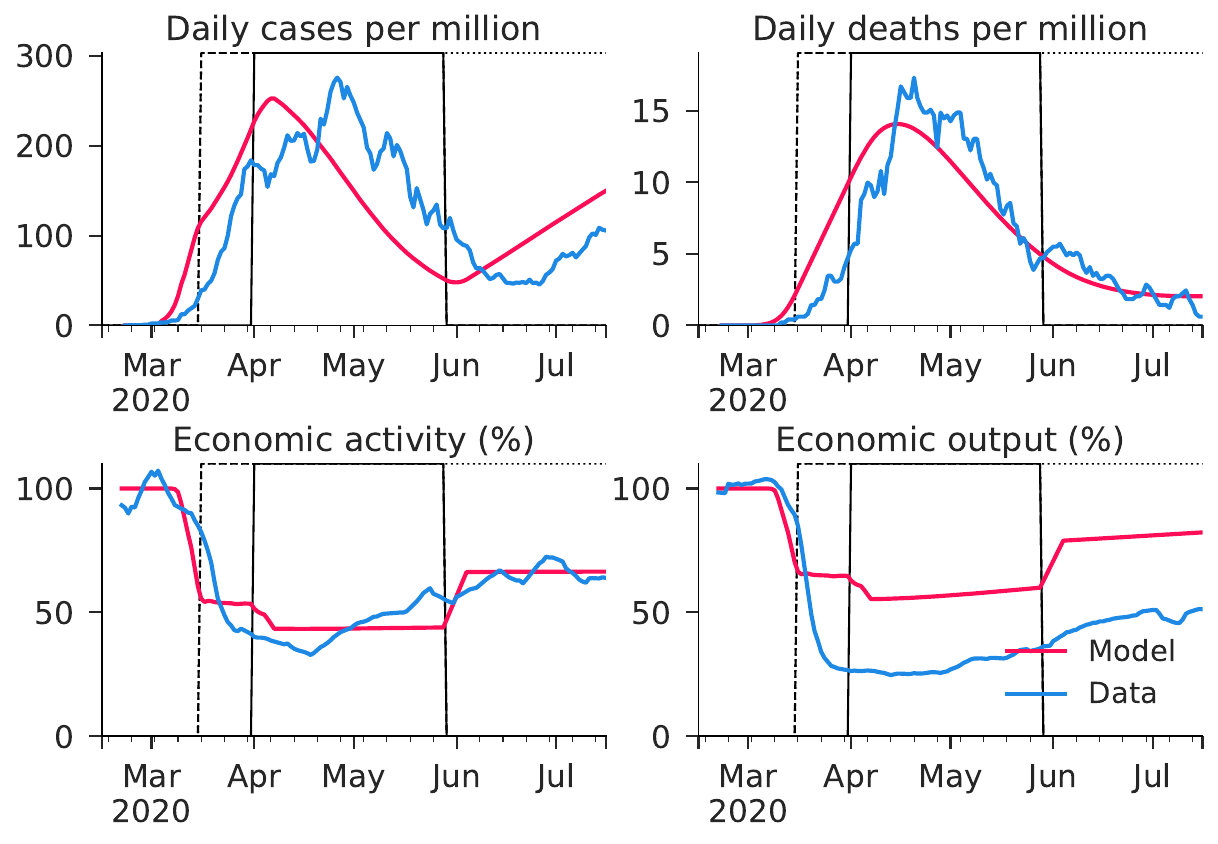}
    \end{subfigure}

    \bigskip

    \begin{subfigure}{.48\textwidth}
      \centering
      \caption{Florida}
      \includegraphics[width=1.0\textwidth]{figures/states/Florida.pdf}
    \end{subfigure}
    ~~
    \begin{subfigure}{.48\textwidth}
      \centering
      \caption{Georgia}
      \includegraphics[width=1.0\textwidth]{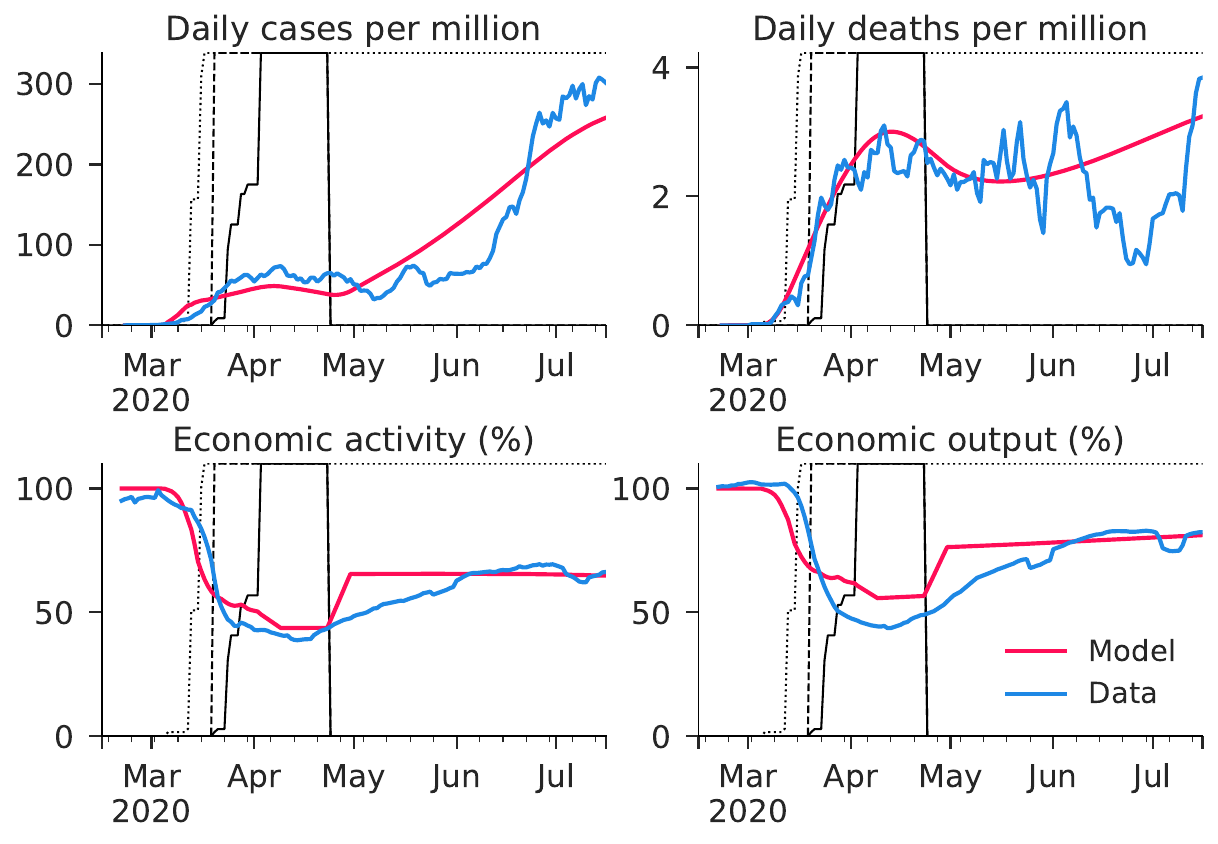}
    \end{subfigure}
    
    \bigskip
    
    \begin{subfigure}{.48\textwidth}
      \centering
      \caption{Hawaii}
      \includegraphics[width=1.0\textwidth]{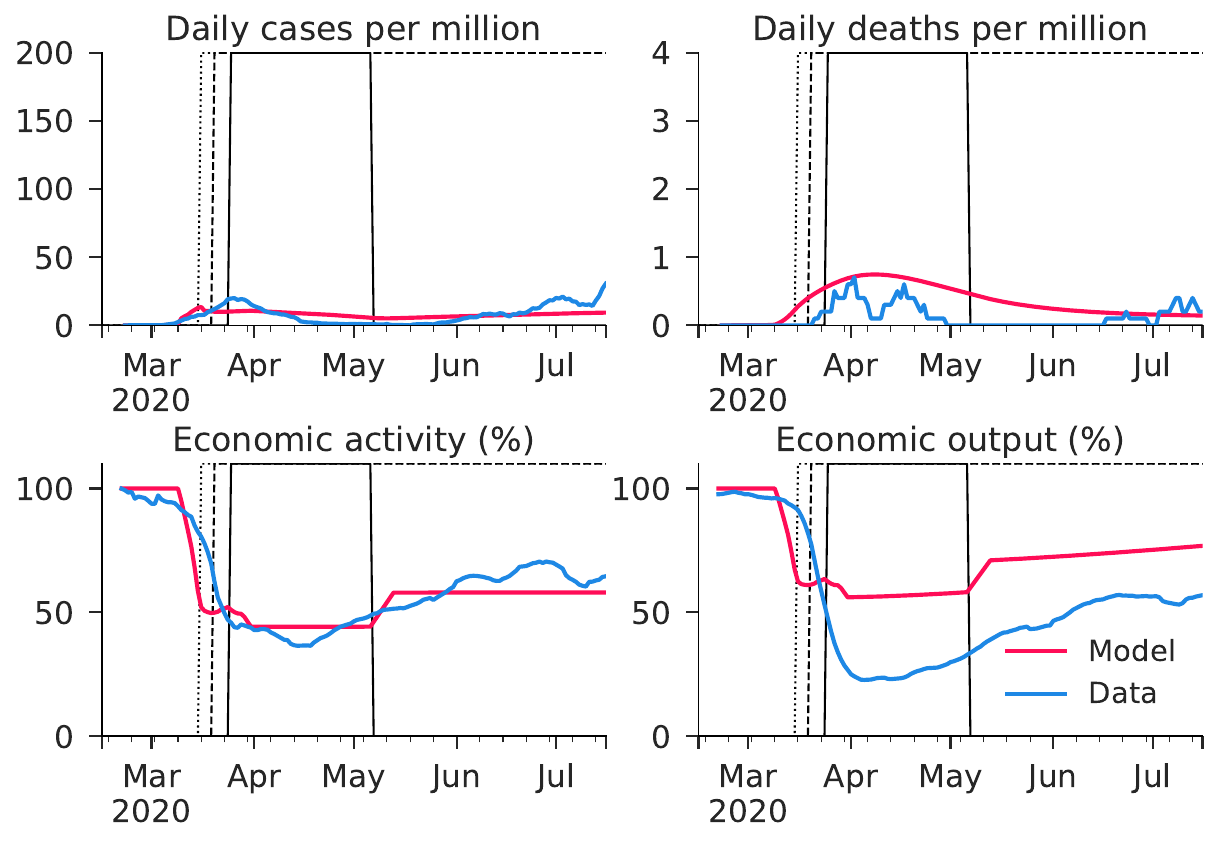}
    \end{subfigure}
    ~~
    \begin{subfigure}{.48\textwidth}
      \centering
      \caption{Idaho}
      \includegraphics[width=1.0\textwidth]{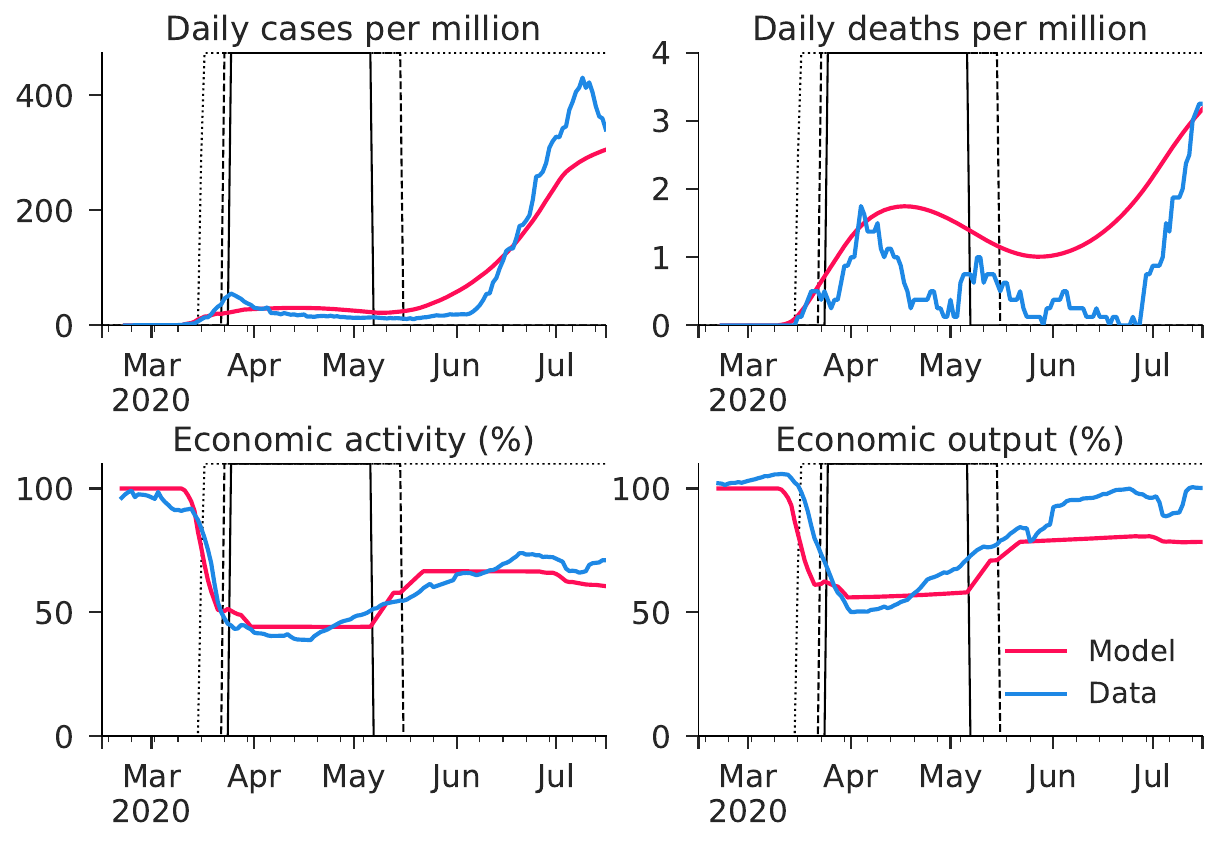}
    \end{subfigure}
\end{figure}
\begin{figure}[t]
    \captionsetup[subfigure]{labelformat=empty}
    \centering
    \begin{subfigure}{.48\textwidth}
      \centering
      \caption{Illinois}
      \includegraphics[width=1.0\textwidth]{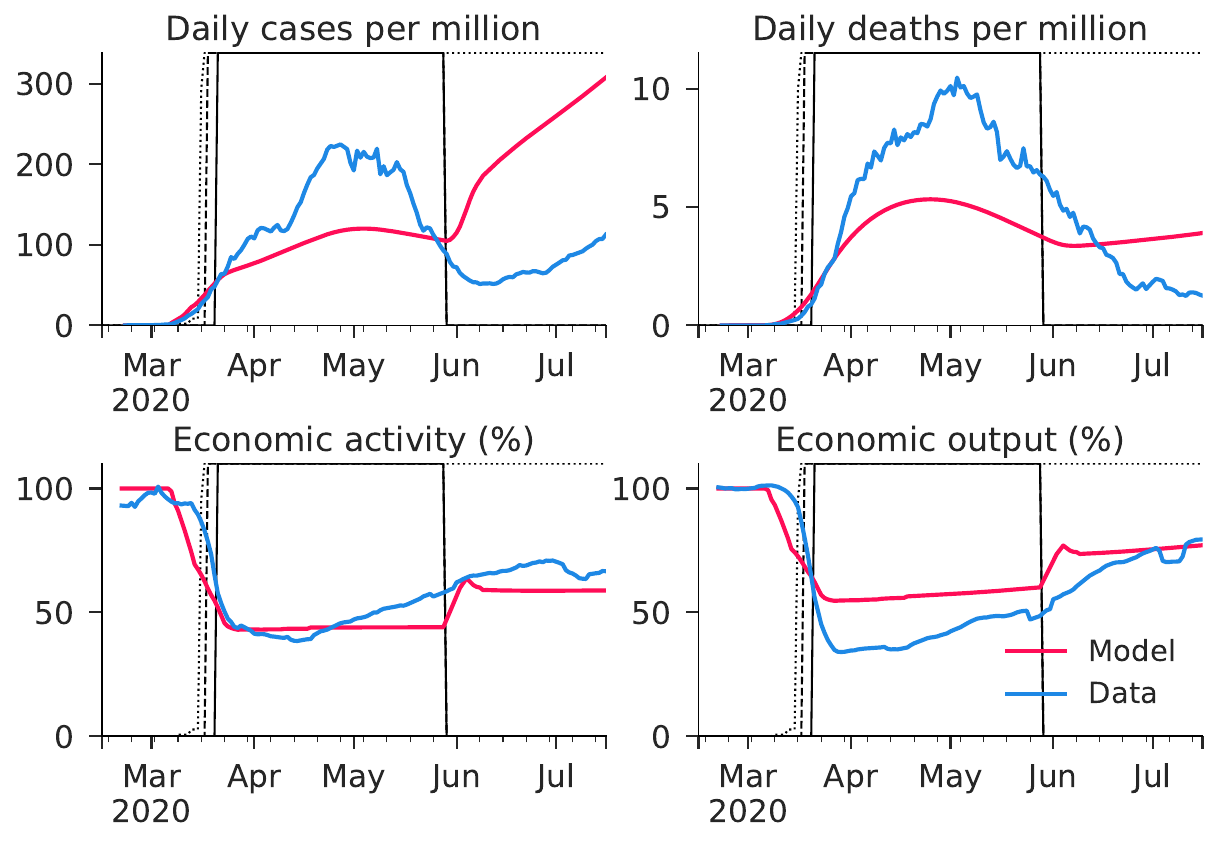}
    \end{subfigure}
    ~~
    \begin{subfigure}{.48\textwidth}
      \centering
      \caption{Indiana}
      \includegraphics[width=1.0\textwidth]{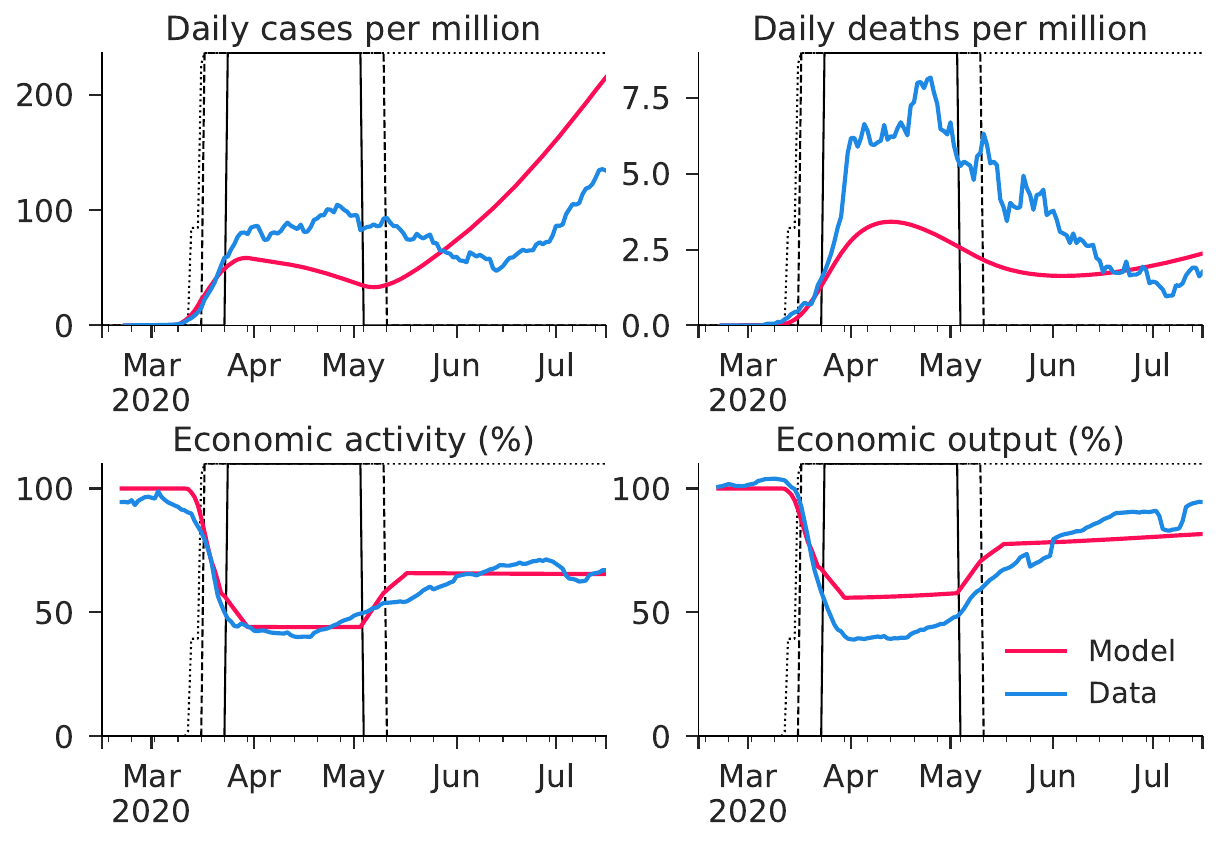}
    \end{subfigure}

    \bigskip

    \begin{subfigure}{.48\textwidth}
      \centering
      \caption{Iowa}
      \includegraphics[width=1.0\textwidth]{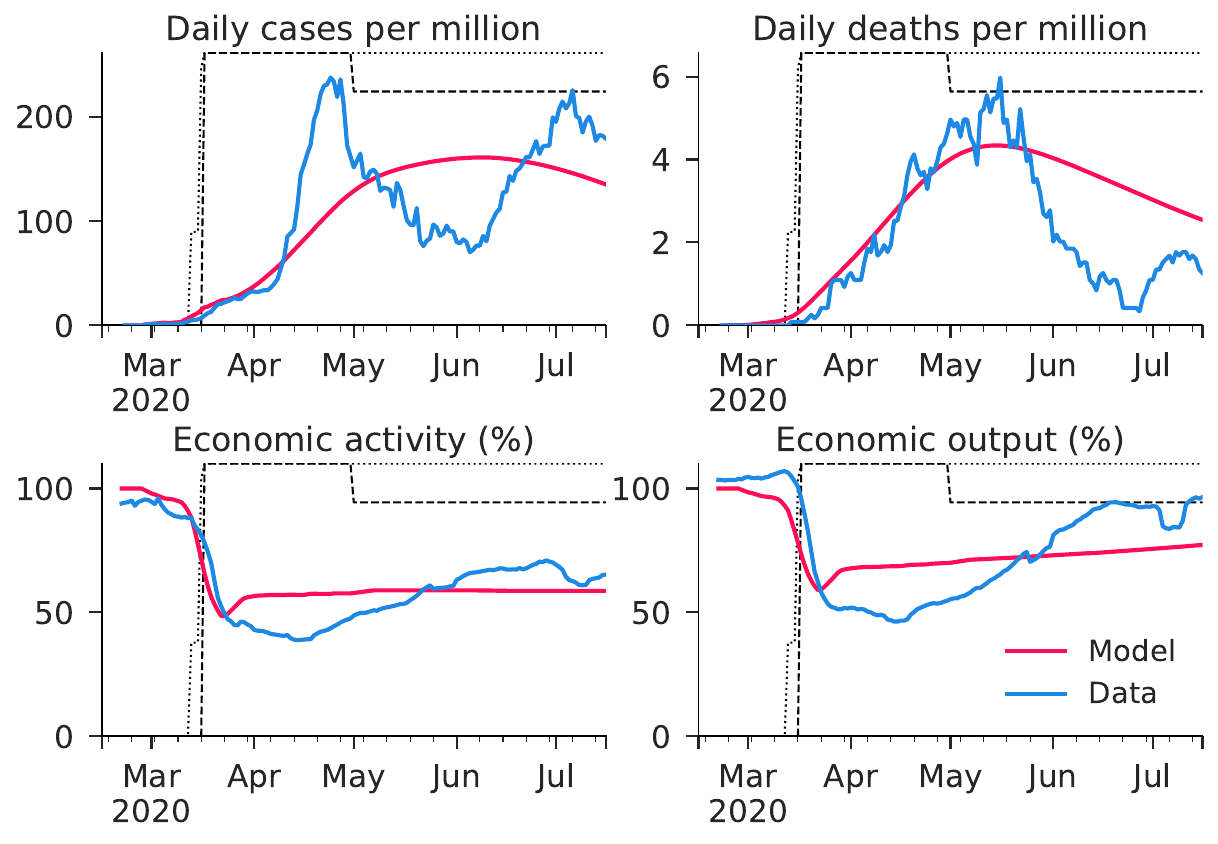}
    \end{subfigure}
    ~~
    \begin{subfigure}{.48\textwidth}
      \centering
      \caption{Kansas}
      \includegraphics[width=1.0\textwidth]{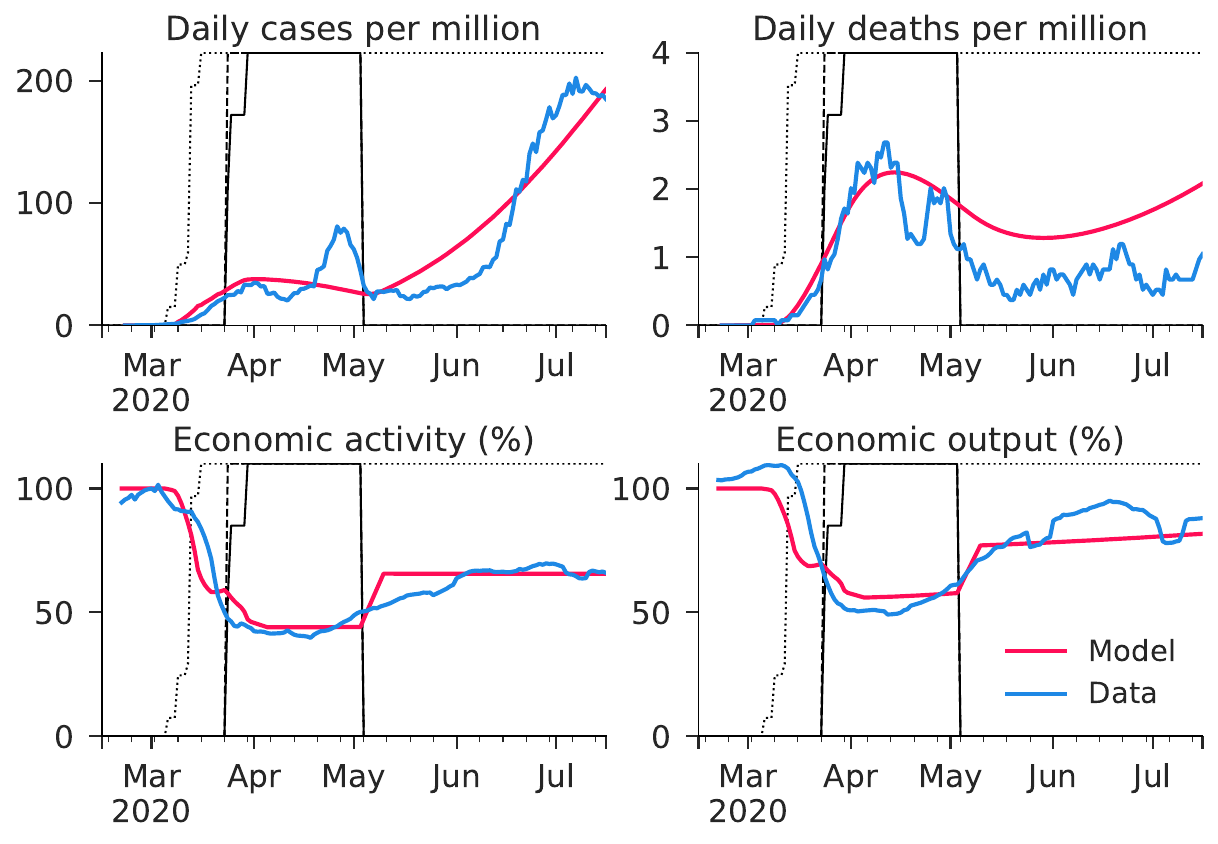}
    \end{subfigure}
    
    \bigskip
    
    \begin{subfigure}{.48\textwidth}
      \centering
      \caption{Kentucky}
      \includegraphics[width=1.0\textwidth]{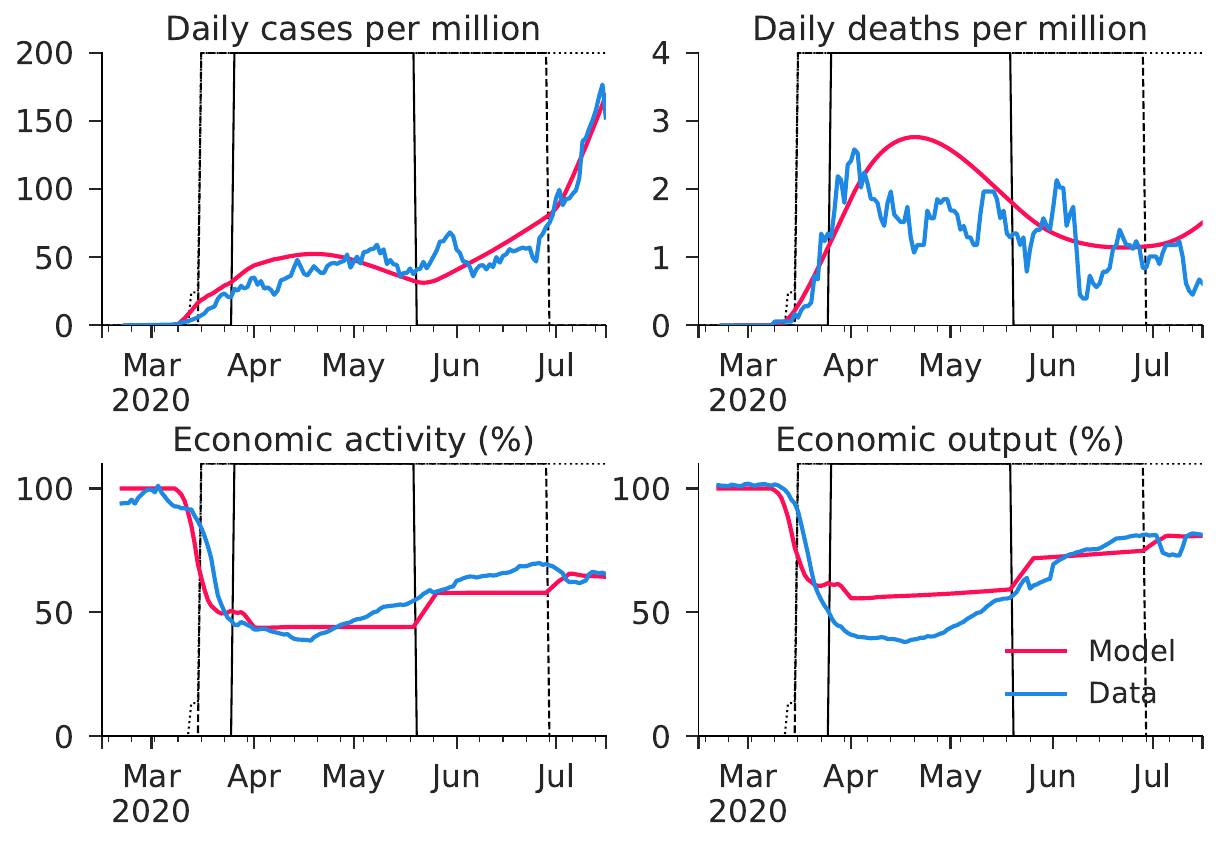}
    \end{subfigure}
    ~~
    \begin{subfigure}{.48\textwidth}
      \centering
      \caption{Louisiana}
      \includegraphics[width=1.0\textwidth]{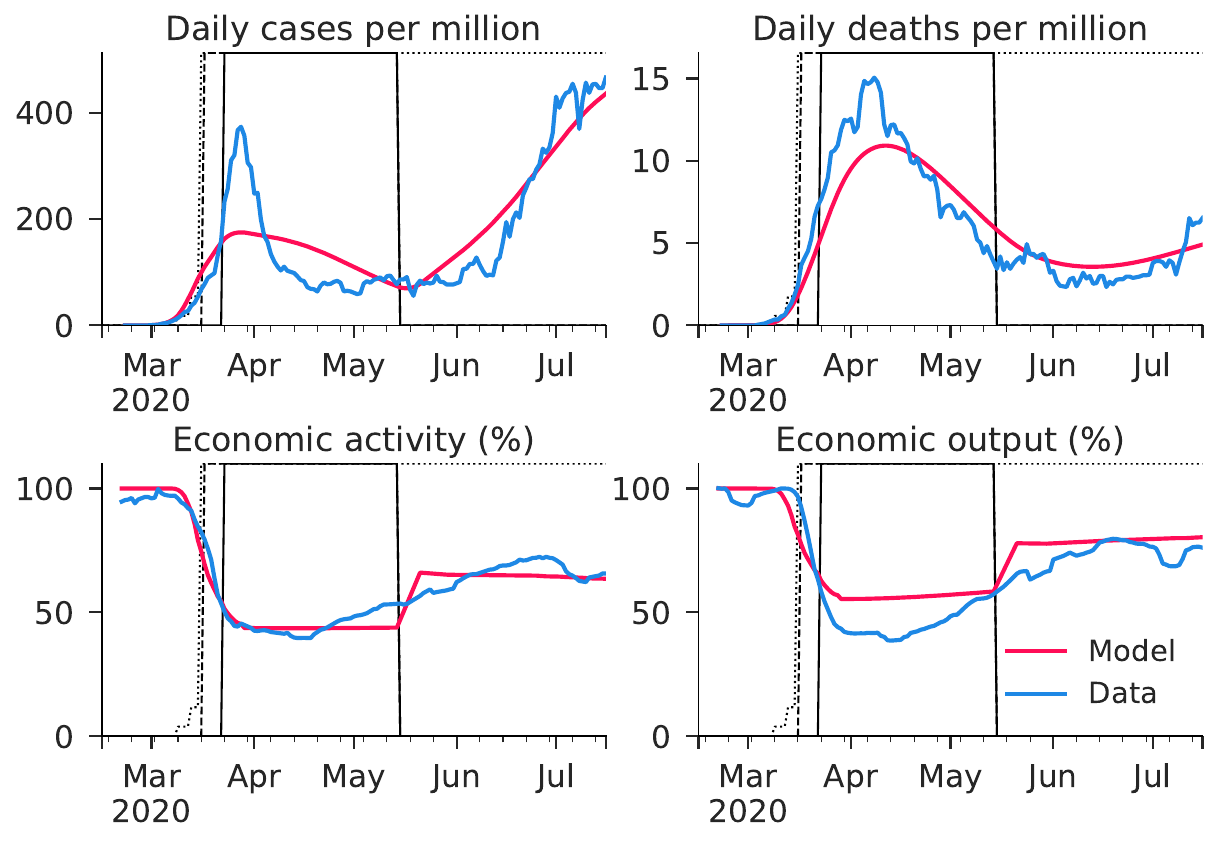}
    \end{subfigure}
\end{figure}

\begin{figure}[t]
    \captionsetup[subfigure]{labelformat=empty}
    \centering
    \begin{subfigure}{.48\textwidth}
      \centering
      \caption{Maryland}
      \includegraphics[width=1.0\textwidth]{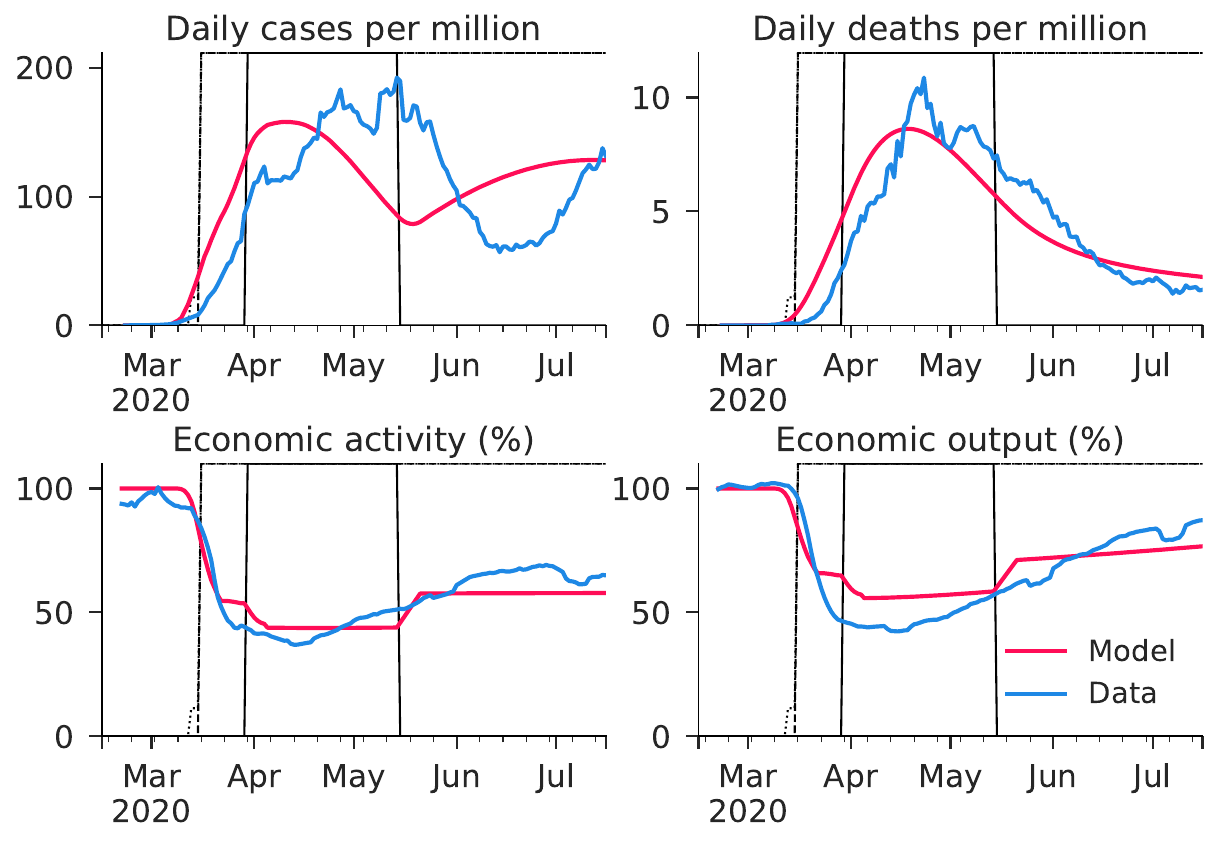}
    \end{subfigure}
    ~~
    \begin{subfigure}{.48\textwidth}
      \centering
      \caption{Michigan}
      \includegraphics[width=1.0\textwidth]{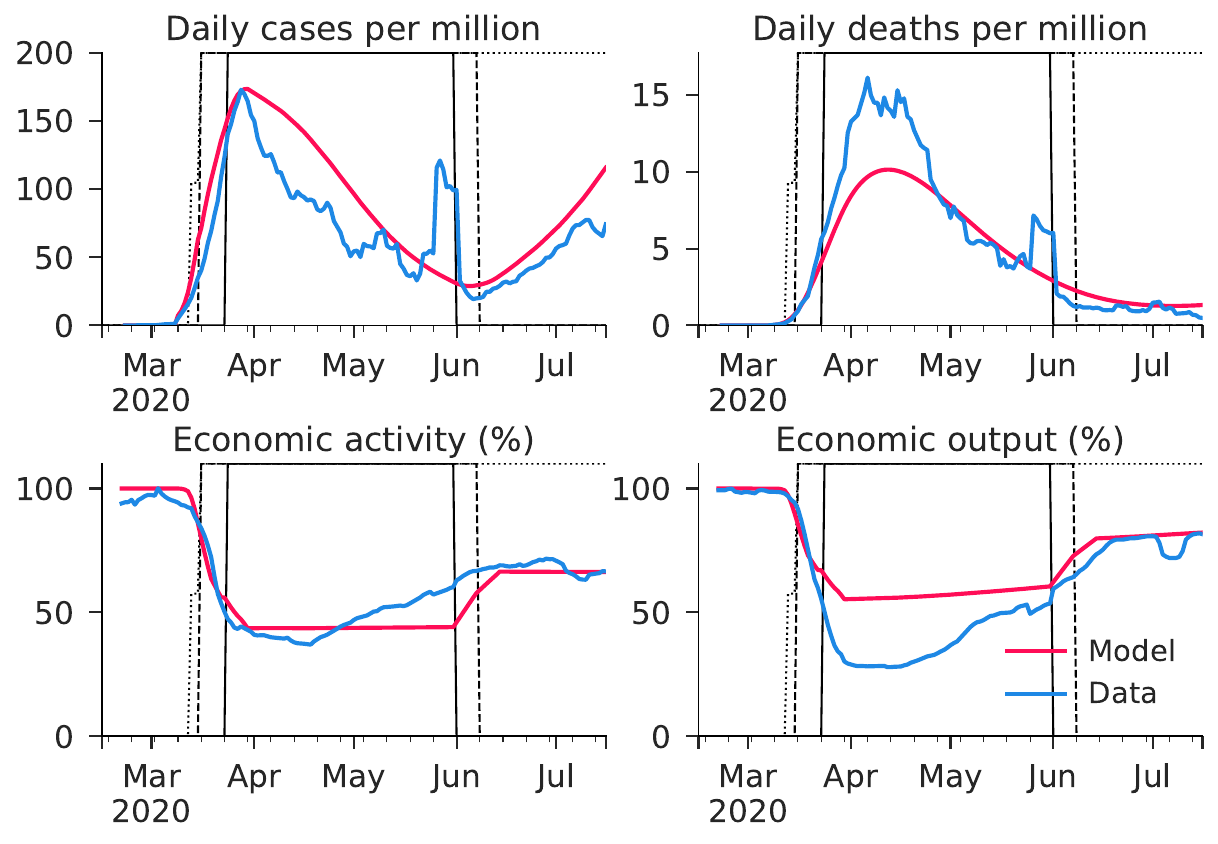}
    \end{subfigure}

    \bigskip

    \begin{subfigure}{.48\textwidth}
      \centering
      \caption{Minnesota}
      \includegraphics[width=1.0\textwidth]{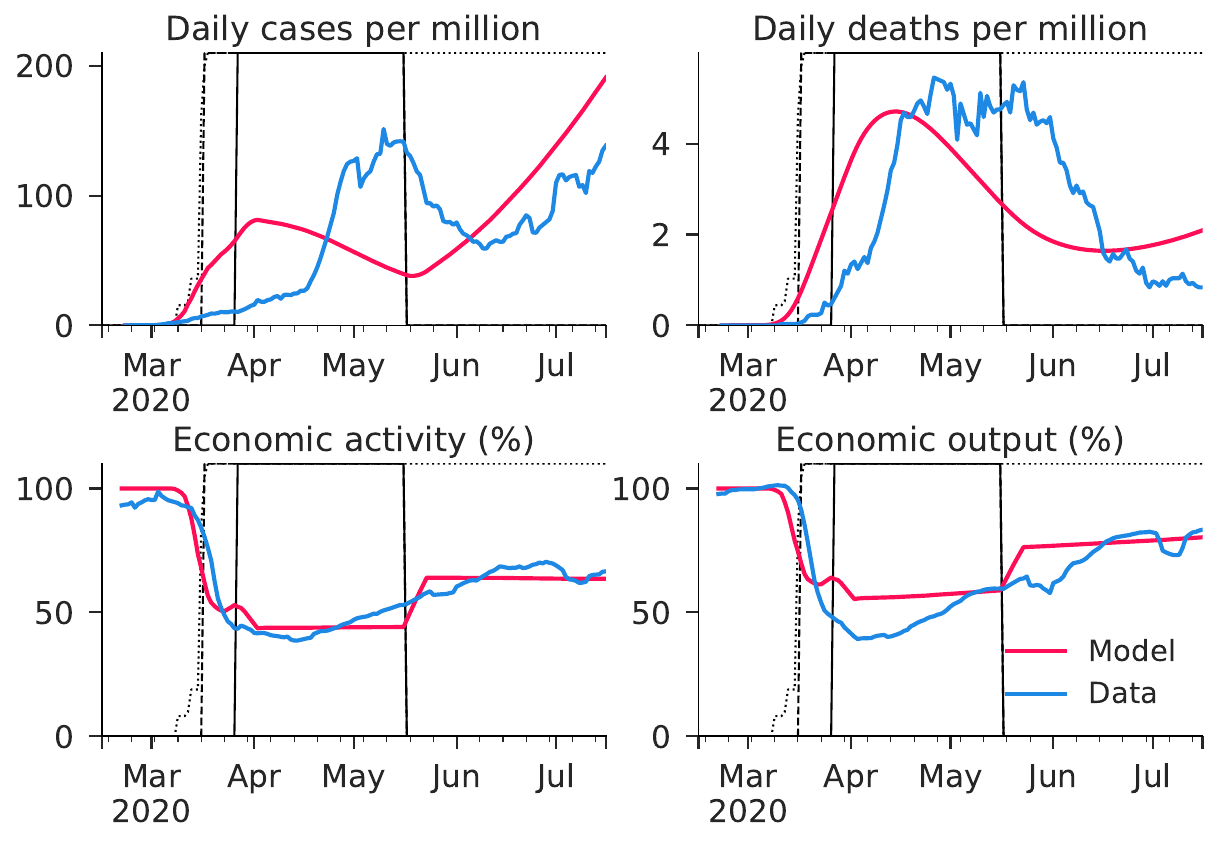}
    \end{subfigure}
    ~~
    \begin{subfigure}{.48\textwidth}
      \centering
      \caption{Mississippi}
      \includegraphics[width=1.0\textwidth]{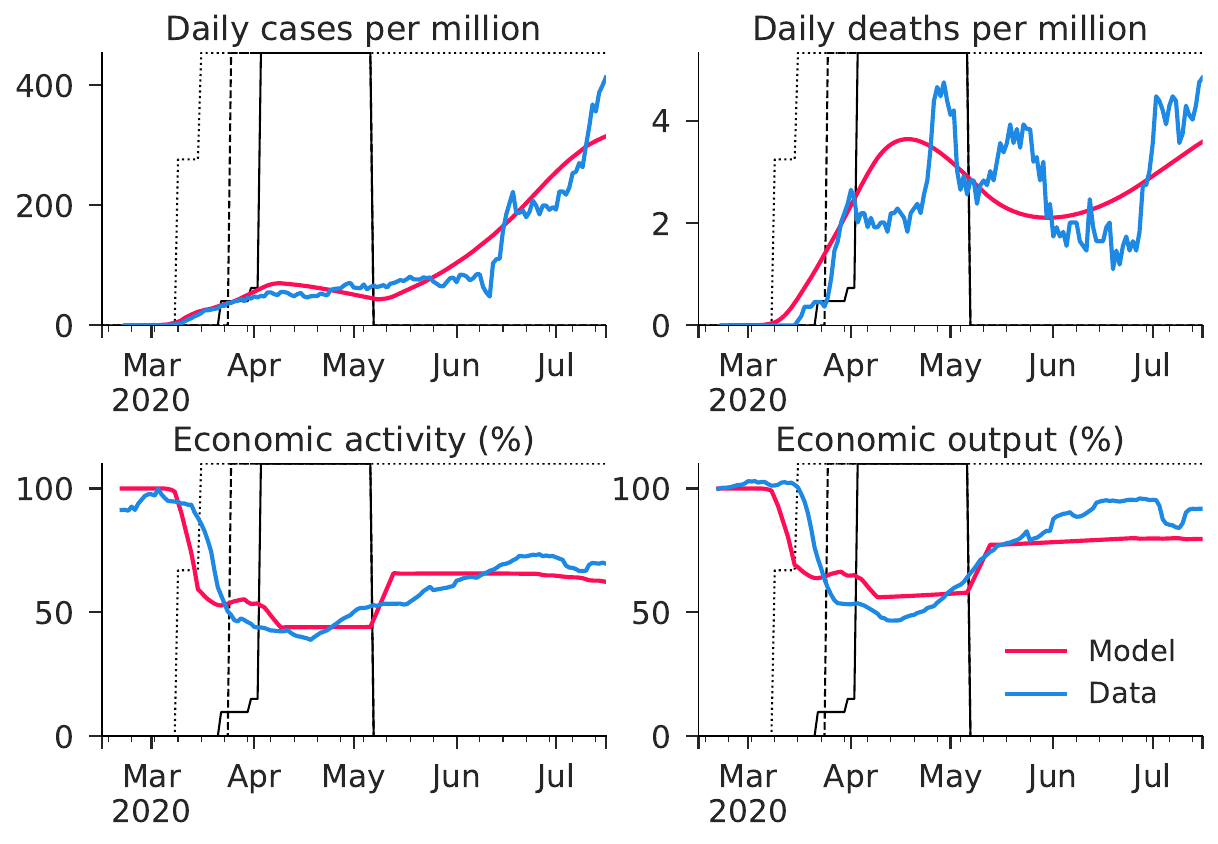}
    \end{subfigure}
    
    \bigskip
    
    \begin{subfigure}{.48\textwidth}
      \centering
      \caption{Missouri}
      \includegraphics[width=1.0\textwidth]{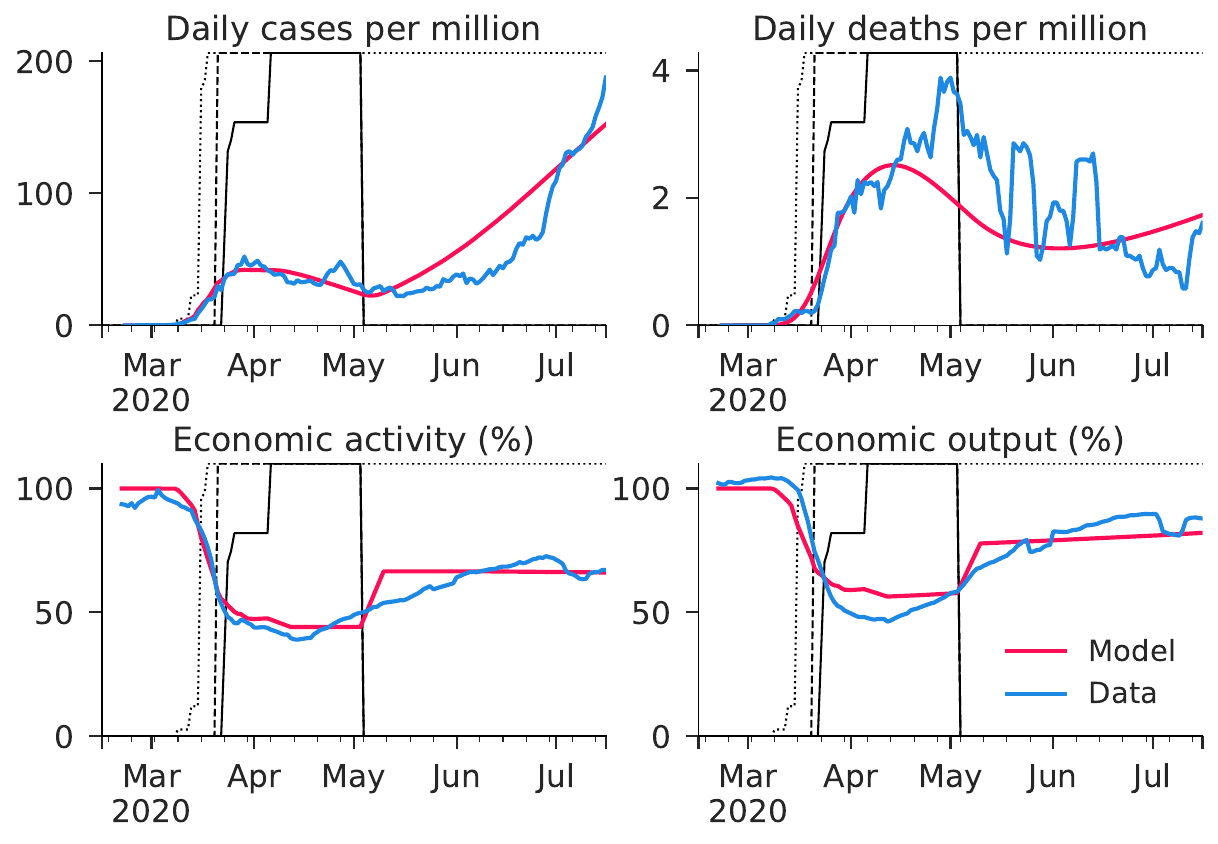}
    \end{subfigure}
    ~~
    \begin{subfigure}{.48\textwidth}
      \centering
      \caption{Montana}
      \includegraphics[width=1.0\textwidth]{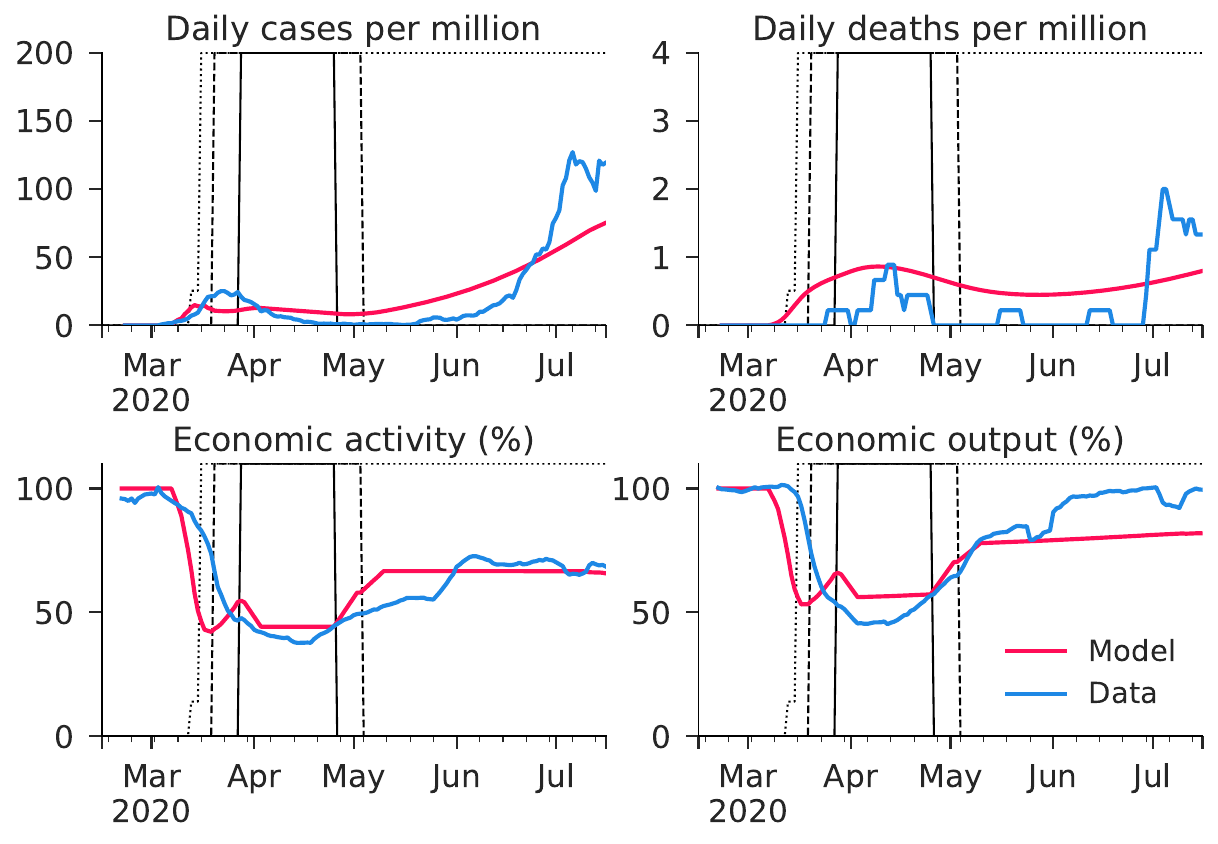}
    \end{subfigure}
\end{figure}

\begin{figure}[t]
    \captionsetup[subfigure]{labelformat=empty}
    \centering
    \begin{subfigure}{.48\textwidth}
      \centering
      \caption{Nebraska}
      \includegraphics[width=1.0\textwidth]{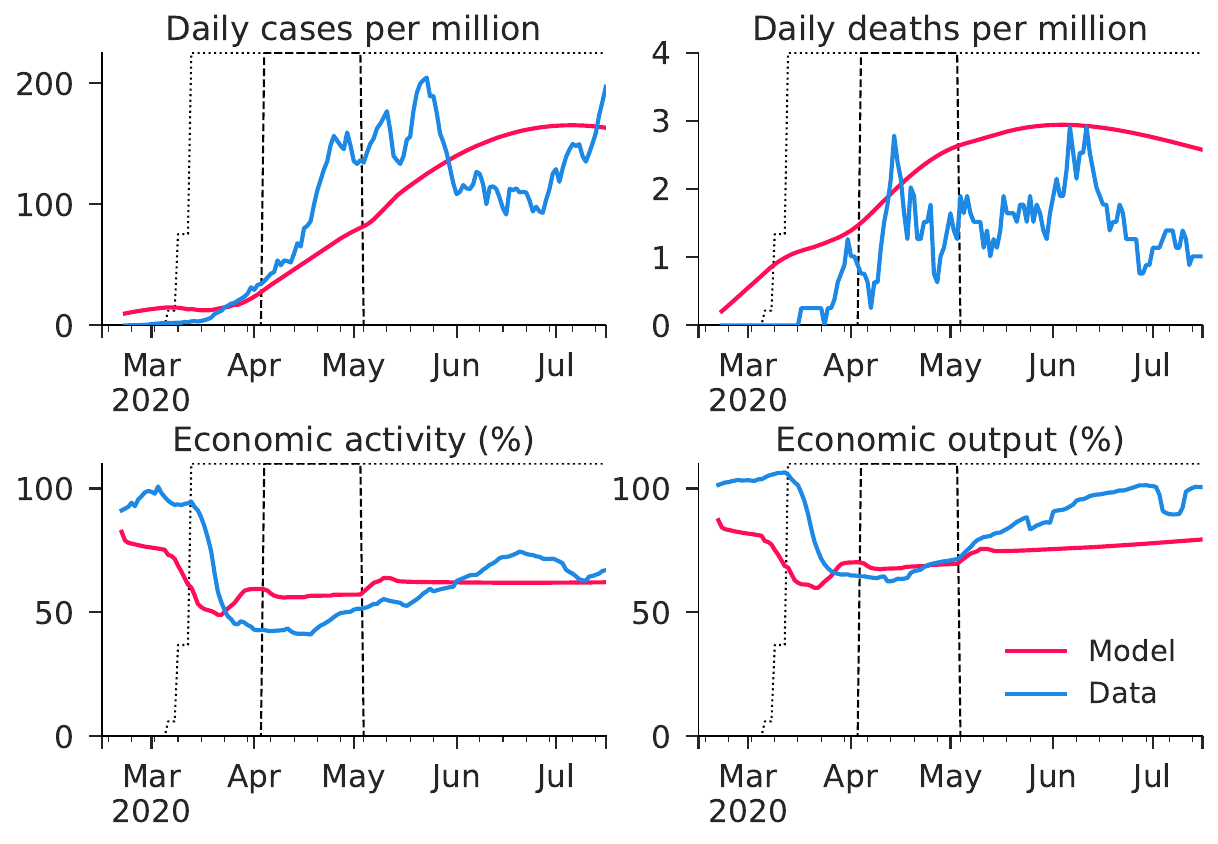}
    \end{subfigure}
    ~~
    \begin{subfigure}{.48\textwidth}
      \centering
      \caption{Nevada}
      \includegraphics[width=1.0\textwidth]{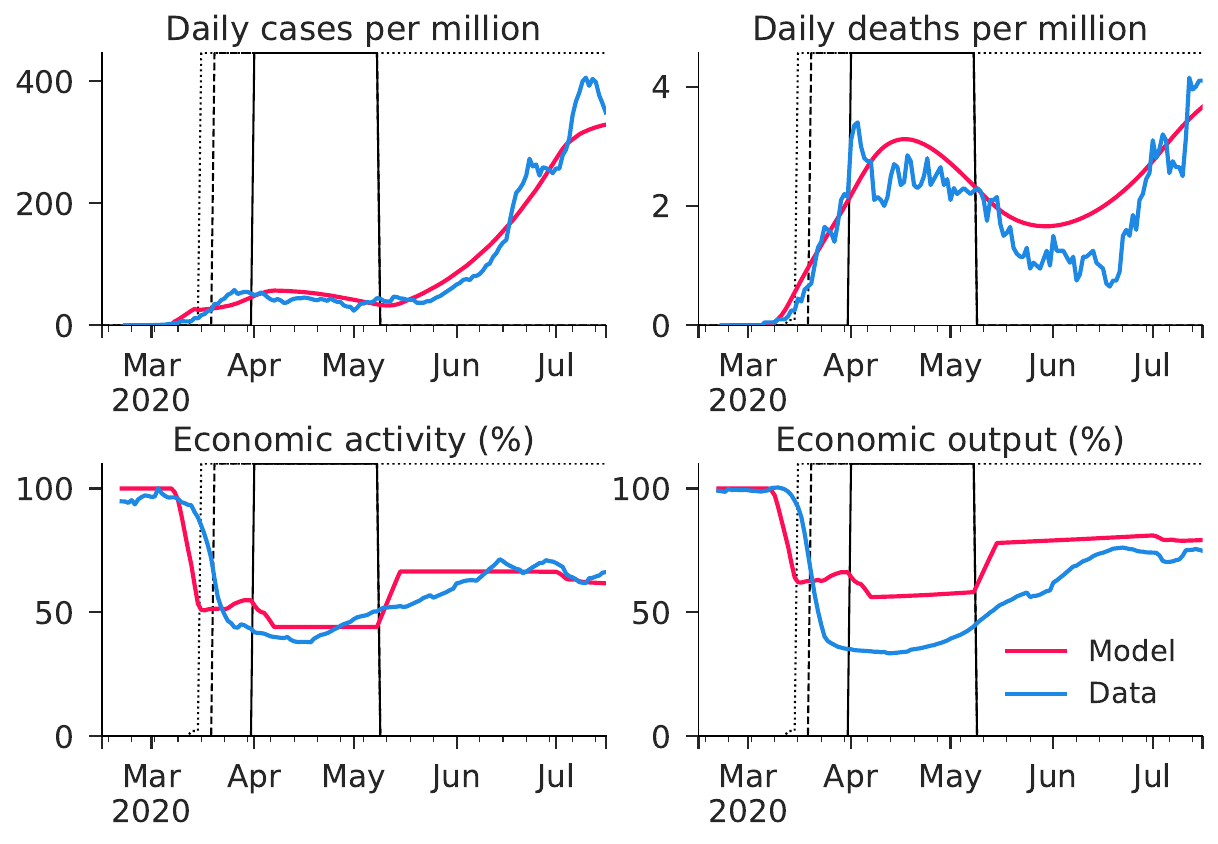}
    \end{subfigure}

    \bigskip

    \begin{subfigure}{.48\textwidth}
      \centering
      \caption{New Mexico}
      \includegraphics[width=1.0\textwidth]{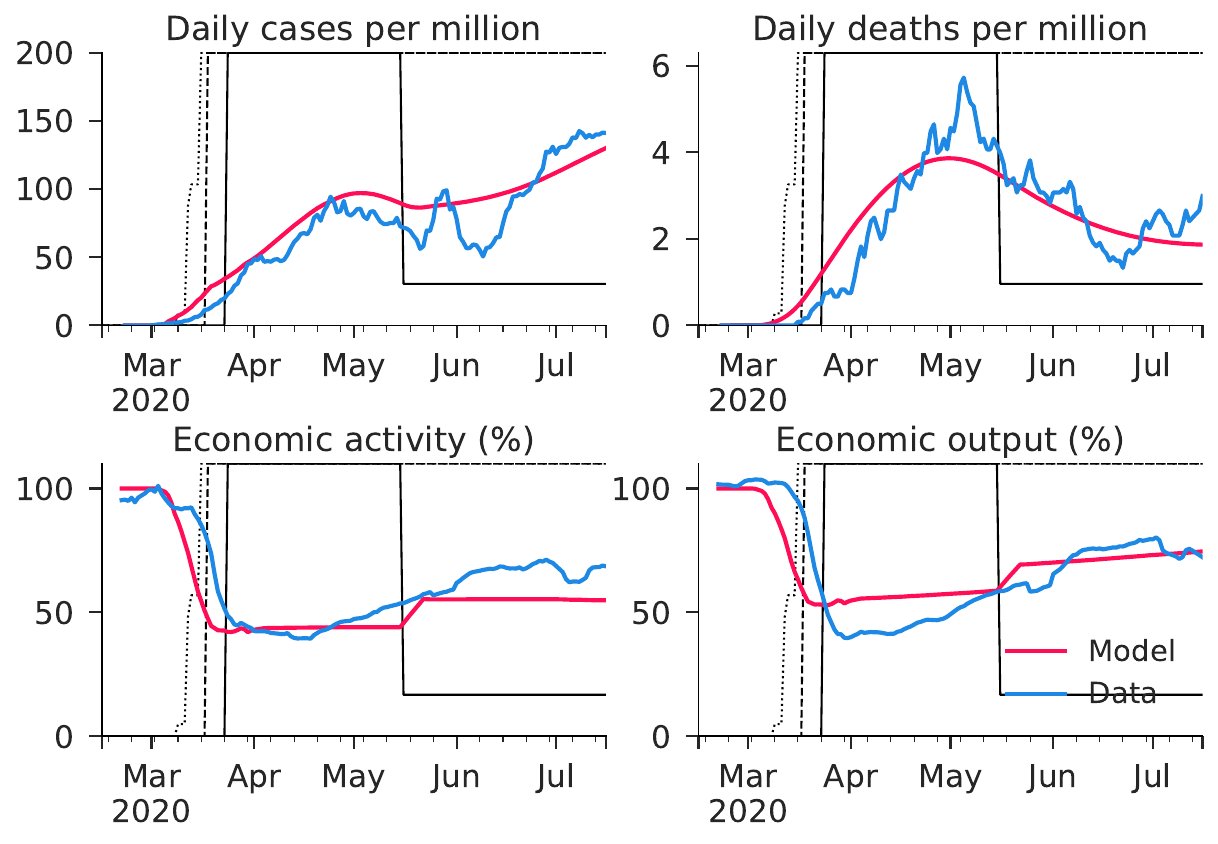}
    \end{subfigure}
    ~~
    \begin{subfigure}{.48\textwidth}
      \centering
      \caption{New York}
      \includegraphics[width=1.0\textwidth]{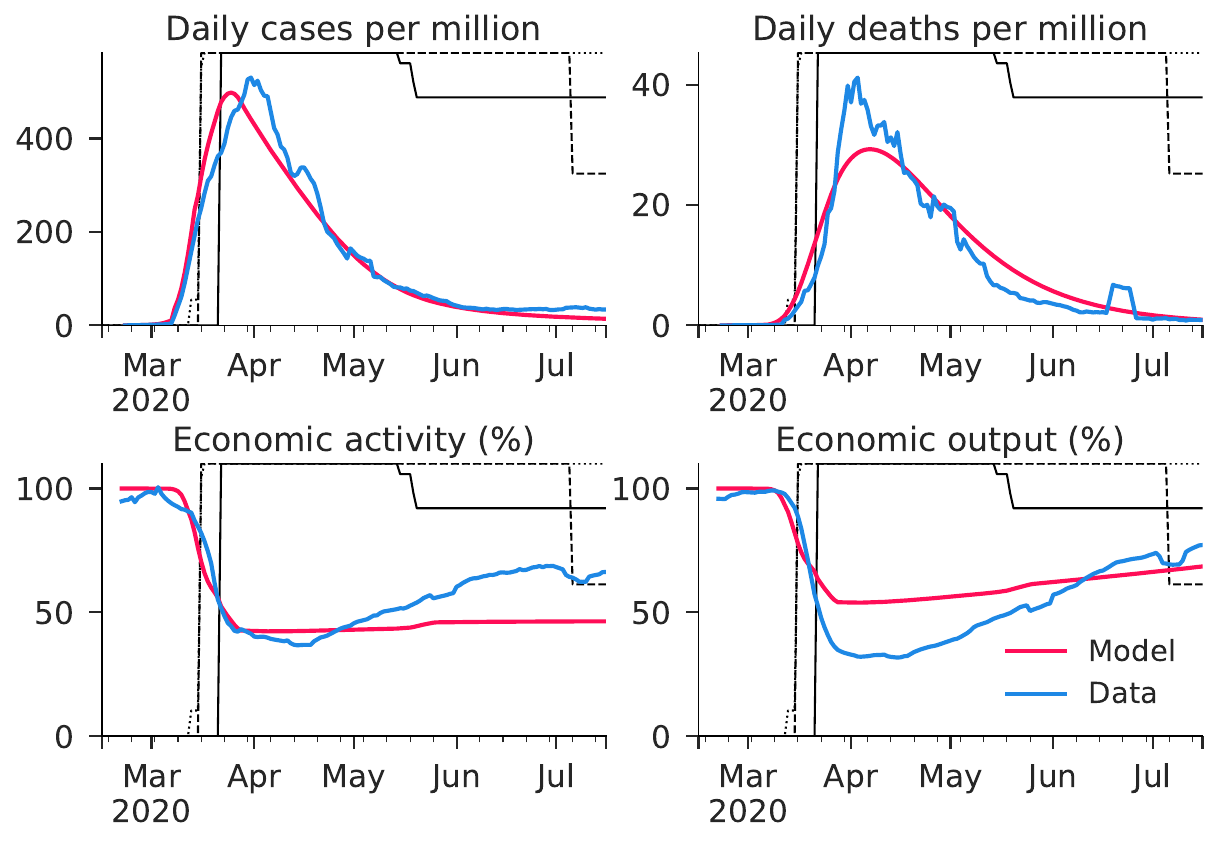}
    \end{subfigure}
    
    \bigskip
    
    \begin{subfigure}{.48\textwidth}
      \centering
      \caption{North Carolina}
      \includegraphics[width=1.0\textwidth]{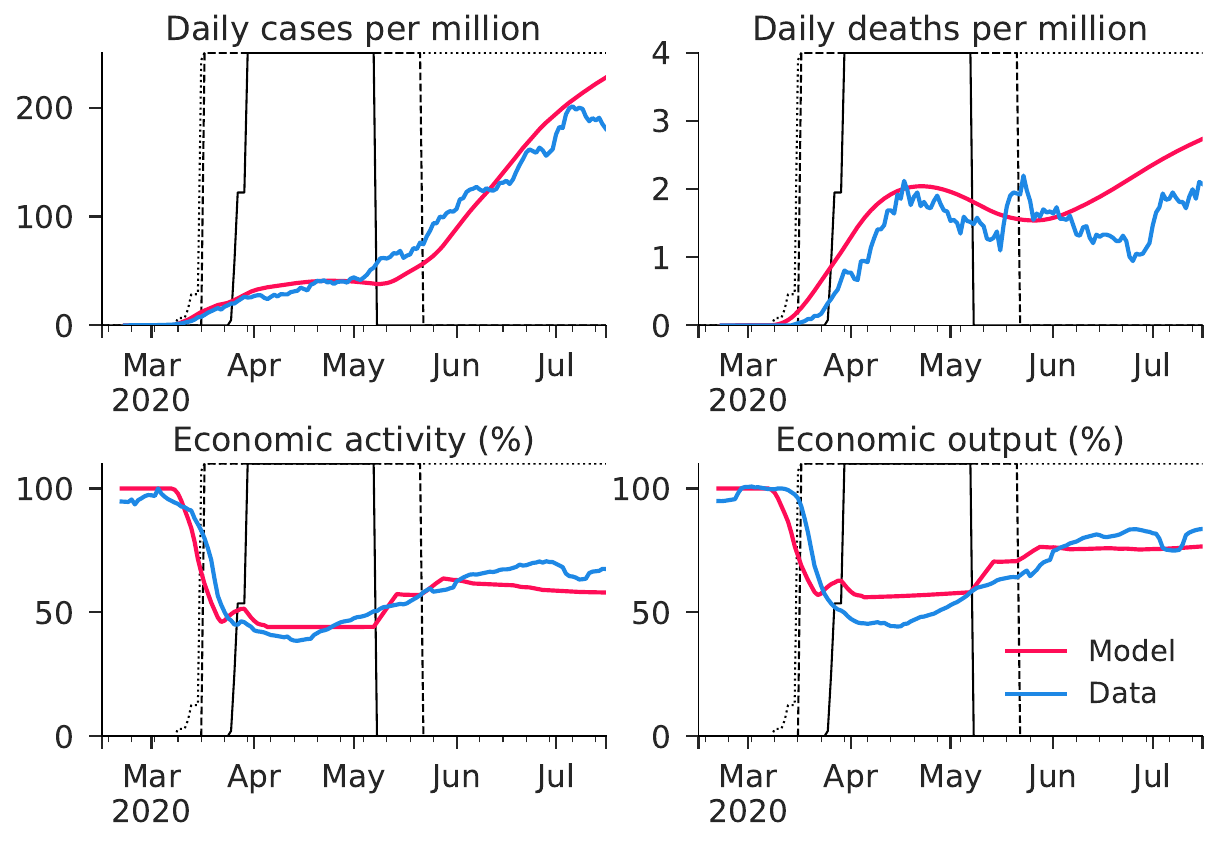}
    \end{subfigure}
    ~~
    \begin{subfigure}{.48\textwidth}
      \centering
      \caption{North Dakota}
      \includegraphics[width=1.0\textwidth]{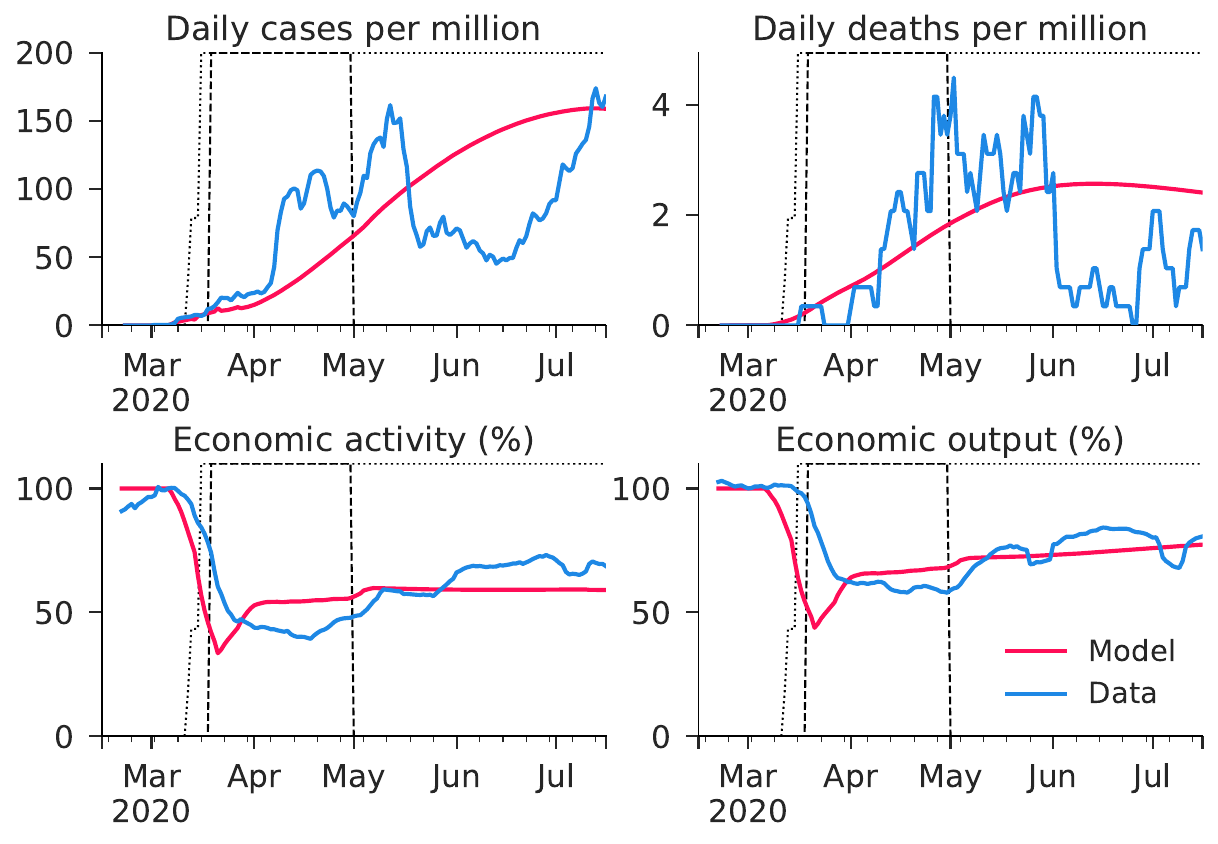}
    \end{subfigure}
\end{figure}

\begin{figure}[t]
    \captionsetup[subfigure]{labelformat=empty}
    \centering
    \begin{subfigure}{.48\textwidth}
      \centering
      \caption{Ohio}
      \includegraphics[width=1.0\textwidth]{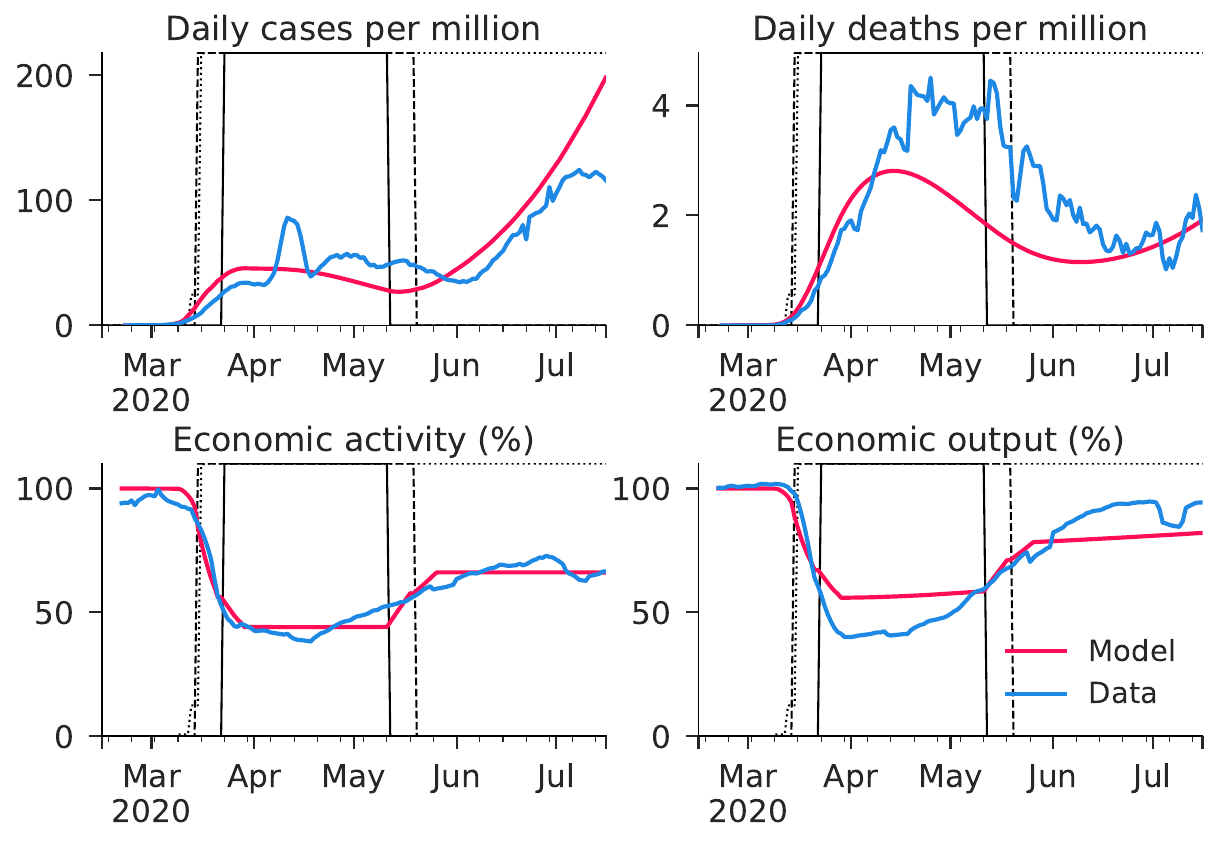}
    \end{subfigure}
    ~~
    \begin{subfigure}{.48\textwidth}
      \centering
      \caption{Oklahoma}
      \includegraphics[width=1.0\textwidth]{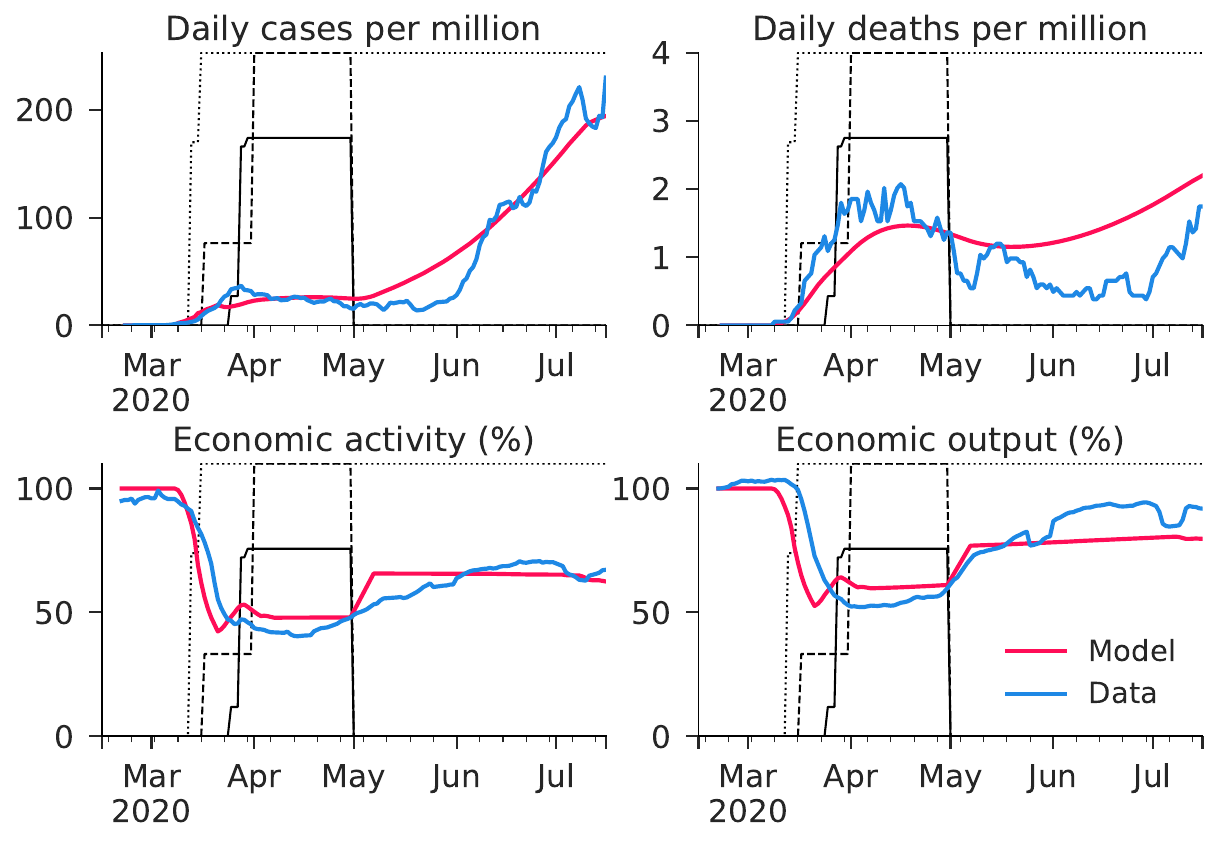}
    \end{subfigure}

    \bigskip

    \begin{subfigure}{.48\textwidth}
      \centering
      \caption{Oregon}
      \includegraphics[width=1.0\textwidth]{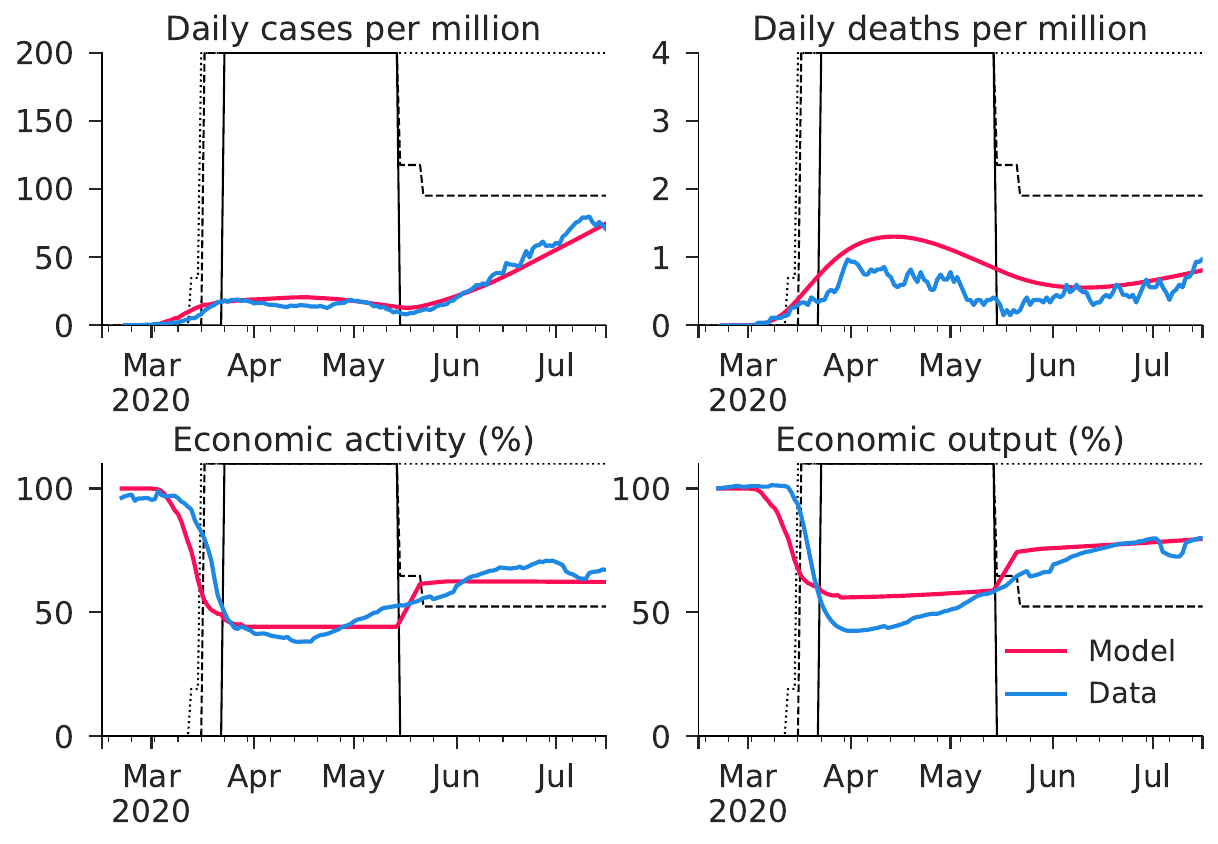}
    \end{subfigure}
    ~~
    \begin{subfigure}{.48\textwidth}
      \centering
      \caption{Pennsylvania}
      \includegraphics[width=1.0\textwidth]{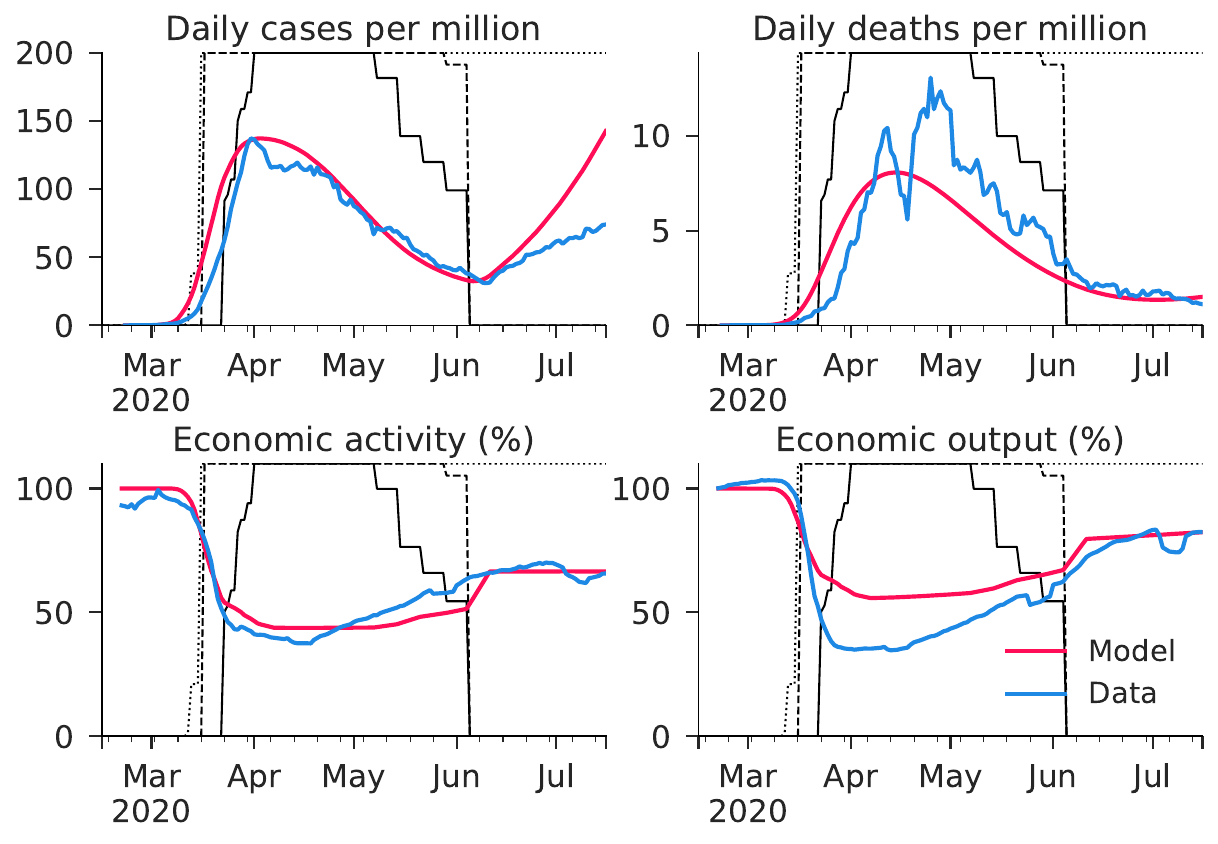}
    \end{subfigure}
    
    \bigskip
    
    \begin{subfigure}{.48\textwidth}
      \centering
      \caption{South Carolina}
      \includegraphics[width=1.0\textwidth]{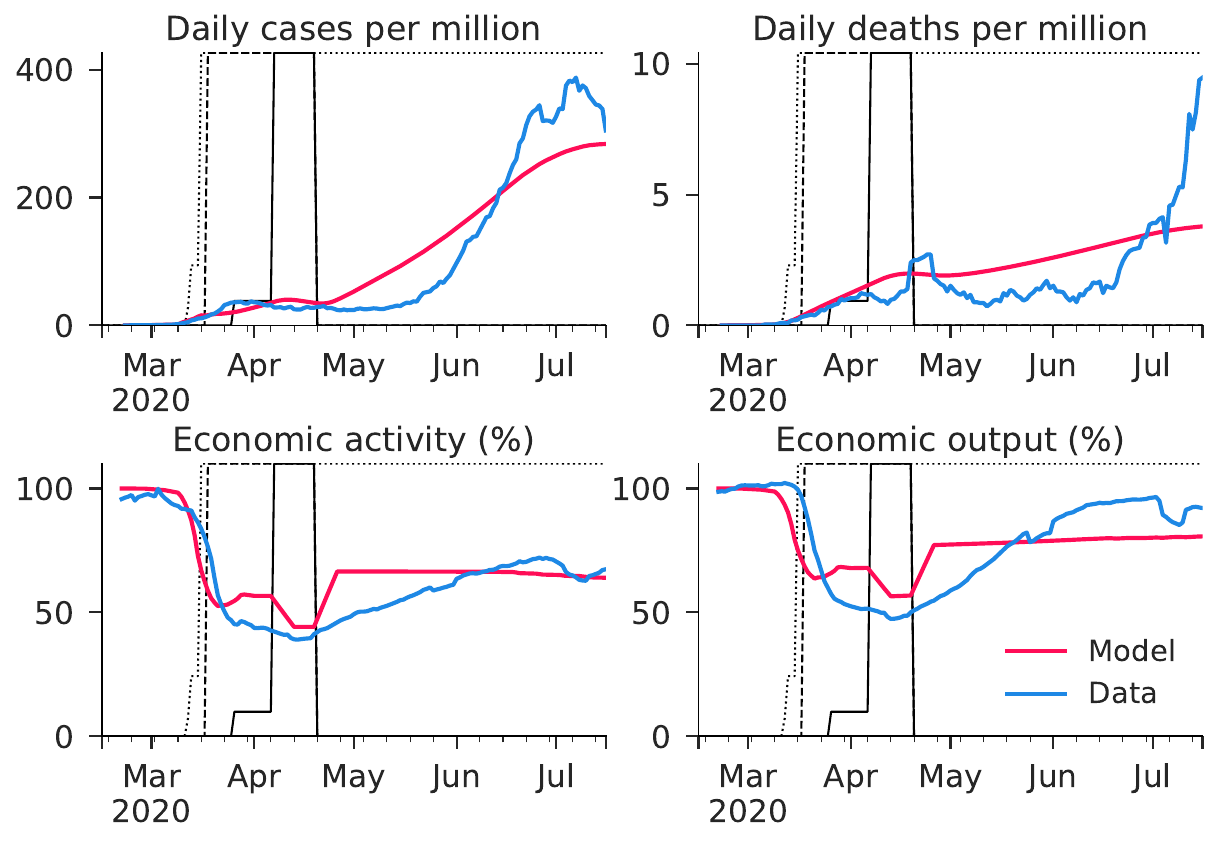}
    \end{subfigure}
    ~~
    \begin{subfigure}{.48\textwidth}
      \centering
      \caption{South Dakota}
      \includegraphics[width=1.0\textwidth]{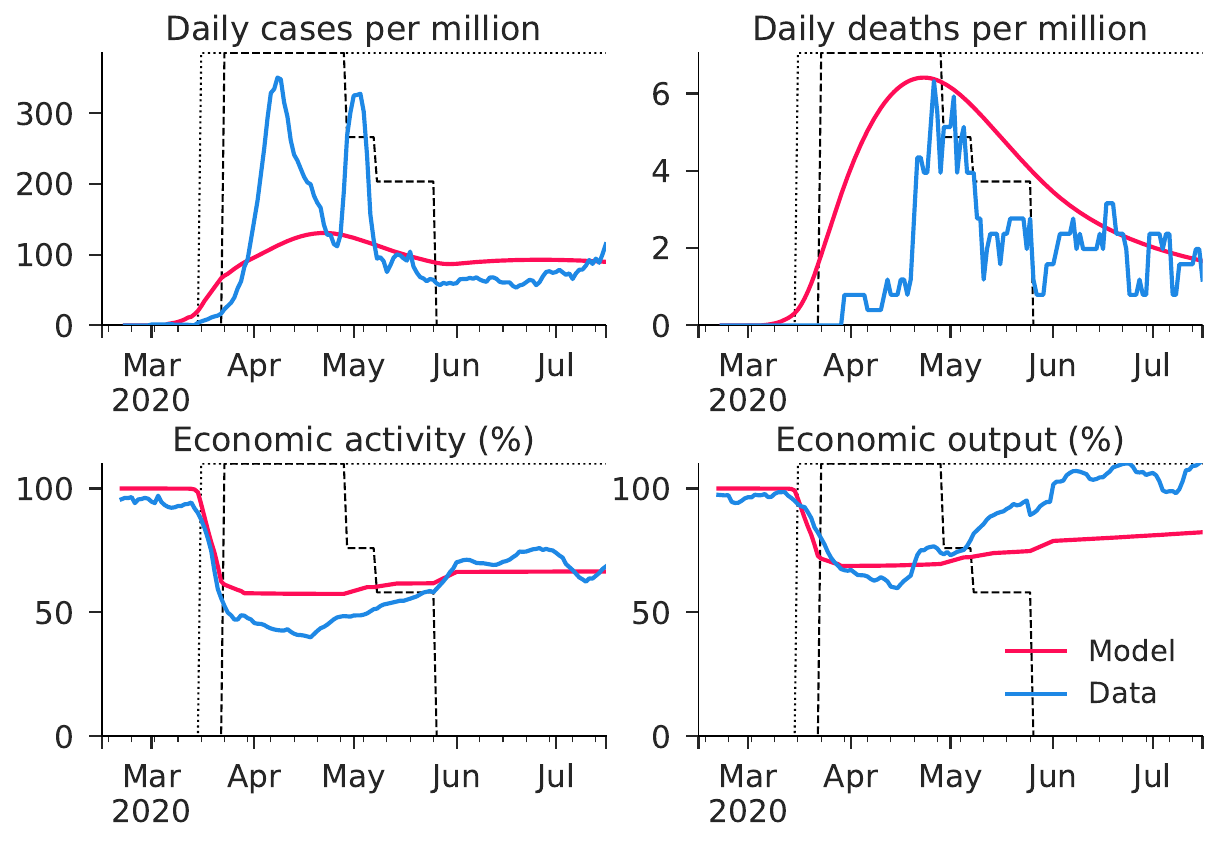}
    \end{subfigure}
\end{figure}

\begin{figure}[t]
    \captionsetup[subfigure]{labelformat=empty}
    \centering
    \begin{subfigure}{.48\textwidth}
      \centering
      \caption{Tennessee}
      \includegraphics[width=1.0\textwidth]{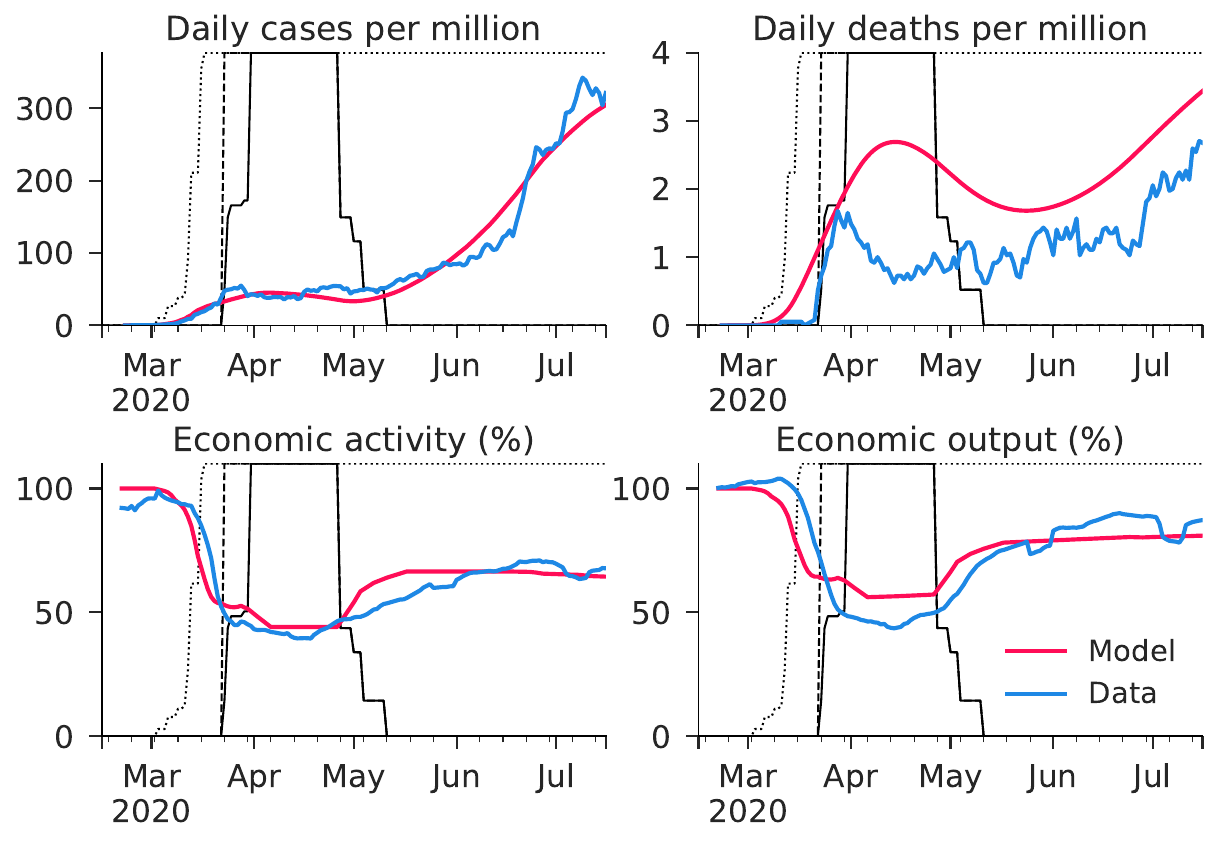}
    \end{subfigure}
    ~~
    \begin{subfigure}{.48\textwidth}
      \centering
      \caption{Texas}
      \includegraphics[width=1.0\textwidth]{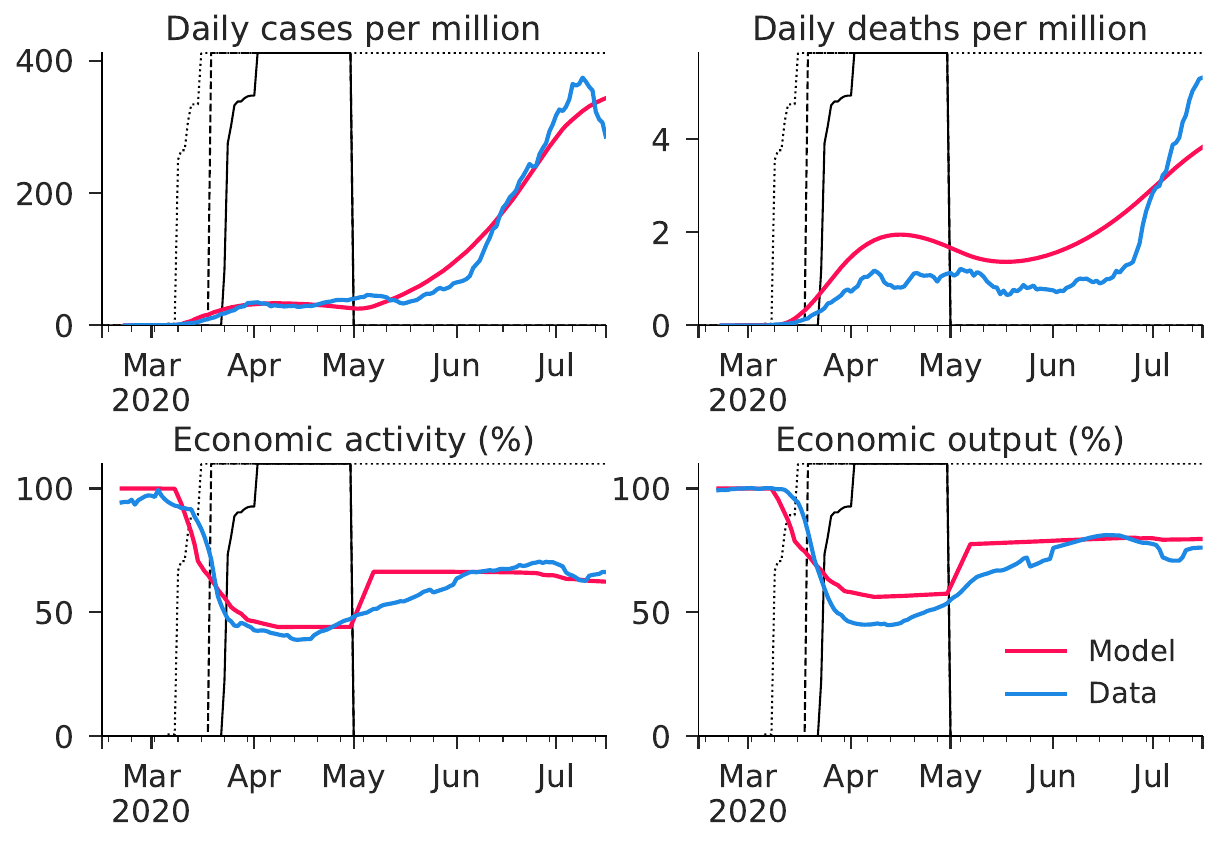}
    \end{subfigure}

    \bigskip

    \begin{subfigure}{.48\textwidth}
      \centering
      \caption{Utah}
      \includegraphics[width=1.0\textwidth]{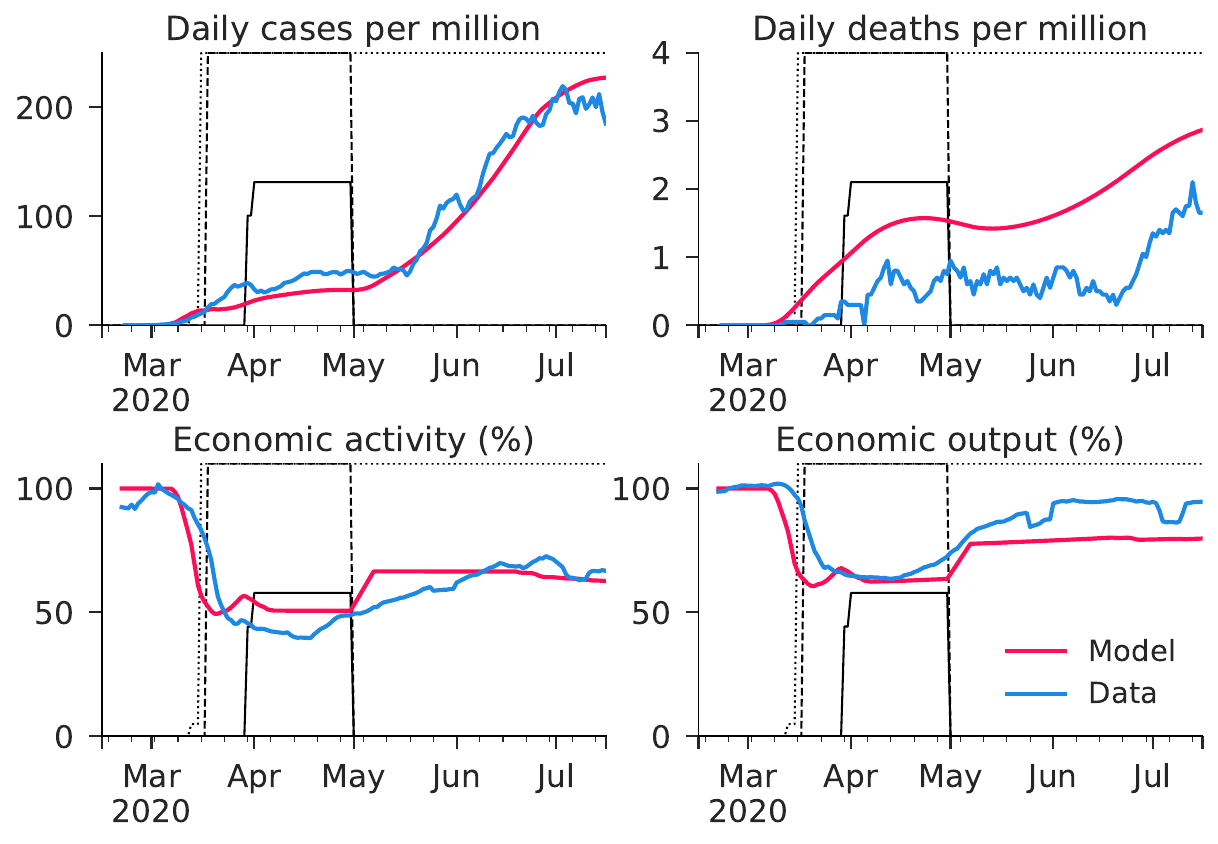}
    \end{subfigure}
    ~~
    \begin{subfigure}{.48\textwidth}
      \centering
      \caption{Virginia}
      \includegraphics[width=1.0\textwidth]{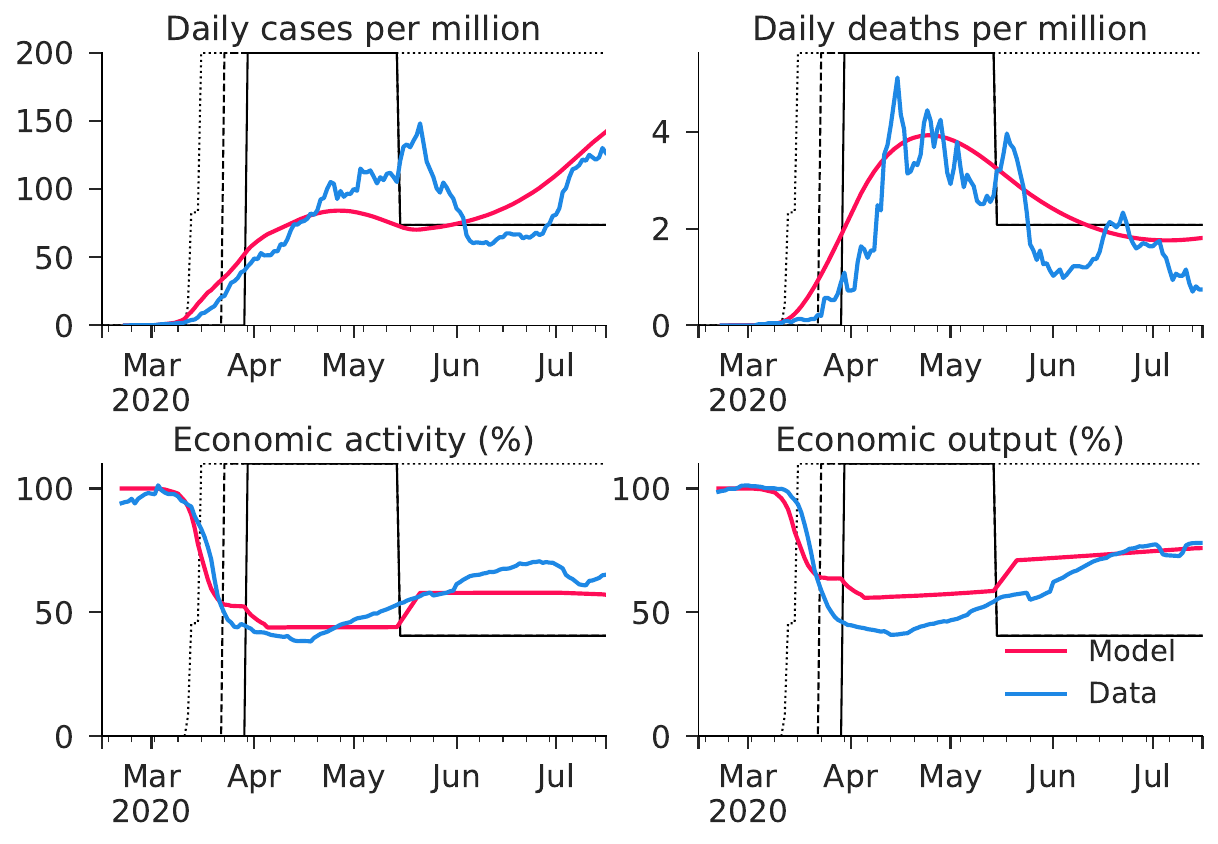}
    \end{subfigure}
    
    \bigskip
    
    \begin{subfigure}{.48\textwidth}
      \centering
      \caption{Washington}
      \includegraphics[width=1.0\textwidth]{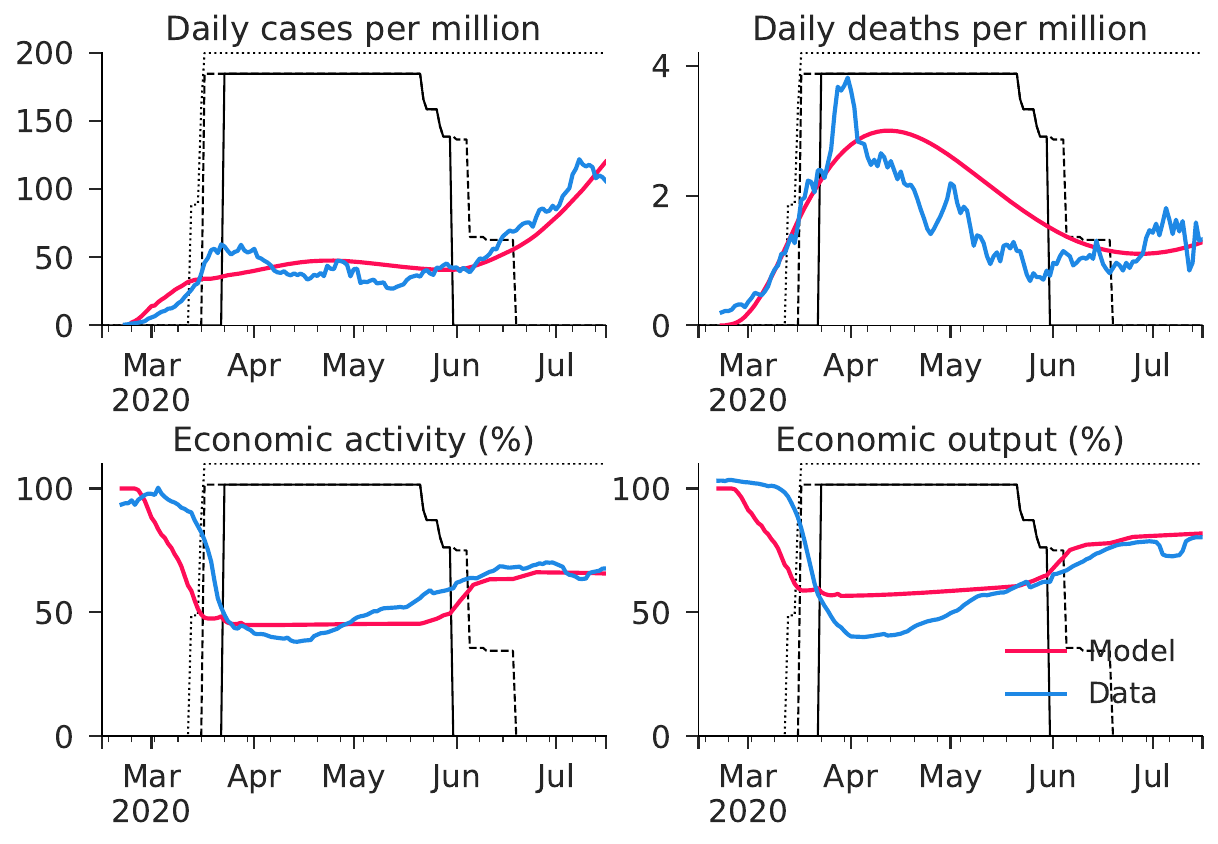}
    \end{subfigure}
    ~~
    \begin{subfigure}{.48\textwidth}
      \centering
      \caption{West Virginia}
      \includegraphics[width=1.0\textwidth]{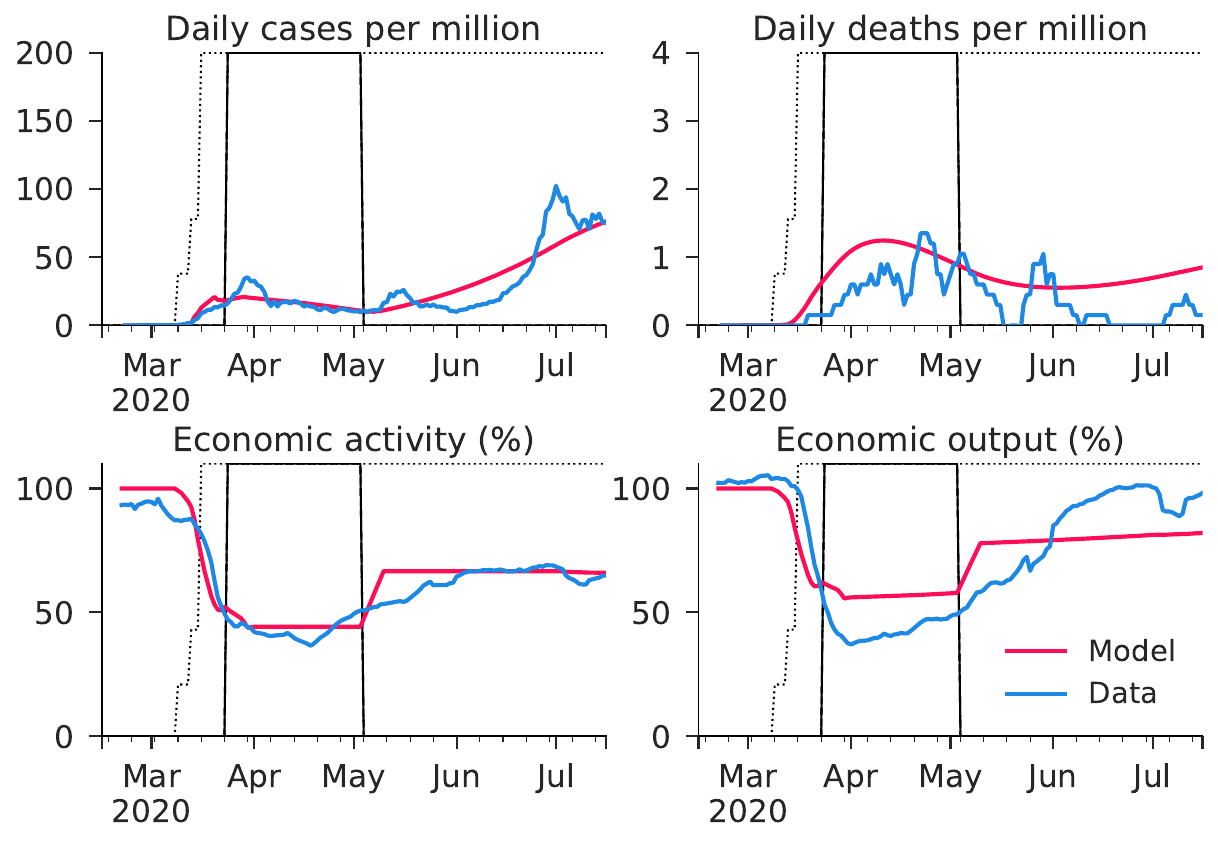}
    \end{subfigure}
\end{figure}

\begin{figure}[t]
    \captionsetup[subfigure]{labelformat=empty}
    \centering
    \begin{subfigure}{.48\textwidth}
      \centering
      \caption{Wisconsin}
      \includegraphics[width=1.0\textwidth]{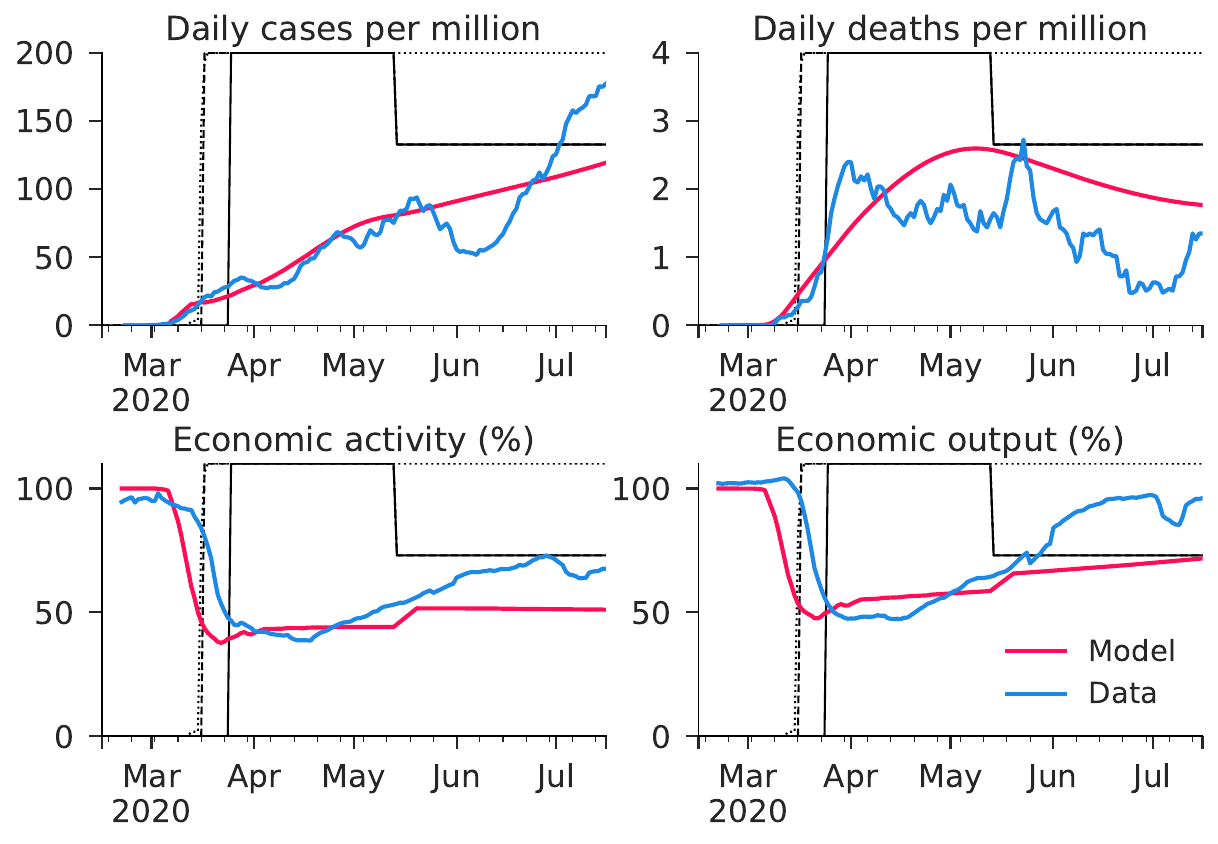}
    \end{subfigure}
    ~~
    \begin{subfigure}{.48\textwidth}
      \centering
      \caption{Wyoming}
      \includegraphics[width=1.0\textwidth]{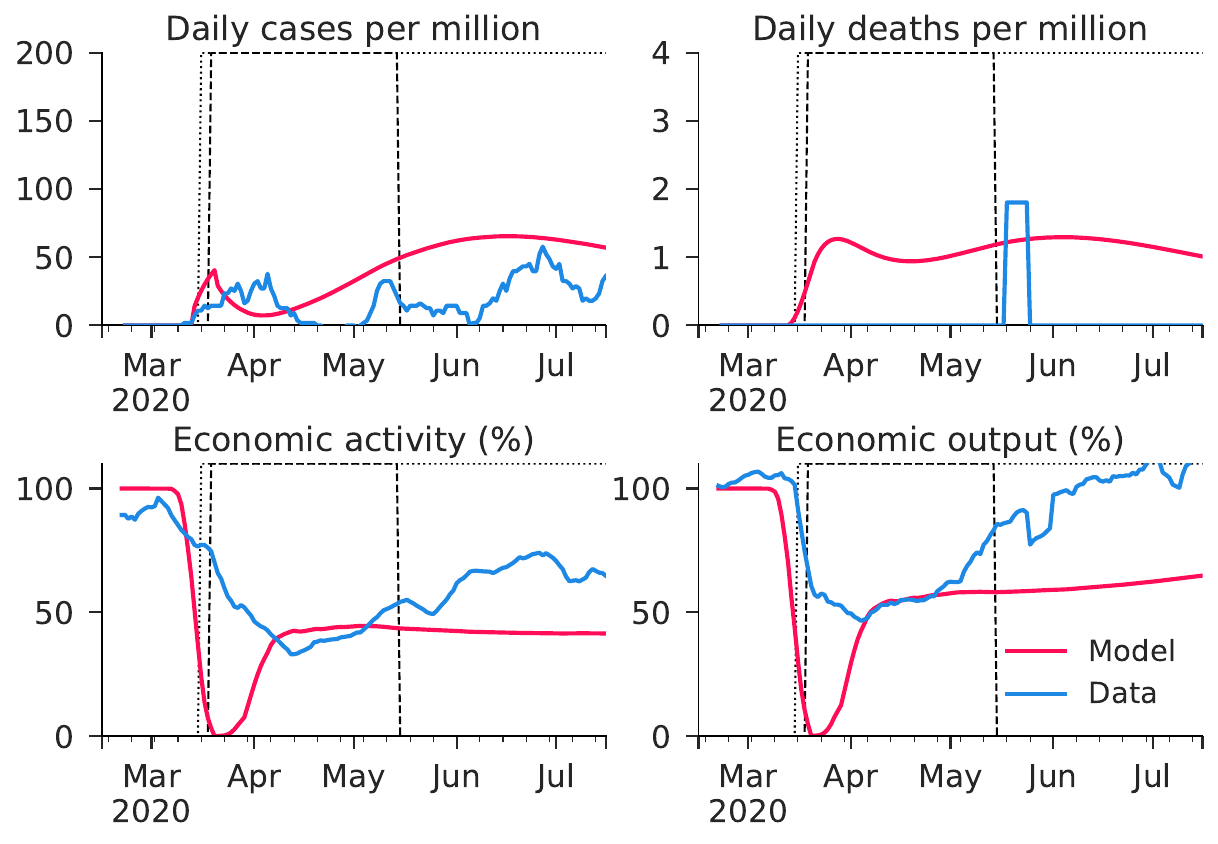}
    \end{subfigure}

\end{figure}

\end{document}

%% file: diagrams/sir_vs_econsir_big_tikzcode.tex
\tikzset{every picture/.style={line width=0.75pt}} 

\begin{tikzpicture}[x=0.75pt,y=0.75pt,yscale=-1,xscale=1]

\draw [line width=0.75]    (122.87,192.57) -- (343,192.57) ;
\draw [shift={(346,192.57)}, rotate = 180] [fill={rgb, 255:red, 0; green, 0; blue, 0 }  ][line width=0.08]  [draw opacity=0] (12.5,-6.01) -- (0,0) -- (12.5,6.01) -- cycle    ;
\draw [line width=0.75]    (408.87,192.57) -- (629,192.57) ;
\draw [shift={(632,192.57)}, rotate = 180] [fill={rgb, 255:red, 0; green, 0; blue, 0 }  ][line width=0.08]  [draw opacity=0] (12.5,-6.01) -- (0,0) -- (12.5,6.01) -- cycle    ;
\draw  [fill={rgb, 255:red, 235; green, 235; blue, 235 }  ,fill opacity=1 ][line width=0.75]  (60.13,61.57) .. controls (60.13,44.21) and (74.21,30.13) .. (91.57,30.13) .. controls (108.93,30.13) and (123,44.21) .. (123,61.57) .. controls (123,78.93) and (108.93,93) .. (91.57,93) .. controls (74.21,93) and (60.13,78.93) .. (60.13,61.57) -- cycle ;
\draw  [fill={rgb, 255:red, 235; green, 235; blue, 235 }  ,fill opacity=1 ] (346.13,61.57) .. controls (346.13,44.21) and (360.21,30.13) .. (377.57,30.13) .. controls (394.93,30.13) and (409,44.21) .. (409,61.57) .. controls (409,78.93) and (394.93,93) .. (377.57,93) .. controls (360.21,93) and (346.13,78.93) .. (346.13,61.57) -- cycle ;
\draw  [fill={rgb, 255:red, 235; green, 235; blue, 235 }  ,fill opacity=1 ] (632.13,61.57) .. controls (632.13,44.21) and (646.21,30.13) .. (663.57,30.13) .. controls (680.93,30.13) and (695,44.21) .. (695,61.57) .. controls (695,78.93) and (680.93,93) .. (663.57,93) .. controls (646.21,93) and (632.13,78.93) .. (632.13,61.57) -- cycle ;
\draw [line width=0.75]    (123,61.57) -- (343.13,61.57) ;
\draw [shift={(346.13,61.57)}, rotate = 180] [fill={rgb, 255:red, 0; green, 0; blue, 0 }  ][line width=0.08]  [draw opacity=0] (12.5,-6.01) -- (0,0) -- (12.5,6.01) -- cycle    ;
\draw [line width=0.75]    (409,61.57) -- (629.13,61.57) ;
\draw [shift={(632.13,61.57)}, rotate = 180] [fill={rgb, 255:red, 0; green, 0; blue, 0 }  ][line width=0.08]  [draw opacity=0] (12.5,-6.01) -- (0,0) -- (12.5,6.01) -- cycle    ;
\draw  [fill={rgb, 255:red, 235; green, 235; blue, 235 }  ,fill opacity=1 ][line width=0.75]  (60.13,192.43) .. controls (60.13,175.07) and (74.21,161) .. (91.57,161) .. controls (108.93,161) and (123,175.07) .. (123,192.43) .. controls (123,209.79) and (108.93,223.87) .. (91.57,223.87) .. controls (74.21,223.87) and (60.13,209.79) .. (60.13,192.43) -- cycle ;
\draw  [fill={rgb, 255:red, 235; green, 235; blue, 235 }  ,fill opacity=1 ] (346.13,192.57) .. controls (346.13,175.21) and (360.21,161.13) .. (377.57,161.13) .. controls (394.93,161.13) and (409,175.21) .. (409,192.57) .. controls (409,209.93) and (394.93,224) .. (377.57,224) .. controls (360.21,224) and (346.13,209.93) .. (346.13,192.57) -- cycle ;
\draw  [fill={rgb, 255:red, 235; green, 235; blue, 235 }  ,fill opacity=1 ] (632.13,192.57) .. controls (632.13,175.21) and (646.21,161.13) .. (663.57,161.13) .. controls (680.93,161.13) and (695,175.21) .. (695,192.57) .. controls (695,209.93) and (680.93,224) .. (663.57,224) .. controls (646.21,224) and (632.13,209.93) .. (632.13,192.57) -- cycle ;

\draw (131,131.4) node [anchor=north west][inner sep=0.75pt]  [font=\large]  {$\overbrace{\textcolor[rgb]{0,0,0}{\beta \thinspace I}\textcolor[rgb]{0,0,0}{(}\textcolor[rgb]{0,0,0}{t}\textcolor[rgb]{0,0,0}{)}\textcolor[rgb]{0,0,0}{\thinspace }\textcolor[rgb]{0.82,0.01,0.11}{[ 1-F( \beta I( t) \psi )]}}^{\textcolor[rgb]{0,0,0}{n( t)}}\textcolor[rgb]{0,0,0}{\thinspace S( t)}$};
\draw (495,161.4) node [anchor=north west][inner sep=0.75pt]  [font=\large]  {$\delta \thinspace I( t)$};
\draw (-1,169.17) node [anchor=north west][inner sep=0.75pt]   [align=left] {\begin{minipage}[lt]{36.198644pt}\setlength\topsep{0pt}
\begin{center}
{\Large \textcolor[rgb]{0.82,0.01,0.11}{Econ}}\\{\Large SIR}
\end{center}

\end{minipage}};
\draw (75.13,49.4) node [anchor=north west][inner sep=0.75pt]  [font=\large]  {${\displaystyle S( t)}$};
\draw (183,6.4) node [anchor=north west][inner sep=0.75pt]  [font=\large]  {$\overbrace{\beta \thinspace I( t)}^{n( t)}\textcolor[rgb]{0,0,0}{\thinspace S}\textcolor[rgb]{0,0,0}{(}\textcolor[rgb]{0,0,0}{t}\textcolor[rgb]{0,0,0}{)}$};
\draw (5,49) node [anchor=north west][inner sep=0.75pt]   [align=left] {\begin{minipage}[lt]{27.200000000000003pt}\setlength\topsep{0pt}
\begin{center}
{\Large SIR}
\end{center}

\end{minipage}};
\draw (364,49.4) node [anchor=north west][inner sep=0.75pt]  [font=\large]  {${\displaystyle I( t)}$};
\draw (646,49.4) node [anchor=north west][inner sep=0.75pt]  [font=\large]  {${\displaystyle R( t)}$};
\draw (75.13,180.27) node [anchor=north west][inner sep=0.75pt]  [font=\large]  {${\displaystyle S( t)}$};
\draw (364,180.4) node [anchor=north west][inner sep=0.75pt]  [font=\large]  {${\displaystyle I( t)}$};
\draw (646,180.4) node [anchor=north west][inner sep=0.75pt]  [font=\large]  {${\displaystyle R( t)}$};
\draw (496,30.4) node [anchor=north west][inner sep=0.75pt]  [font=\large]  {$\delta \thinspace I( t)$};

\end{tikzpicture}